\documentclass[journal,onecolumn,11pt]{IEEEtran}

\usepackage[utf8]{inputenc}
\usepackage[T1]{fontenc}

\usepackage{amsmath,amssymb,amsfonts,amsthm}
\usepackage{mathtools}
\usepackage{bm}

\usepackage{cite}
\usepackage{microtype}
\usepackage{graphicx}
\usepackage{booktabs}
\usepackage{array}
\usepackage{tabularx}
\newcolumntype{L}{>{\raggedright\arraybackslash}X}
\usepackage{longtable}
\usepackage{xcolor}
\usepackage{tikz}
\usetikzlibrary{decorations.pathreplacing}
\usepackage{url}

\usepackage{enumitem}

\newtheorem{theorem}{Theorem}

\newtheorem{proposition}[theorem]{Proposition}
\newtheorem{corollary}{Corollary}

\newtheorem{lemma}[theorem]{Lemma}

\newcommand{\R}{\mathbb{R}}

\newcommand{\NDF}{\mathrm{NDF}}

\newcommand{\LAE}{L_{\mathrm{AE}}}

\newcommand{\eps}{\varepsilon}

\newcommand{\NDFe}{\mathrm{NDF}_{\varepsilon}}

\usepackage[unicode=true,
pdfencoding=auto,
colorlinks=true,
linkcolor=black,
citecolor=black,
urlcolor=blue]{hyperref}

\newcommand{\supppart}[2]{%
	\par\addvspace{1.5\baselineskip}%
	\noindent\hrule height 0.4pt\relax
	\vspace{0.6\baselineskip}%
	{\centering\normalfont\large\bfseries Appendix Part~#1: #2\par}%
	\vspace{0.4\baselineskip}%
	\noindent\hrule height 0.4pt\relax
	\vspace{0.8\baselineskip}}

\begin{document}

\title{Parametric Electromagnetic Information:\\
		Field-Manifold Geometry and the Stability of Learned Field Representations}

\author{M.~D.~Migliore, \IEEEmembership{Senior Member, IEEE}
	\thanks{M. D. Migliore is with the DIEI (Department of Electrical
		and Information Engineering ``Maurizio Scarano''), University of
		Cassino and Southern Lazio, Via Di Biasio 43, I-03043 Cassino,
		Italy (e-mail: mdmiglio@unicas.it).}
	\thanks{He is also with European University of Technology (EUT+),
		CNIT (Consorzio Interuniversitario per le Telecomunicazioni) and
		Eledia Center.}
	\thanks{The manuscript was checked with AI-based
		language-editing tools to identify errors of English grammar and
		style. The author carefully reviewed the whole text and takes full
		responsibility for its content.}
}

\markboth{Migliore: Parametric Electromagnetic Information}
{Migliore: Parametric Electromagnetic Information}

\maketitle

\begin{abstract}
	This paper extends Electromagnetic Information Theory (EIT) from sources
	with arbitrary currents of bounded norm to sources and scatterers driven
	by a few parametric controls, and applies it to learned field
	representations. Analyzing the set of accessible fields with the
	Kolmogorov $\varepsilon$-entropy and $\varepsilon$-capacity and with
	the pullback metric, built from the field derivatives, we distinguish
	the intrinsic dimension $d$ of this set, i.e., the number of
	independent field variations, from its metric volume, which sets how
	many fields can be distinguished at a prescribed resolution and grows
	with electrical size. Under white Gaussian noise, the pullback metric
	is proportional to the Fisher information matrix of the controls, so
	that the same geometry describes both communication, through the
	number of distinguishable field states, and sensing, through the
	accuracy with which the controls can be estimated. Known covering and Lipschitz-width
	bounds then give a lower bound on the sensitivity of any decoder that
	reconstructs every field at the prescribed accuracy, i.e., on how much
	it amplifies errors on the latent coordinates; this bound grows with
	electrical size. Applied to physics-constrained autoencoders on arrays
	and a full-wave scatterer, the theory shows that key properties of a
	neural representation of the fields are set by the electromagnetic
	problem and can be predicted before training: the minimum number of
	latent coordinates, fixed by the topology of the set of fields, and
	the growth of the decoder sensitivity with electrical size. The
	trained networks are consistent with these predictions.
\end{abstract}


\section{Introduction}
\label{sec:introduction}

Over the past decade, \textit{Electromagnetic Information Theory}
(EIT) has emerged as a framework bridging Maxwell's equations and
Shannon's information theory. Its classical descriptor is the linear
Number of Degrees of Freedom ($\mathrm{NDF}$), which bounds channel
capacity and the resolution of inverse source and scattering problems when the source currents are arbitrary,
apart from a bound on their norm, so that every significant field mode
can be excited
\cite{miller1998spatial,poon2005degrees,bucci1989degrees,
	bucci1997electromagnetic,migliore2006role,janaswamy2011degrees,
	franceschetti2017wave,gustafsson2025degrees,gradoni2018stochastic,
	pizzo2022spatial,wan2023mutual,zhu2024electromagnetic}.

\begin{figure}[htbp]
	\centering
	\includegraphics[width=0.6\linewidth]{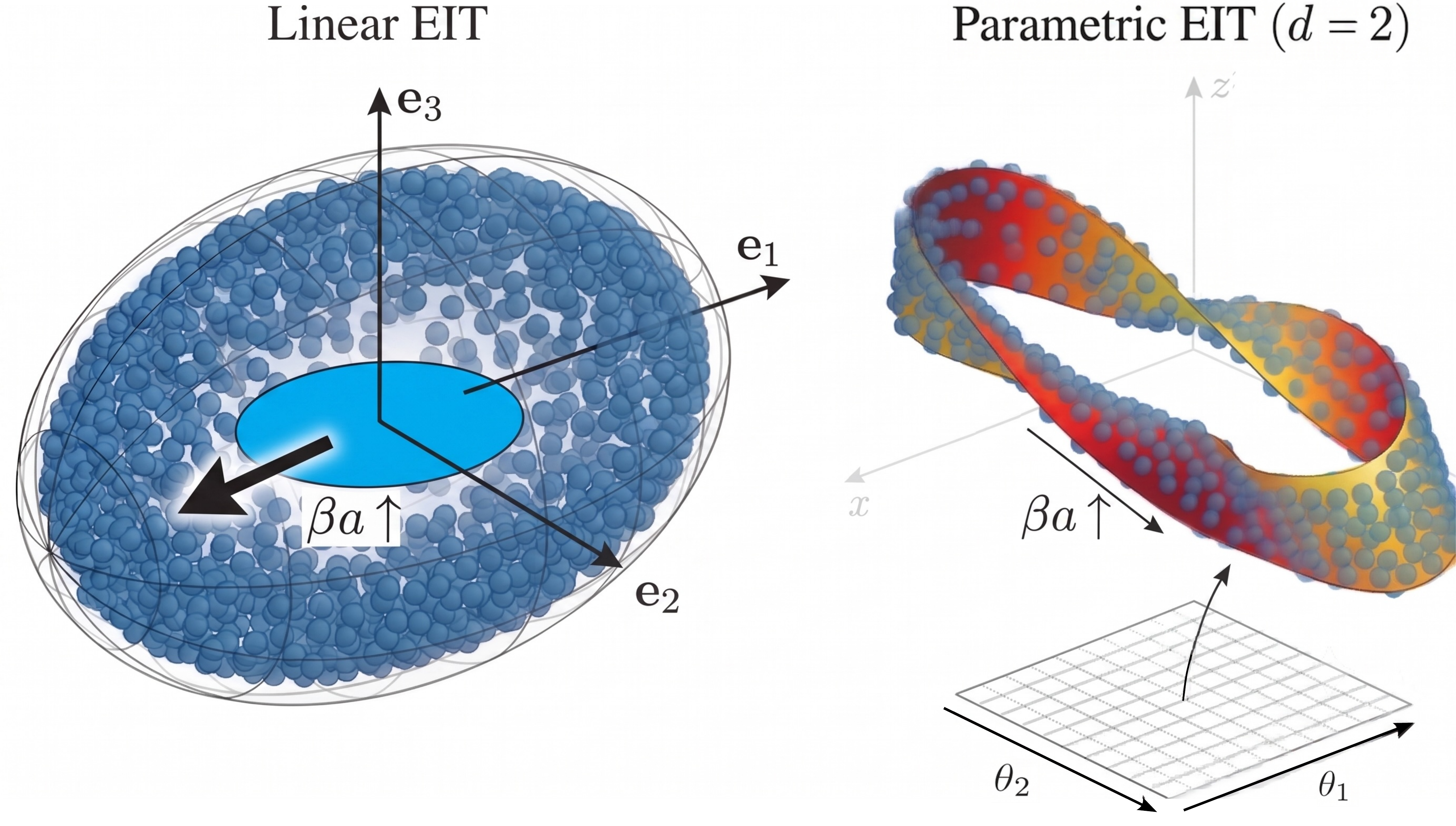}
\caption{Linear and parametric field sets. Left: in linear EIT, the
	fields radiated by arbitrary currents of bounded norm form an
	ellipsoid, whose semi-axes lie along the singular directions
	$\mathbf e_n$ of the radiation operator; as the electrical size
	$\beta a$ grows, more semi-axes exceed the resolution, i.e., the set
	gains new independent directions, whose number is the
	$\mathrm{NDF}$. Right: in parametric EIT, two controls
	$(\theta_1,\theta_2)$ map onto a curved surface of fields ($d=2$); as
	$\beta a$ grows, the surface is stretched, but its dimension does not
	change, so that it holds more balls of radius $\varepsilon$, i.e., more
	distinguishable fields. The balls represent covering or packing
	elements according to whether $\varepsilon$-entropy or
	$\varepsilon$-capacity is considered.}
	\label{fig:manifoldcompr}
\end{figure}

This assumption is often inappropriate for modern
high-frequency systems, where phase-only or otherwise restricted
controls are preferred to fully reconfigurable apertures. Such
constraints define a low-dimensional \textit{parametric control
	space}, for which the linear $\mathrm{NDF}$ may overestimate the
accessible field complexity, hardware requirements, and surrogate-model
dimension.

In \cite{migliore2026unification}, EIT was formulated as a
	geometric problem in phase space, using symplectic geometry: the NDF,
	i.e., the number of significant semi-axes of the ellipsoid of fields
	generated by arbitrary currents, is a phase-space volume. This paper
	continues that geometric approach for fields generated by specified
	physical controls, but in a different space: the field family
	$\mathcal E$, i.e., the set of all fields that the controls can
	produce, is a curved manifold in field space, studied with Riemannian geometry,
	as in information geometry and array-manifold analysis, together with
	the Kolmogorov $\varepsilon$-entropy and $\varepsilon$-capacity
	approach of \cite{migliore2008electromagnetics}. Maxwell's
equations, the excitation constraint, and the observation norm determine
the map from controls to fields; its image can have few
independent directions of variation and still contain many
distinguishable field configurations (Fig.~\ref{fig:manifoldcompr}). We describe these
features by the intrinsic dimension $d$ and by the metric volume, whose
electrical-size growth is characterized by the exponent $\nu$;
	both follow from the pullback metric $\mathbf g$, the Gram matrix of
the field derivatives with respect to the controls
\cite{amari2016information}. The pullback metric is thus a common thread of this paper: it
	governs both communication and sensing applications, within limits
	set by global features that it does not see, such as folds of the
	family.

Our contribution is to derive from this geometry two limits that hold
for any compact representation of the fields, i.e., any scheme that
encodes each field into a few latent coordinates and reconstructs it
from them: the minimum number of latent coordinates, and the minimum
sensitivity of the reconstruction at a prescribed accuracy. We use
established results in non-linear approximation, information geometry,
manifold learning, and array processing
\cite{tikhomirov1993varepsilon,frankel2004geometry,amari2016information,niyogi2008finding,cohen2022optimal,petrova2023lipschitz,rim2026performance,manikas2001manifold,  castellanos2025em}.
Since the number of distinguishable fields grows with electrical size,
so does the second limit: at a fixed number of latent coordinates and a
fixed accuracy, the reconstruction must amplify errors on the latent
coordinates more and more, provided that distant control settings do
not produce nearly identical fields. The bound
follows from known results \cite{cohen2022optimal,petrova2023lipschitz};
we call its growth with electrical size \emph{Lipschitz migration},
since the amplification is measured by a Lipschitz constant. This bound must be satisfied by every
representation, whatever its implementation, and provides a benchmark,
derived from the field family before training, for a trained
autoencoder \cite{hinton2006reducing, guo2022vae}, i.e., a neural network that
learns both the encoding and the reconstruction.

The paper is organized in three parts. The first develops the theory,
independent of any specific representation:
Section~\ref{sec:entropy_capacity} describes the family of accessible
fields at a prescribed resolution, Section~\ref{sec:fisher_bridge}
gives its statistical interpretation, and Section~\ref{sec:limits}
derives the limits of any continuous representation that reconstructs
the whole family at a prescribed accuracy. The second,
Section~\ref{sec:stability_validation}, applies the theory to a
physics-constrained autoencoder, in which the physics is built into the
architecture, since the last layer of the decoder is fixed by the
radiation operator, rather than added as a penalty to the training loss
as in physics-informed neural networks \cite{raissi2019physics}. The
third consists of three case studies, on the latent dimension of a
load-controlled array, the decoder sensitivity of a steered array, and
the scattering regimes and folding of dielectric cylinders
(Secs.~\ref{sec:array_case_study}--\ref{sec:one_dimensional_manifolds}),
each comparing the predictions of the theory with computed fields and
trained autoencoders.

The symbols and acronyms are listed at the end of the paper, in Appendix~\ref{supp:symbols}.

\section{Parametric Field Families and Their Information Content}
\label{sec:entropy_capacity}

The central quantity of this paper is the volume of the parametric
field set $\mathcal E$, which can be accessed in two complementary ways.
The Kolmogorov $\varepsilon$-entropy and the associated
$\varepsilon$-capacity \cite{tikhomirov1993varepsilon,
	migliore2008electromagnetics} are global: they count the fields of
$\mathcal E$ that are distinguishable at resolution $\varepsilon$, and,
at fine resolution, the growth of this count with $1/\varepsilon$ is set
by the dimension and the volume of $\mathcal E$. The pullback metric is
local: it gives the volume swept by small changes of the controls, and
its integral over the control space gives the volume of $\mathcal E$.
The two descriptions agree when distinct control settings produce
distinct and well-separated fields. Since the metric compares only
nearby settings, it cannot reveal distant settings that produce
identical or nearly identical fields, as with grating lobes, nor a
manifold that folds back or closes on itself. These global geometric
and topological features reduce the number of distinguishable fields
and constrain how the family can be represented; they are treated at
the end of this section.

\subsection{Metric Entropy and $\varepsilon$-Capacity of Electromagnetic Field Varieties}
\label{subsec:metric_entropy_definition}

Let $\mathcal H_E$ be a complex Hilbert space, equipped with the usual
$L^2$ norm, representing the field-observation space.

We define the electrical size of the source as $s\coloneqq\beta a$, where
$\beta$ is the free-space wavenumber and $a$ is a characteristic linear
dimension of the source, such as the radius of a scatterer or the
	length $L_{\mathrm{ap}}$ of an array (then $s=\beta L_{\mathrm{ap}}$).
	We write $\beta a$ when the connection with standard EIT is relevant,
	and $s$ in the general results.

For a compact set of fields $\mathcal E\subset\mathcal H_E$ and any
resolution $\varepsilon>0$, an absolute worst-case distance in
	the norm of $\mathcal H_E$ in the units of the field, we define the
following quantities \cite{tikhomirov1993varepsilon}:
\begin{itemize}
	\item The \emph{covering number} $N_\varepsilon(\mathcal E)$ is the
	minimum number of open balls of radius $\varepsilon$ in $\mathcal H_E$
	required to cover $\mathcal E$:
	\begin{equation}
	\label{eq:covering-def}
		N_\varepsilon(\mathcal E)
		\coloneqq
		\min\left\{
		k\in\mathbb N:
		\mathcal E\subset
		\bigcup_{i=1}^{k}
		B_{\mathcal H_E}(x_i,\varepsilon),
		\;
		x_i\in\mathcal H_E
		\right\},
	\end{equation}
	where $B_{\mathcal H_E}(x,\varepsilon)$ denotes the open ball of
	radius $\varepsilon$ centered at $x$ in $\mathcal H_E$;
	
	\item The \emph{packing number} $M_\varepsilon(\mathcal E)$ is the
	maximum cardinality of a subset of points in $\mathcal E$ that are
	mutually $\varepsilon$-separated, satisfying
	$\|E_i-E_j\|_{\mathcal H_E}>\varepsilon$ for any distinct pair
	$E_i,E_j\in\mathcal E$.
\end{itemize}

From these cardinalities, the Kolmogorov metric entropy is defined as
\[
H_\varepsilon(\mathcal E)
\coloneqq
\log_2N_\varepsilon(\mathcal E),
\]
whereas, at the same resolution scale, the $\varepsilon$-capacity is
defined as
\[
C_\varepsilon(\mathcal E)
\coloneqq
\log_2M_{2\varepsilon}(\mathcal E),
\]
both expressed in bits throughout this paper. Physically, metric
entropy measures \emph{representation complexity}---the minimum number
of bits describing an unknown element of $\mathcal E$ within
worst-case distortion $\varepsilon$---while $\varepsilon$-capacity measures
\emph{communication volume}---the number of bits that index the
	largest set of discrete field states whose radius-$\varepsilon$
uncertainty regions do not overlap under an isotropic noise floor
$\varepsilon$.

The covering and packing numbers bound each other when the
	resolution $\eps$ is changed by a factor of two
\cite{tikhomirov1993varepsilon}:
\begin{equation}
	M_{2\eps}(\mathcal E)
	\leq
	N_\eps(\mathcal E)
	\leq
	M_\eps(\mathcal E),
	\label{eq:sandwich}
\end{equation}
so that, as $\eps\to0$, the $\eps$-capacity
$\log_2M_{2\eps}(\mathcal E)$ follows the same scaling law as the
$\eps$-entropy $\log_2N_\eps(\mathcal E)$. For the finite-dimensional
field sets considered below, doubling $\eps$ changes the entropy only by
a bounded offset, negligible relative to $\log_2(1/\eps)$; hence the results
	obtained for the entropy also hold for the capacity, i.e., for the
	communication problem
\cite{migliore2026unification}. These definitions apply to any
	compact set of fields; below, they are applied to the linear family
	$\mathcal E_{\mathrm{lin}}$ of classical EIT and to the parametric
	family $\mathcal E$ studied in this paper.

As a preliminary step, we recall the classical approach to EIT. The
source currents are unconstrained apart from a bound on their norm,
set to one; in the following, \emph{free}, \emph{arbitrary}, or
\emph{unconstrained} currents always mean currents of bounded norm.
The field set $\mathcal E_{\mathrm{lin}}\subset\mathcal H_E$ is thus
the image of the unit ball under the radiation operator, which maps
source currents to the fields they produce
\cite{migliore2026unification}. Let
\[
\sigma_1\geq\sigma_2\geq\cdots\geq0
\]
denote the singular values of this operator, i.e., the field
	norms produced by the unit-norm source currents of its singular basis;
	$\mathcal E_{\mathrm{lin}}$ is thus an ellipsoid with semi-axes
	$\sigma_n$. At resolution $\varepsilon$, let
\begin{equation}
	\NDFe
		\coloneqq
		\#\left\{
		n:
		\sigma_n>\varepsilon
		\right\},
		\label{eq:ndf-def}
\end{equation}
where $\#$ denotes cardinality. Thus, $\NDFe$ is the number of singular
modes that a source of unit norm excites above the resolution $\eps$
	(for a source of different norm, $\eps$ is divided by that norm);
	equivalently, $\sigma_n/\sigma_1>\eps/\sigma_1$
	\cite{migliore2026unification}.

The effective dimension of $\mathcal E_{\mathrm{lin}}$ is governed by
$\NDFe$, which, for one- and two-dimensional sources, respectively,
grows as
\[
\NDFe\propto(\beta a)^q,
\qquad
q=1\ \text{or}\ 2.
\]
The corresponding linear $\varepsilon$-entropy is
\cite{migliore2008electromagnetics}:
\begin{equation}
	\begin{aligned}
		H_\varepsilon(\mathcal E_{\mathrm{lin}})
		&=
		2\NDFe\log_2\left(\frac{1}{\varepsilon}\right)
		+
		2\sum_{n=1}^{\NDFe}\log_2\sigma_n
		{}+\mathcal O(\NDFe)\\
		&=
		C'(\beta a)^q
		\log_2\left(\frac{1}{\varepsilon}\right)
		+
		\mathcal O\bigl((\beta a)^q\bigr),
	\end{aligned}
	\label{eq:intro-entropy-ndf}
\end{equation}
where the factor $2$ counts the real and imaginary parts of each
	complex mode, and $C'>0$ collects the geometry- and resolution-dependent
proportionality factor and is independent of $\beta a$ over the range
of resolutions in which $\NDFe$ is set by the electrical size of the
source. Since the
singular values drop steeply beyond $\NDFe$ \cite{bucci1989degrees},
this range is wide, although $\NDFe$ still grows slowly as $\eps\to0$.

Consider now  sources driven by a prescribed set of physical
	controls. As discussed above, the standard EIT model assumes that the source currents are arbitrary, apart from their norm, so that
	every significant mode can be excited.
	With a prescribed set of controls, the accessible fields form instead
	a curved subset of the field space, whose dimension is set by the
	number of controls rather than by the NDF.
	Let $\Theta\subseteq\mathbb R^p$ be the prescribed domain of the $p$
	independent control parameters,
and let
$\Phi:\Theta\to\mathcal H_E$
be the corresponding continuously differentiable
parameter-to-field map. For $\theta\in\Theta$, let $D\Phi(\theta)$
denote the differential of $\Phi$, represented in coordinates by its
Jacobian. We define
\begin{equation}
	\label{eq:manifold-def}
	\mathcal E\coloneqq\Phi(\Theta)\subset\mathcal H_E
\end{equation}
as the set of all electromagnetic field configurations generated as
the control parameters vary over $\Theta$. We assume 
$\mathcal E$ to be compact, and let
\[
d\coloneqq\max_{\theta\in\Theta}\operatorname{rank}_{\mathbb R}D\Phi(\theta)
	\le p .
\]
At points where the rank equals $d$, the image locally forms a
$d$-dimensional field manifold, immersed when $d=p$. The number of
	control parameters is therefore only an upper bound on $d$;
	Sec.~\ref{subsec:p-vs-d} discusses when, and why, the two differ. The
	geometry of $\mathcal E$ is described by the pullback metric
	$\mathbf g(\theta)=[g_{ij}(\theta)]$, with
	$g_{ij}=\operatorname{Re}\langle\partial_{\theta_i}\Phi,
	\partial_{\theta_j}\Phi\rangle_{\mathcal H_E}$, i.e., the Gram matrix of
	the tangent vectors; its rank equals that of $D\Phi$, and when distinct
	control settings produce distinct fields and each control changes the
	field independently of the others (an injective, regular
	parametrization, $p=d$), the Riemannian volume of
	$\mathcal E$ is $\mathrm{Vol}_d(\mathcal E)=\int_\Theta
	\sqrt{\det\mathbf g(\theta)}\,\mathrm d\theta$. The metric is examined
	in Sec.~\ref{sec:fisher_bridge}.

When the electrical size is varied, we denote by
$\Phi_s:\Theta\to\mathcal H_E$ the parameter-to-field map at
electrical size $s$ and define
\[
\mathcal E_s\coloneqq\Phi_s(\Theta);
\]
this dependence is suppressed when no ambiguity arises.

To obtain the behavior of the covering number at fine resolution,
	i.e., for small $\eps$, we assume, at fixed electrical size, that
$\mathcal E$ is a compact, embedded $C^2$ submanifold of
	$\mathcal H_E$ of dimension $d$, i.e., a smooth $d$-dimensional
	surface of fields that does not intersect itself, with finite, nonzero
	volume and possibly a boundary made of smooth pieces, as when a
	control ranges over a bounded interval rather than a full period. Its reach $\tau_{\mathcal E}$
	\cite{federer1959curvature} is the largest radius of a tube around
	$\mathcal E$ that does not intersect itself; it is limited by the
	smallest curvature radius of $\mathcal E$ and by half the distance
	between parts of $\mathcal E$ that fold back close to each other. For $\eps$ small compared
	with $\tau_{\mathcal E}$, every open ball of radius $\eps$ centered on
	$\mathcal E$ contains a nearly flat piece of $\mathcal E$, whose volume
	is close to that of a $d$-dimensional ball of the same radius, so that
	the covering number is essentially the number of such $d$-dimensional
	balls needed to cover $\mathcal E$, i.e., the ratio between the intrinsic
Riemannian volume of $\mathcal E$ and the volume of a $d$-dimensional
Euclidean ball of radius $\eps$ \cite{tikhomirov1993varepsilon}:
\begin{equation}
	\label{eq:covering-asymp}
	N_\varepsilon(\mathcal E)
		\asymp
		\frac{\mathrm{Vol}_d(\mathcal E)}{\mathrm{Vol}_d\!\bigl(B_d(\varepsilon)\bigr)}
		=
		\frac{\mathrm{Vol}_d(\mathcal E)}{V_d\varepsilon^d},
\end{equation}
where $\mathrm{Vol}_d(\mathcal E)$ denotes the intrinsic
$d$-dimensional Riemannian volume of $\mathcal E$ (e.g., its length for $d=1$,
	its area for $d=2$), $B_d(\varepsilon)\subset\mathbb R^d$ is the
Euclidean ball of radius $\varepsilon$, and
\[
V_d
	=
	\mathrm{Vol}_d \bigl(B_d(1)\bigr)
	=
	\frac{\pi^{d/2}}{\Gamma(1+d/2)}
\]is the volume of the unit ball in $\mathbb R^d$ ($V_1=2$, $V_2=\pi$),
with $\Gamma$ denoting the Euler Gamma function. Here $\asymp$ denotes
	equality up to a factor depending only on $d$, which accounts for the
	overlap of the balls: one for $d=1$ and, for $d\ge2$, the
	covering density of $\mathbb R^d$, e.g., $2\pi/(3\sqrt3)\approx1.21$ for
	$d=2$, with balls centered on a hexagonal grid \cite{conway1999sphere}.
	A boundary adds only relative corrections of order $\eps$, and the
	packing number obeys the same law with a different constant. All these
	constants contribute to $H_\eps$ only an additive term independent of
	$\eps$. Taking the base-2 logarithm separates the contributions of the
	resolution and the geometry:
	
	\begin{equation}
		\begin{aligned}
			H_\varepsilon(\mathcal E)
			&=
			\log_2\left(
			\frac{\mathrm{Vol}_d(\mathcal E)}{V_d\varepsilon^d}
			\right)
			+
			\mathcal O(1)\\
			&=
			d\log_2\left(\frac{1}{\varepsilon}\right)
			+
			\log_2\mathrm{Vol}_d(\mathcal E)
			-\log_2V_d
			+
			\mathcal O(1).
		\end{aligned}
		\label{eq:intro-entropy}
	\end{equation}
	
Accordingly, the \emph{slope} with respect to $\log_2(1/\varepsilon)$
gives the intrinsic dimension $d$, whereas the \emph{intercept},
	$\log_2\mathrm{Vol}_d(\mathcal E)$ up to a constant that depends only on
	$d$, encodes the Riemannian volume. This basis-independent law
	is the backbone of the paper: the linear NDF, the intrinsic dimension,
	and the reduced-order representations actually built from the fields,
	such as a truncated SVD or an autoencoder, are different descriptions of
	the same entropy.

By \eqref{eq:intro-entropy}, whose volume and $\mathcal O(1)$ terms
	become negligible relative to $\log_2(1/\eps)$ as $\eps\to0$, the
	dimension $d$ is also the asymptotic slope of the entropy,
\begin{equation}
	d=\lim_{\eps\to0}\frac{H_{\eps}(\mathcal E)}{\log_2(1/\eps)},
		\label{eq:eid-definition}
\end{equation}
a characterization that depends only on the set of fields, not on
	the parametrization; we call $d$ the \emph{Electromagnetic
		Intrinsic Dimension} (EID).

At the prescribed resolution, a parametric family $\mathcal E$ lies
	essentially in the significant linear field space, of real
	dimension $2\NDFe$; $d\le2\NDFe$ when $\mathcal E$ lies exactly in
	that space (in general, $d$ also counts field variations weaker than
	$\eps$), and, in strongly parameter-constrained systems, $d$ is
typically much smaller.

In all the families of this paper, the intrinsic dimension $d$ is
independent of the electrical size: the controls and their domain
$\Theta$ are the same at every size, and the rank of $D\Phi_s$ could drop only at isolated sizes. We
assume this from now on. As the electrical size increases, the highly
oscillatory nature of the radiated fields induces a geometric
stretching of the manifold $\mathcal E_s$; since $d$ does not change,
this dependence must be completely absorbed by the Riemannian volume.
A longer manifold, however, gives more distinguishable fields only if
it does not also become tighter: a curve can grow longer through small wiggles that all stay within a
distance $\eps$ of a shorter curve, without adding distinguishable
fields. We therefore assume that the $\mathcal O(1)$ term of
\eqref{eq:intro-entropy} and the largest resolution $\eps$ at which it
holds do not depend on $s$; this condition is made precise in
Sec.~\ref{subsec:stable_widths_lipschitz}. Assuming a power law
\[
\mathrm{Vol}_d(\mathcal E_s)
\propto
s^\nu,
\]
substitution into \eqref{eq:intro-entropy} parameterizes the
$\varepsilon$-entropy by the volumetric stretching exponent $\nu$:
\begin{equation}
	H_\varepsilon(\mathcal E_s)
	=
	d\log_2\left(\frac{1}{\varepsilon}\right)
	+
	\nu\log_2 s
	+
	\mathcal O(1).
	\label{eq:intro-entropy-nu}
\end{equation}

Equations \eqref{eq:intro-entropy-ndf} and
\eqref{eq:intro-entropy-nu} expose a structural duality: in both, the entropy grows with
$\log_2(1/\varepsilon)$, but electrical size changes the \emph{slope} in
the linear case ($2\NDFe\propto(\beta a)^q$, over the range of
resolutions described above) and only the metric \emph{intercept}
$\nu\log_2 s$ in the parametric case, where the slope remains equal
to $d$.

\subsection{Normalization, Noise, and Observation Conditions}
\label{subsec:normalization}

In classical EIT the geometry is simple: the radiation operator is
linear, so it maps the unit ball of source currents onto an ellipsoid
whose significant semi-axes are the $\NDFe$ singular values above
$\eps$, and scaling the source currents scales all semi-axes by the same
factor, without changing the shape. A parametric family has a richer
geometry: the volume of the curved manifold $\mathcal E_s$ grows with
electrical size both because the field norm grows and because the
shape of the fields changes faster with the controls. This subsection
separates the two contributions and relates them to the measurement noise;
how the growth is distributed among the directions of the controls is
discussed in Sec.~\ref{sec:fisher_bridge}.

The field norm is measured by the
root-mean-square norm over $\mathcal E_s$,
\begin{equation}
	\|F_s\|^2\coloneqq\frac{1}{|\Theta|}\int_\Theta
	\|\Phi_s(\theta)\|^2_{\mathcal H_E}\,\mathrm d\theta,
	\label{eq:rms-norm}
\end{equation}
where $|\Theta|$ is the volume of the control domain. In the
following, a \emph{relative} distance is a distance divided by
$\|F_s\|$; in particular, $\hat\eps\coloneqq\eps/\|F_s\|$ is the
relative accuracy, expressed in decibels as $20\log_{10}\hat\eps$. A change of the source constraint, e.g., from
fixed excitation to fixed radiated power, rescales the whole family by
a factor $\kappa_s>0$, which may depend on $s$ but not on the controls:
then $\mathbf g$ becomes $\kappa_s^2\mathbf g$, the volume
$\kappa_s^d\,\mathrm{Vol}_d(\mathcal E_s)$, and
$H_\eps(\kappa_s\mathcal E_s)=H_{\eps/\kappa_s}(\mathcal E_s)$. The
normalized family $\hat{\mathcal E}_s\coloneqq\mathcal E_s/\|F_s\|$ is
invariant under every such rescaling, and
$\mathrm{Vol}_d(\mathcal E_s)=\|F_s\|^d\,\mathrm{Vol}_d(\hat{\mathcal E}_s)$.
If $\|F_s\|\propto s^{\nu_{\mathrm g}}$ and
$\mathrm{Vol}_d(\hat{\mathcal E}_s)\propto s^{\hat\nu}$, then
\begin{equation}
	\nu=\hat\nu+d\,\nu_{\mathrm g}.
	\label{eq:nu-split}
\end{equation}

The volume of the family thus changes with $s$ for two distinct
reasons: the field norm $\|F_s\|$ grows, which gives a gain term
$s^{d\nu_{\mathrm g}}$, and, at a fixed field norm, the spatial
distribution of the fields changes faster with the controls, which
gives a shape term $s^{\hat\nu}$. The \emph{gain exponent}
$\nu_{\mathrm g}$ depends on the source constraint, while the
\emph{shape exponent} $\hat\nu$, like $d$, does not.

We now connect $\eps$ to the measurement noise, modeled as additive
white Gaussian noise (AWGN) of standard deviation $\sigma$ on the
observed field samples. With $\eps=c_\sigma\sigma$, neighboring
states of a $2\eps$-packing are at least $2\eps$ apart, and the constant
$c_\sigma=\mathcal O(1)$ sets the probability that a decision rule
choosing the nearest state confuses them; for example,
$c_\sigma\approx2.2$ gives an error probability of $10^{-3}$ between
neighboring states at small $d$, whereas for large $d$,
$c_\sigma\simeq\sqrt{d/2}$ and the noise region becomes the
deterministic $\eps$-ball \cite{migliore2026unification}
(Appendix~\ref{supp:notation}). For a parametric family, only the noise
component along the $d$-dimensional manifold can cause a decision error
(to first order, i.e., for noise small compared with the reach of the
manifold), whatever the number of observation samples.
Let $\mathrm{SNR}=\|F_s\|^2/\sigma^2$ be the signal-to-noise ratio of
the observed fields. Then
$\log_2(\|F_s\|/\eps)=\tfrac12\log_2\mathrm{SNR}-\log_2c_\sigma$,
and substituting \eqref{eq:nu-split} into \eqref{eq:intro-entropy-nu}
gives
\begin{equation}
	\begin{aligned}
		H_\eps(\mathcal E_s)
		&=d\log_2\frac{\|F_s\|}{\eps}+\hat\nu\log_2 s+\mathcal O(1)\\
		&=\frac{d}{2}\log_2\mathrm{SNR}+\hat\nu\log_2 s
		+\mathcal O(1),
	\end{aligned}
	\label{eq:entropy-snr}
\end{equation}
with $c_\sigma$ absorbed in the $\mathcal O(1)$ term. For a given field
set, the same derivation applied to the packing number gives the
$\eps$-capacity
\begin{equation}
	C_\varepsilon(\mathcal E)
	=
	\log_2M_{2\varepsilon}(\mathcal E)
	\approx
	\frac{d}{2}\log_2\mathrm{SNR}
	+
	\log_2\mathrm{Vol}_d(\hat{\mathcal E})
	+
	\mathcal O(1).
	\label{eq:shannon-mapping}
\end{equation}

The terms of \eqref{eq:shannon-mapping} have a clear interpretation in
the language of MIMO communications \cite{TseVis:B05}. The pre-log
$d/2$ plays the role of the multiplexing gain, i.e., the number of
parallel complex channels; for the unconstrained linear family,
\eqref{eq:intro-entropy-ndf} gives instead $\NDFe$, which grows with
the radiating aperture. At a fixed source constraint, e.g., unit
excitation of each element, a larger aperture raises the SNR as
$s^{2\nu_{\mathrm g}}$, so that the first term grows as
$d\,\nu_{\mathrm g}\log_2 s$, the counterpart of the array gain. The
normalized volume acts as a \emph{rate offset}, which grows as
$\hat\nu\log_2 s$, provided that the entropy law holds uniformly for
the sets $\hat{\mathcal E}_s$ (Sec.~\ref{subsec:stable_widths_lipschitz}):
the fields change more rapidly with the controls, so that more fields
of equal norm can be distinguished, even at a fixed SNR. The total
growth is $\nu\log_2 s$, by \eqref{eq:nu-split}.

As an example, for the steered array of Sec.~\ref{sec:steering_case_study}, with $N$
elements at fixed inter-element spacing, the electrical size grows as
$N$. With unit excitation, the SNR grows as $N$
($\nu_{\mathrm g}=1/2$), and the number of resolvable beam positions in
a fixed scan range also grows as $N$ ($\hat\nu=1$). Doubling the
aperture thus adds about $1.5$~bits: half a bit from the higher SNR and
one bit from the larger number of distinguishable beams. Only the
latter contribution remains at a fixed radiated power.

This makes the rate--hardware trade-off quantitative. Accessing all
$\NDFe\propto(\beta a)^q$ significant modes of an unconstrained linear
family generally requires hardware whose complexity grows without bound
with the aperture, whereas a fixed-$d$ control law, e.g., a phase-only
steered array, keeps the pre-log of \eqref{eq:shannon-mapping} fixed at
$d/2$ and gains from a larger aperture only through the term
$\nu\log_2 s$, which grows logarithmically with it: $\nu$ is the
exchange rate between aperture and rate, $\nu_{\mathrm g}$ through the
field norm and $\hat\nu$ through the shape of the fields.

The analysis above concerns a common rescaling of all fields. To
isolate the change in the spatial distribution of the field alone,
independently of its level, each field can instead be normalized
individually: $\check\Phi_s(\theta)=\Phi_s(\theta)/A(\theta)$, with
$A=\|\Phi_s(\theta)\|_{\mathcal H_E}$ the field level, i.e., its norm.
Since
$\operatorname{Re}\langle\partial_i\Phi_s,\Phi_s\rangle=A\,\partial_iA$,
the pullback metric of $\check\Phi_s$ is
\begin{equation}
	\check g_{ij}=\frac{g_{ij}-\partial_iA\,\partial_jA}{A^2},
	\label{eq:perconfig-metric}
\end{equation}
i.e., only changes of the pattern shape are measured. When $A$ is
constant on $\mathcal E_s$, then $A=\|F_s\|$, and
\eqref{eq:perconfig-metric} reduces to the metric of the normalized
family $\hat{\mathcal E}_s$, giving the shape exponent $\hat\nu$.
Otherwise, it also weights each field by $1/A^2$, which changes the
volume and, if the relative amplitude varies with $s$, its growth.

\subsection{Control Parameters, Intrinsic Dimension, and Injectivity}
\label{subsec:p-vs-d}

Broadly speaking, a compact representation describes each field of
	$\mathcal E$ by a few real numbers, from which the field can be
	recovered. We define the \emph{minimum embedding dimension} $L_{\min}$ as the
	smallest number of real numbers that can describe every field, with
	distinct fields always described by distinct numbers and nearby fields
	by nearby numbers. It depends only on the fields, not on
	the controls, which may be more numerous than the field variations they
	produce, or may produce the same field more than once.

Three mechanisms make the field manifold differ from the control space.

\emph{Redundant or unobservable controls} give $d<p$: some
	combinations of the controls leave the observed field unchanged.

\emph{Weak controls} change the field only slightly, so that
$\mathcal E$ is thin along the corresponding directions. To describe
this effect, we define the dimension seen at resolution $\eps$,
$d_\eps$, as the local slope of $H_\eps$ versus $\log_2(1/\eps)$ over a
range of resolutions around $\eps$. When $\eps$ exceeds the field
change produced by the weak controls, these directions are not
resolved, $\mathcal E$ looks lower-dimensional, and $d_\eps$ is below
$d$. Conversely, $d_\eps$ can exceed $d$ around resolutions of the order
of the reach $\tau_{\mathcal E}$, where distant parts of $\mathcal E$
that come close to each other separate as $\eps$ decreases.  In the families of this paper, $d_\eps$ is close
to $d$ at fine resolution (Appendix~\ref{supp:pvsd}).

\emph{Symmetries} can make $\Phi$ non-injective, as grating
	lobes do for element spacings above $\lambda/2$. If $\Phi$ covers
	almost every field exactly $m$ times, the pullback volume
	$\int_\Theta\sqrt{\det\mathbf g}\,\mathrm d\theta$ counts each field $m$
	times, so that $\mathrm{Vol}_d(\mathcal E)$ is this integral divided by
	$m$. Near-symmetries make $\Phi$ nearly $m$-to-one: distinct control
	settings then produce nearly identical fields.

Accordingly, if the control space is compact, with periodic controls
taken on a circle, and $\Phi$ is injective, the field manifold is a
continuous deformation of the control space that neither tears nor
glues it, and $L_{\min}$ equals the value for the control space itself,
e.g., $1$ for an interval, $2$ for a circle and $3$ for a torus. If some distinct controls
give the same field, because of redundant controls or exact symmetries,
$L_{\min}$ is set by the distinct fields only and may be smaller. Weak
controls and near-symmetries leave $L_{\min}$ unchanged, since the
fields remain distinct, but at a finite resolution $\eps$ they can make
distinct fields indistinguishable, or require a more sensitive
representation to separate them. Their effect is measured by the
injectivity modulus $\delta(\Delta)$, the smallest distance between two
fields whose controls differ by at least $\Delta$. In all cases, for a
smooth field manifold that does not intersect itself,
$d\le L_{\min}\le2d$ by the Whitney embedding theorem
\cite{whitney1944self}.

As an example, a closed curve of fields, such as that traced by a load
phase or an incidence angle over a full turn, has $L_{\min}=2$: any
continuous code with one coordinate gives the same value to some pair
of fields half a period apart, $\Phi(\phi)$ and $\Phi(\phi+\pi)$, so that every representation with a single coordinate
has worst-case error
\begin{equation}
	\eps\ \ge\ \frac{w}{2},
	\qquad
	w\coloneqq\min_\phi\|\Phi(\phi)-\Phi(\phi+\pi)\|.
	\label{eq:antipodal}
\end{equation}
The \emph{antipodal width} $w$ is the smallest distance between two
fields half a period apart. It is small when the curve folds, so that
these fields come close, and the bound then rules out a single
coordinate only at fine accuracy (an elementary proof is given in
Appendix~\ref{supp:pvsd}).

\section{The Pullback Metric and Its Fisher-Information Interpretation}
\label{sec:fisher_bridge}

  Communication and sensing systems are described by two families of
  quantities: information-theoretic ones, such as the channel capacity of
 Shannon's theory, and estimation-theoretic ones, such as the Fisher
	information and the Cram\'er--Rao bound. Section~\ref{sec:entropy_capacity}
	related the geometry of the field manifold to the first: the
	$\eps$-capacity counts how many field states the whole family contains
	at a given noise level. This section relates the same geometry to the
	second. The Fisher information answers a complementary, local question:
	how accurately the controls can be estimated from one noisy
	observation, and along which directions. Both answers rest on the same metric $\mathbf g$, and hence on the
	geometry of the set of fields.

The estimation viewpoint also connects the framework directly
	with Information Geometry \cite{amari2016information} and with
Electromagnetic Signal and Information Theory (ESIT)
\cite{DiRMig:24}. 

Consider the
observation model, with
$\mathcal H_E=\mathbb C^{N_{\mathrm{obs}}}$,
\[
\mathbf y=\Phi(\theta)+\mathbf n,
	\qquad
	\mathbf n\sim\mathcal{CN}\bigl(\mathbf 0,\sigma^2\mathbf I_{N_{\mathrm{obs}}}\bigr).
\]
The Kullback--Leibler divergence between the densities of $\mathbf y$
at $\theta$ and $\theta+\Delta\theta$ is
		$\sigma^{-2}\|\Phi(\theta)-\Phi(\theta+\Delta\theta)\|^2_{\mathcal H_E}$
	\cite{cover2006elements,amari2016information}.
	Its local quadratic expansion gives the Fisher information matrix (FIM) $\mathbf J(\theta)$
	\cite{kay1993fundamentals,vantrees2002optimum},
\begin{equation}
	\bigl[\mathbf J(\theta)\bigr]_{ij}
	=
	\frac{2}{\sigma^2}g_{ij}(\theta),
	\label{eq:fim-equals-metric}
\end{equation}
so that, up to the factor $2/\sigma^2$, the pullback metric is the
Fisher--Rao metric of the Gaussian model \cite{amari2016information}.

\begin{figure}
	\centering
	\includegraphics[width=0.6\linewidth]{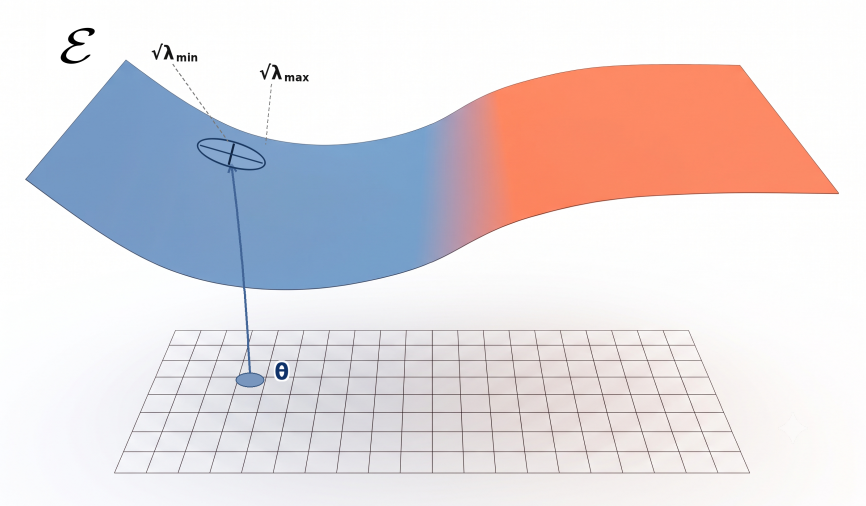}
	\caption{Geometric and statistical readings of the same pullback
		metric. A small parameter ball around $\theta$ (bottom grid, parameter
		space $\Theta$) is carried by the differential $D\Phi(\theta)$ onto a
		tangent ellipse on the field manifold $\mathcal E$ (top), with
		semi-axes $\sqrt{\lambda_{\min}(\theta)}$ and
		$\sqrt{\lambda_{\max}(\theta)}$ set by the eigenvalues of the
		pullback metric $g_{ij}(\theta)$ of Sec.~\ref{subsec:metric_entropy_definition}.
		Under the additive Gaussian noise model of this section, the same
		eigenvalues set the Fisher information,
		$\mathbf J(\theta)=\tfrac{2}{\sigma^2}\mathbf g(\theta)$, Eq.~\eqref{eq:fim-equals-metric}:
		a long semi-axis (large $\lambda$) marks a direction of high field
		sensitivity and hence low Cram\'er--Rao uncertainty, while a short
		semi-axis (small $\lambda$) marks one close to non-identifiability,
		with diverging Cram\'er--Rao uncertainty, as when two steering directions merge.}
	\label{fig:pullback-fisher-ellipse}
\end{figure}

	Let $\lambda_k(\theta)$, $k=1,\ldots,d$, be the eigenvalues of
	$\mathbf g(\theta)$ (Fig.~\ref{fig:pullback-fisher-ellipse}). If they
	scale as $\lambda_k\propto s^{2\nu_k}$ at every $\theta\in\Theta$,
	we call $\nu_k$ the
	\emph{directional exponents}: $\sqrt{\lambda_k}$, the rate at which the
	field changes along the $k$-th principal direction, grows as
	$s^{\nu_k}$. Hence $\sqrt{\det\mathbf g}\propto s^{\sum_k\nu_k}$, with
	$\sum_k\nu_k=\nu$, and $\det\mathbf J\propto s^{2\nu}$, while the
	Cram\'er--Rao bound along the $k$-th local principal direction scales
	as $s^{-2\nu_k}$: the individual exponents set the estimation accuracy
	along each direction, and their sum the growth of the volume.
	Since the normalization divides every $\lambda_k$ by
	$\|F_s\|^2$, $\nu_k=\hat\nu_k+\nu_{\mathrm g}$, with $\hat\nu_k$ the
	directional exponents of the normalized family and
	$\sum_k\hat\nu_k=\hat\nu$, consistently with \eqref{eq:nu-split}.
	Moreover, $\mathbf J=2\,\mathrm{SNR}\,\hat{\mathbf g}$, where
	$\hat{\mathbf g}$ is the metric of the normalized family, so that the
	bound scales as $s^{-2\nu_k}$ at fixed noise and source amplitude, and
	as $s^{-2\hat\nu_k}/\mathrm{SNR}$ at fixed SNR. For beam steering with $N$ elements, these are the
	classical $N^{-3}$ law at fixed per-element SNR and $N^{-2}$ law at
	fixed array-output SNR \cite{vantrees2002optimum}.
	When an eigenvalue of $\mathbf g$ approaches
zero, the field becomes insensitive to the corresponding combination of
	controls, and the Cram\'er--Rao bound diverges along it, as when two steering directions merge. The effective number of control directions along which the field
	changes is measured by the participation ratio of the eigenvalues,
	\begin{equation}
		d_{\mathrm{eff}}(\theta)
		=
		\frac{\bigl(\sum_{i=1}^{d}\lambda_i(\theta)\bigr)^2}{\sum_{i=1}^{d}\lambda_i^2(\theta)},
		\label{eq:deff}
	\end{equation}
	which equals $d$ when all eigenvalues are equal and decreases toward
	one as some directions become much weaker than the others; unlike the
	Riemannian volume, it depends on the parametrization.

By \eqref{eq:fim-equals-metric},
$\sqrt{\det\mathbf J}=(\sqrt2/\sigma)^d\sqrt{\det\mathbf g}$, so that
$\mathrm{Vol}_d(\mathcal E)=(\sigma/\sqrt2)^d\int_\Theta
\sqrt{\det\mathbf J}\,\mathrm d\theta$: the volume of the field set is
its Fisher volume measured in units of the noise. The volume law
behind \eqref{eq:shannon-mapping}, at $\eps=c_\sigma\sigma$, is thus
the classical count of distinguishable probability distributions
	\cite{balasubramanian1997,myung2000},
\begin{equation}
	C_\eps(\mathcal E)
		=
		\log_2\int_\Theta\sqrt{\det\mathbf J(\theta)}\,\mathrm d\theta
		+\mathcal O(1),
		\label{eq:capacity-fisher}
\end{equation}
at high SNR, when distinct controls give distinct fields, with an
	$\mathcal O(1)$ term that depends on $d$ and $c_\sigma$. The same law gives the Shannon capacity of the link from the
	controls to the observed field
	\cite{clarke1990information,clarke1994jeffreys}.
	Fig.~\ref{fig:kolmogorov-fisher} illustrates why the Fisher count and
	the packing of $\mathcal E$ measure the same volume, and why they
	differ where distant parts of $\mathcal E$ come close.

The correspondence requires, however, that distant control states
	produce distinct fields. Fisher information is computed from small
	parameter changes, so it measures only how well each state can be
	distinguished from its immediate neighbors, whereas the packing
	compares all pairs of fields. When distant controls produce the same
	field, as for grating lobes, the Fisher volume inherits the
	multiplicity $m$ of the pullback volume and must be divided by it, as
	in Sec.~\ref{subsec:p-vs-d}. 
	When they produce close but distinct
	fields, i.e., when $\mathcal E$ folds back to within $2\eps$ of itself,
		the Fisher volume is still the volume of $\mathcal E$, but $\eps$
	exceeds the reach
	$\tau_{\mathcal E}$ of Sec.~\ref{subsec:metric_entropy_definition}, and
		the Fisher volume can overestimate the packing count, by an amount
	that depends on the separation and arrangement of the branches. 
	In both cases, a maximum-likelihood estimator can return a distant
	state instead of the true one, so that its error far exceeds the
	Cram\'er--Rao bound: this is the threshold effect of estimation theory
	\cite{vantrees2002optimum}. In the second case, the
	counts agree again at high SNR, once $2\eps$ falls below the distance
	between the branches.
	
\begin{figure}[t]
	\centering
	\includegraphics[width=0.6\linewidth]{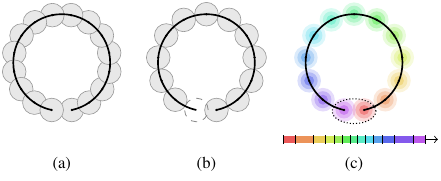}
	\caption{Three ways of counting the field states of a
			parametric system, drawn for one control parameter ($d=1$): the black
			curve is the set $\mathcal E$ of the fields that the control can
			produce, and the dots mark the centers of the balls or clouds.
			(a)~Covering: the smallest number $N_\eps$ of balls of radius $\eps$
			that contain every field of $\mathcal E$ gives the $\eps$-entropy
			$H_\eps=\log_2N_\eps$, i.e., how many reference fields are needed to
			represent any field within $\eps$. (b)~Packing: the largest number
			$M_{2\eps}$ of fields of $\mathcal E$ whose mutual distances exceed
			$2\eps$, i.e., the centers of non-overlapping balls of radius $\eps$,
			gives the $\eps$-capacity $C_\eps=\log_2M_{2\eps}$, i.e., how many
			fields are mutually distinguishable. (c)~Noise: each observed field state is
			blurred by Gaussian noise into a cloud of width $\sim\sigma$, colored
			according to the control parameter $\theta$; for $\eps=c_\sigma\sigma$,
			two clouds whose centers are at least $2\eps$ apart, as in (b), are
			reliably distinguished. On the parameter axis (colored bar
			below), each cloud corresponds to an interval of length
			$\propto\mathbf J^{-1/2}$, with $\mathbf J$ the Fisher information
			\eqref{eq:fim-equals-metric}, a scalar for $d=1$, short where the
			field changes rapidly with
			$\theta$; the number of intervals is the Fisher volume of
			\eqref{eq:capacity-fisher}. At the bottom, the two ends of the curve,
			produced by distant values of $\theta$, come within $2\eps$: the
			packing must leave out a ball (dashed circle), whereas the Fisher
			count, which compares each state only with its neighbors, keeps both
			clouds (dotted ellipse). The figure is the parametric counterpart of
			\cite[Fig.~3]{migliore2026unification}, where, for the linear set of
			fields radiated by bounded sources, covering, packing, and a
			stochastic (rate--distortion) count give, asymptotically, the same
			number of degrees of freedom, fixed by the phase-space volume of the
			source; here, at fixed dimension $d$, they give the same volume
			$\mathrm{Vol}_d(\mathcal E)$ of the manifold.}
	\label{fig:kolmogorov-fisher}
\end{figure}

\section{Limits of Stable Field Representations}
\label{sec:limits}

\subsection{Decoder-Stable Widths and a Metric-Volume Lower Bound}
\label{subsec:stable_widths_lipschitz}

A compact representation, as introduced in Sec.~\ref{subsec:p-vs-d},
consists of two maps. For a latent dimension $L\in\mathbb N$, a
continuous encoder $\Psi_{\mathrm{enc}}\in\mathcal C(\mathcal E,
\mathbb R^L)$ maps each field into its latent code
$z\in\mathbb R^L$, a vector of $L$ latent coordinates, and a
decoder $\Phi_{\mathrm{dec}}\in\mathcal C(\mathbb R^L,\mathcal H_E)$
reconstructs the field from it. Unlike $d$ and $\NDFe$, $L$ is not a
property of the fields but a design choice. The minimum embedding
dimension $L_{\min}$ requires only a continuous one-to-one encoder,
with no limit on the sensitivity of the decoder.

In practice, the latent code is stored, transmitted, quantized, or
estimated, and is therefore known only approximately: a useful
representation must be accurate and stable: small errors on the latent
coordinates must produce small errors on the field.
Approximation theory measures accuracy by widths. For a linear field
set, the smallest worst-case error achievable with $L$ coordinates,
i.e., with an $L$-dimensional subspace, is the Kolmogorov $L$-width
\cite{pinkus2012n}; non-linear widths extend it to encoder--decoder
pairs that are not linear \cite{devore1989optimal,cohen2022optimal}.

Stability, instead, is guaranteed only in the linear case. A linear
decoder is always \emph{Lipschitz}, i.e., it amplifies the errors on
the coordinates by at most a finite factor, its largest singular value.
A non-linear decoder may not be Lipschitz: a continuous decoder can map
a single latent coordinate onto a square, as the Peano curve does, and
hence onto fields of any dimension, by amplifying small changes of the
coordinate without bound, whereas a Lipschitz map cannot map a set onto
one of higher dimension \cite{falconer2014fractal}. Stability is
therefore imposed by requiring the maps to be Lipschitz; in non-linear
widths, this is usually required of both the encoder and the decoder
\cite{cohen2022optimal}. Since the latent errors of interest arise
after the encoder, we require it of the decoder only, as in the
Lipschitz widths of \cite{petrova2023lipschitz}, and keep the encoder
merely continuous, as in \cite{devore1989optimal}. The decoder must
thus satisfy, for a finite constant $\gamma>0$,
\begin{equation}
	\begin{aligned}
		\left\|
		\Phi_{\mathrm{dec}}(z_1)
		-
		\Phi_{\mathrm{dec}}(z_2)
		\right\|_{\mathcal H_E}
		&\le
		\gamma
		\left\|
		z_1-z_2
		\right\|_{\mathbb R^L},
		\\
		&\hspace{-2.2cm}
		\forall\, z_1,z_2\in\mathbb R^L.
	\end{aligned}
	\label{eq:decoder-lipschitz}
\end{equation}
We denote by $\operatorname{Lip}(\Phi_{\mathrm{dec}})$ the smallest
value of $\gamma$ for which \eqref{eq:decoder-lipschitz} holds. Alone,
it is not a measure of stability: multiplying the latent coordinates by
any $R>0$ and dividing the decoder input by $R$ leaves every
reconstructed field unchanged, but divides
$\operatorname{Lip}(\Phi_{\mathrm{dec}})$ by $R$. We instead introduce
the normalized decoder sensitivity
\begin{equation}
	\Gamma_{\mathrm{dec}}
	(\Psi_{\mathrm{enc}},\Phi_{\mathrm{dec}})
	\coloneqq
	\operatorname{Lip}(\Phi_{\mathrm{dec}})
	\operatorname{diam}
	\bigl(\Psi_{\mathrm{enc}}(\mathcal E)\bigr),
	\label{eq:normalized-sensitivity}
\end{equation}
where $\operatorname{diam}\bigl(\Psi_{\mathrm{enc}}(\mathcal E)\bigr)$
is the Euclidean diameter of the latent representation.
$\Gamma_{\mathrm{dec}}$ is invariant under the above rescaling and
measures how much the decoder can amplify errors on the latent
coordinates, relative to the size of the latent representation; it does
not measure how the whole encoder--decoder pair responds to errors on
its input field.

We thus define the \emph{decoder-stable width} as the smallest
worst-case error achievable with $L$ coordinates by a continuous
encoder and a Lipschitz decoder whose sensitivity does not exceed a
prescribed bound $\bar\Gamma$:
\begin{equation}
	\delta_{L,\bar\Gamma}^{\mathrm{dec}}
	(\mathcal E)_{\mathcal H_E}
	\coloneqq
	\inf_{\substack{
			\Psi_{\mathrm{enc}}
			\in\mathcal C(\mathcal E,\mathbb R^L)
			\\
			\Phi_{\mathrm{dec}}
			\in\mathrm{Lip}(\mathbb R^L,\mathcal H_E)
			\\
			\Gamma_{\mathrm{dec}}
			(\Psi_{\mathrm{enc}},\Phi_{\mathrm{dec}})
			\le\bar\Gamma
	}}
	\;
	\sup_{f\in\mathcal E}
	\left\|
	f-
	\Phi_{\mathrm{dec}}
	\bigl(\Psi_{\mathrm{enc}}(f)\bigr)
	\right\|_{\mathcal H_E},
	\label{eq:decoder-stable-width}
\end{equation}
where $\mathrm{Lip}(\mathbb R^L,\mathcal H_E)$ denotes the class of
Lipschitz maps from $\mathbb R^L$ to $\mathcal H_E$. We call
\emph{admissible} any pair with a continuous encoder and a Lipschitz
decoder, as in \eqref{eq:decoder-stable-width}, whatever its
sensitivity.

The sensitivity required of a decoder follows from a classical
covering argument \cite{cohen2022optimal,petrova2023lipschitz}. A
decoder with sensitivity $\Gamma_{\mathrm{dec}}$ maps the latent
representation onto a set covered by about
$(\Gamma_{\mathrm{dec}}/\eps)^L$ balls of radius $\eps$ in field space,
whereas covering $\mathcal E$ requires about $\mathrm{Vol}_d/\eps^d$ of
them. A decoder that reproduces a large field set with few latent
coordinates must therefore stretch distances, i.e., be sensitive (Appendix~
\ref{app:minimum_constant}).

By linking $\delta_{L,\bar\Gamma}^{\mathrm{dec}}$ to the metric
entropy via the packing--covering relationship
\eqref{eq:sandwich}, we obtain the following lower bounds on the
decoder sensitivity (Fig.~\ref{fig:stretching-schematic}).

\begin{figure}
	\centering
	\includegraphics[width=0.6\linewidth]{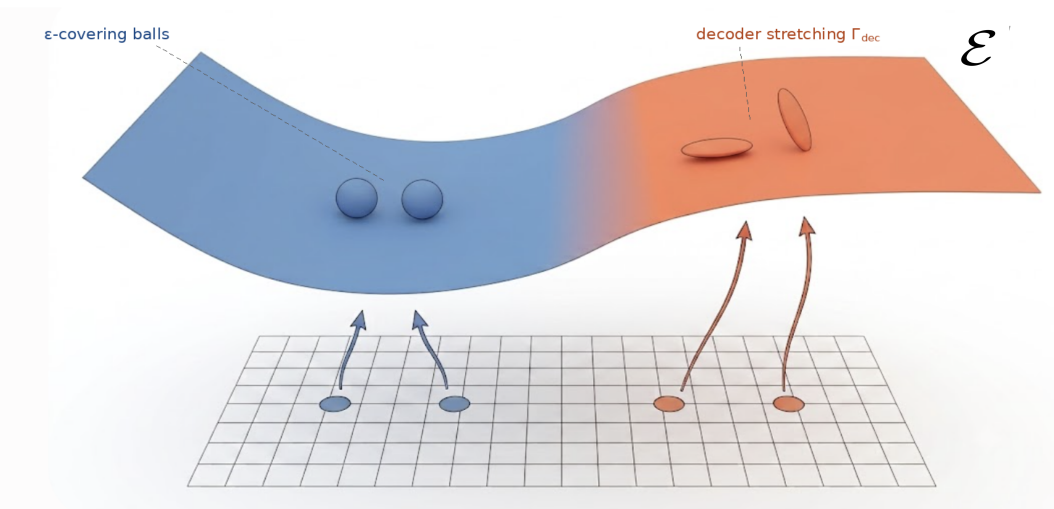}
	\caption{Schematic intuition for the decoder-sensitivity bound. The
		field manifold $\mathcal E$ (top ribbon) is decoded from a latent
		code $\mathbf z\in\mathbb R^L$ (bottom grid) through
		$\Phi_{\mathrm{dec}}$ (arrows). By the normalization underlying
		$\Gamma_{\mathrm{dec}}$, latent balls all have the same size: where
		the manifold has low curvature (blue), they map to compact
		$\varepsilon$-covering balls; where it folds tightly (orange), they
		are forced into elongated images---the local decoder stretching made
			precise by the bound below.}
	\label{fig:stretching-schematic}
\end{figure}

\begin{proposition}[Metric-volume lower bound for stable decoding]
	\label{prop:lipschitz}
	
	Let $\mathcal{E}\subset\mathcal H_E$ be a smooth, compact $d$-dimensional field manifold, possibly with boundary, with Riemannian volume $\mathrm{Vol}_d(\mathcal{E})$.		
	For all sufficiently small $\eps>0$, if an admissible encoder--decoder pair $(\Psi_{\mathrm{enc}},\Phi_{\mathrm{dec}})$ as in \eqref{eq:decoder-stable-width} reconstructs every field of $\mathcal{E}$ to within accuracy $\eps$ (in the $\mathcal{H}_E$ norm) with $L=d$ latent coordinates, then its normalized sensitivity satisfies
	\begin{equation}
		\Gamma_{\mathrm{dec}}(\Psi_{\mathrm{enc}},\Phi_{\mathrm{dec}}) \ge c_1 \left[ \mathrm{Vol}_{d}(\mathcal{E}) \right]^{1/d},
		\label{eq:migration}
	\end{equation}
	with $c_1 > 0$ independent of $\eps$. If, in addition, the field manifolds do not develop ever tighter bends
	or ever closer folds as the electrical size grows, $c_1$ is also
	independent of electrical size.
\end{proposition}

\begin{corollary}[Stability--dimension trade-off]
	\label{cor:stability-tradeoff}
	Under the hypotheses of Proposition~\ref{prop:lipschitz}, but for an arbitrary latent dimension $L$,  not necessarily equal to $d$, and for all sufficiently small $\eps>0$, the normalized sensitivity satisfies
	\begin{equation}
		\Gamma_{\mathrm{dec}}(\Psi_{\mathrm{enc}},\Phi_{\mathrm{dec}}) \;\ge\; c_1(L)\, \eps^{(L-d)/L}\, \big[\mathrm{Vol}_{d}(\mathcal{E})\big]^{1/L} ,
		\label{eq:stability-tradeoff}
	\end{equation}
	with $c_1(L) > 0$ independent of $\eps$ and, under the same condition on bends and folds, of electrical size. 
\end{corollary}

Both bounds are necessary, not constructive: they hold for \emph{any}
admissible pair achieving accuracy $\eps$ at that $L$, without
guaranteeing that a pair attaining them exists.

The bound of the Lipschitz constant by the covering number is
Proposition~3.6 of \cite{petrova2023lipschitz}, stated for Lipschitz
widths, which require no encoder. Combined with the non-asymptotic form of the
	entropy law,
\begin{equation}
	H_{2\eps}(\mathcal E)\ge d\log_2\frac{1}{\eps}
		+\log_2\mathrm{Vol}_d(\mathcal E)-C_0,
		\qquad 0<\eps<\eps_{\max} ,
		\label{eq:uniform-entropy-main}
\end{equation}
it gives \eqref{eq:migration} for $L=d$, with
	$c_1=\tfrac12\,2^{-(C_0+\log_2c_L)/d}$, where $c_L$ is such that any set
	of unit diameter in $\mathbb R^L$ is covered by at most $c_L(1+1/r)^L$
	balls of radius $r$ (Appendix~\ref{app:minimum_constant}).

For $\mathrm{Vol}_d(\mathcal E_s)\propto s^\nu$, the bound gives the
scaling law $\Gamma_{\mathrm{dec}}\gtrsim s^{\nu/d}$ only if
$\mathcal E_s$ does not fold back onto itself as $s$ grows. Otherwise,
distant parts of $\mathcal E_s$ come closer than $\eps$, a single ball
covers several of them, and the covering number, and with it the bound
on the decoder sensitivity, grows more slowly than the volume. Formally,
\eqref{eq:uniform-entropy-main} must hold with $\eps_{\max}$ and $C_0$
independent of electrical size. Such constants are commonly controlled
through curvature, reach, and local volume
\cite{niyogi2008finding,aamari2019estimating}; we use the sufficient
condition
\begin{equation}
	\mathrm{Vol}_d\bigl(\mathcal E\cap B_{\mathcal H_E}(x,\rho)\bigr)
		\le K\,V_d\,\rho^d,
		\qquad
		\forall x\in\mathcal E,\ 0<\rho\le\rho_0 ,
		\label{eq:local-regularity}
\end{equation}
with $\rho_0>0$ and $K\ge1$ independent of electrical size: at scales
	below $\rho_0$, the part of $\mathcal E$ inside a ball of radius $\rho$
	centered on $\mathcal E$ has at most $K$ times the volume of a
	$d$-dimensional ball of the same radius, i.e., $\mathcal E$ cannot
	fold back into the ball, or curl up inside it, an unbounded number of
	times. This yields
	\eqref{eq:uniform-entropy-main} with $\eps_{\max}=\rho_0/4$ and
	$C_0=\log_2(K\,V_d\,4^d)$, both independent of electrical size
	(Lemma~\ref{lem:regularity}, Appendix~\ref{app:regularity}).

Unlike for generic data, in EIT this condition can be checked, often
	analytically, because the field families are known from Maxwell's
	equations. Physically, it suffices that electrical size makes the field
	manifold \emph{longer}, not \emph{tighter}, i.e., that its
	curvature radius and the closest approach between its distant parts,
	which set $\rho_0$, remain bounded below as relative distances. The case studies check it for each family, analytically or numerically.

Multiplying all fields by a common factor, e.g., by raising the
	transmitted power, multiplies by the same factor the decoder
	sensitivity $\Gamma_{\mathrm{dec}}$, the reconstruction error, and the
	size $\mathrm{Vol}_d^{1/d}$ of the family. Proposition~\ref{prop:lipschitz}
	and Corollary~\ref{cor:stability-tradeoff} therefore hold unchanged,
	with the same constants, when all these quantities are measured
	relative to the field norm $\|F_s\|$: for a relative accuracy
	$\hat\eps=\eps/\|F_s\|$ and $L=d$,
\begin{equation}
	\frac{\Gamma_{\mathrm{dec}}}{\|F_s\|}
		\ge
		c_1\bigl[\mathrm{Vol}_d(\hat{\mathcal E}_s)\bigr]^{1/d}
		\propto s^{\hat\nu/d},
		\label{eq:relative-migration}
\end{equation}
provided that the relative accuracy is small enough and the
	regularity scale $\rho_0/\|F_s\|$ and the normalized volume
	$\mathrm{Vol}_d(\hat{\mathcal E}_s)$ are bounded below independently of
	$s$. Hence, at a fixed relative accuracy, the
	growth of the required decoder sensitivity with electrical size is set
	by the shape exponent $\hat\nu$ alone: the array or scattering gain
	raises the field norm, but does not make the fields harder to
	represent; only the change in the shape of the patterns does.

\subsection{An Exact Sensitivity Bound for a Non-Self-Intersecting Arc}
\label{subsec:sharp-arc}
For a non-self-intersecting arc and a single latent coordinate, the decoder maps the latent interval onto a curve of length at most $\Gamma_{\mathrm{dec}}$. The nearest-point projection within distance $\eps<\tau_{\mathcal E}$ has Lipschitz constant at most $\tau_{\mathcal E}/(\tau_{\mathcal E}-\eps)$ \cite[Th.~4.8(8)]{federer1959curvature}. Combining the two bounds gives the following result (proof in Appendix~\ref{app:sharp}).
\begin{proposition}[Exact bound for arcs]
	\label{prop:sharp}
	Let $\mathcal E\subset\mathcal H_E$ be a $C^2$ embedded arc of length
		$\mathrm{Vol}_1(\mathcal E)$, reach $\tau_{\mathcal E}>0$, curvature
		radius everywhere at least $R$, and no self-approach below a scale
		$\rho_1$, i.e., every ball of radius $\rho\le\rho_1$ centered on
		$\mathcal E$ meets it in a single piece. If an admissible
		encoder--decoder pair with $L=1$ reconstructs every field of
		$\mathcal E$ within $\eps$, where $\eps<\tau_{\mathcal E}$ and
		$2\eps\le\min(R/2,\rho_1)$, then
	\begin{equation}
		\Bigl(1-\frac{\eps}{\tau_{\mathcal E}}\Bigr)
			\Bigl[\mathrm{Vol}_1(\mathcal E)-2R\arcsin\frac{2\eps}{R}\Bigr]
			\le\Gamma_{\mathrm{dec}}.
			\label{eq:sharp}
	\end{equation}
	Conversely, the arclength chart reconstructs $\mathcal E$ exactly
		with $\Gamma_{\mathrm{dec}}=\mathrm{Vol}_1(\mathcal E)$.
\end{proposition}

The bound has a simple interpretation: the decoded curve stays within $\eps$
	of $\mathcal E$, so its length, at most $\Gamma_{\mathrm{dec}}$, cannot
	be much shorter than $\mathrm{Vol}_1(\mathcal E)$; the two corrections
	account for the end portions of the arc that the tolerance $\eps$ allows
	to be skipped and for the length change under projection onto
	$\mathcal E$. The argument does not extend directly to $d\ge2$ or
	$L>d$. Since the arclength chart attains
	$\Gamma_{\mathrm{dec}}=\mathrm{Vol}_1$, the smallest sensitivity
	compatible with accuracy $\eps$ is
	$\Gamma^\star_{\mathrm{dec}}(\eps)=\mathrm{Vol}_1-\mathcal O(\eps)$.

\subsection{Latent Dimension and Uniform Accuracy}
\label{subsec:opdimension}
For a prescribed reconstruction accuracy $\eps$ and normalized
decoder-sensitivity upper bound $\bar\Gamma$, we define the
\emph{operational stable latent dimension}
\begin{equation}
	L_{\mathrm{stab}}(\eps,\bar\Gamma) \coloneqq \min\left\{ L\in\mathbb N : \delta_{L,\bar\Gamma}^{\mathrm{dec}}(\mathcal E)_{\mathcal H_E} \le\eps \right\},
	\label{eq:stable-latent-dimension}
\end{equation}
the smallest number of latent coordinates with which some
	representation reconstructs every field within $\eps$ with sensitivity
	at most $\bar\Gamma$. The minimum embedding dimension $L_{\min}$ of
	Sec.~\ref{subsec:p-vs-d} requires only an exact one-to-one encoding,
	with no limit on sensitivity; the two requirements added by
	$L_{\mathrm{stab}}$ can only increase the number of coordinates,
\begin{equation}
	L_{\min}\le L_{\mathrm{stab}}(\eps,\bar\Gamma),
		\label{eq:dim-chain-2}
\end{equation}
whenever $\eps$ is smaller than the error that any representation
	with fewer coordinates must make, given by \eqref{eq:antipodal} for
	closed orbits. The theory thus distinguishes three cases in which these numbers of coordinates differ: $d\ll2\NDFe$ when
	a few controls act on a field with many significant modes;
	$d<L_{\min}$ when the field manifold is closed, as a loop ($d=1$,
	$L_{\min}=2$) or a torus ($d=2$, $L_{\min}=3$); and
	$L_{\min}<L_{\mathrm{stab}}$ when no representation with $L_{\min}$
	coordinates meets both the accuracy and the sensitivity targets, e.g.,
	because distant parts of $\mathcal E$ come close.

\section{From the Bounds to a Learned Representation}
\label{sec:stability_validation}\label{subsec:scope}

     The results of Secs.~\ref{sec:entropy_capacity}--\ref{sec:limits}
	depend only on the set of fields, not on how it is represented or
     processed, and therefore apply to any machine-learning model of
     parametric field sets. As an example, we apply them to
	autoencoders, which learn the encoder--decoder pairs of
	Sec.~\ref{sec:limits} directly, so that their latent dimension and
	decoder sensitivity can be compared with the bounds. 
	
\label{subsec:data_driven_ae}

An autoencoder \cite{hinton2006reducing} consists of an encoder,
which maps each field to a few latent coordinates, and a decoder,
which reconstructs the field from them. Unlike a truncated singular
value decomposition, which projects the fields onto a flat subspace of real dimension $2\NDFe$, it can learn a curved parametrization of
$\mathcal E$ with as few as $L_{\mathrm{stab}}$ coordinates, when such a
parametrization exists. We denote by $\LAE$ the number of latent
coordinates, the Autoencoder Latent Dimension. 

Let $N_{\mathrm{obs}}\in\mathbb N$ denote the number of complex field
observations and $N_{\mathrm{SVD}}\in\mathbb N$ the number of retained
dominant singular modes. The architecture comprises an encoder
$\Psi_{\mathrm{enc},\boldsymbol{\eta}}
:\mathbb R^{2N_{\mathrm{obs}}}\to\mathbb R^{\LAE}$
and a decoder
$\Phi_{\mathrm{dec},\boldsymbol{\xi}}
:\mathbb R^{\LAE}\to\mathbb R^{N_{\mathrm{SVD}}}$,
where $\boldsymbol{\eta}$ and $\boldsymbol{\xi}$ denote their respective
trainable parameters.
All complex field samples are represented in real-stacked form,
$\mathbf x=[\mathrm{Re}\,\mathbf E^{\mathsf T},\,
	\mathrm{Im}\,\mathbf E^{\mathsf T}]^{\mathsf T}
	\in\mathbb R^{2N_{\mathrm{obs}}}$, with $\mathbf E\in\mathbb C^{N_{\mathrm{obs}}}$
	the vector of the $N_{\mathrm{obs}}$ complex field samples.
The output layer is not learned: it is fixed to the left singular
vectors of the real-stacked radiation operator, truncated at a
threshold low enough that the truncation error is negligible compared
with the reconstruction error of the networks (Appendix~
\ref{app:training}). This restricts the output to solutions of Maxwell's equations for the
given radiation problem:
\begin{equation}
	\mathbf x_{\mathrm{pred}}
		=
		\sum_{i=1}^{N_{\mathrm{SVD}}}
		\alpha_i\mathbf u_i
		=
		\mathbf U_r\,
		\Phi_{\mathrm{dec},\boldsymbol\xi}(\mathbf z),
		\label{eq:svd_rec}
\end{equation}
where $\mathbf u_i$ is the $i$-th retained left singular vector and
$\mathbf U_r=[\mathbf u_1,\dots,\mathbf u_{N_{\mathrm{SVD}}}]
\in\mathbb R^{2N_{\mathrm{obs}}\times N_{\mathrm{SVD}}}$; hence
$\mathbf u_i$ and the decoder coefficients $\alpha_i$ are real.
The network thus learns only the curved part of the representation,
while the propagation physics is fixed by the basis
\cite{hanson2013operator}.
The decoder $\Phi_{\mathrm{dec},\boldsymbol\xi}$ maps to the coefficient
space, so the decoder entering $\Gamma_{\mathrm{dec}}$ is the composite
field decoder $\widetilde\Phi_{\mathrm{dec},\boldsymbol\xi}\coloneqq\mathbf
	U_r\Phi_{\mathrm{dec},\boldsymbol\xi}$. As $\mathbf U_r$ has
	orthonormal columns, it preserves distances, and
	$\mathrm{Lip}(\widetilde\Phi_{\mathrm{dec},\boldsymbol\xi})=
	\mathrm{Lip}(\Phi_{\mathrm{dec},\boldsymbol\xi})$. With finite
	weights and Lipschitz activations, such as GELU
	\cite{hendrycks2016gelu}, both networks are globally Lipschitz
	\cite{virmaux2018lipschitz}.

If a trained autoencoder reconstructs every field of $\mathcal E$
	within $\eps$ and has $\Gamma_{\mathrm{dec}}\le\bar\Gamma$, then, by
	\eqref{eq:stable-latent-dimension}, $L_{\mathrm{stab}}(\eps,\bar\Gamma)\le\LAE$:
	$\LAE$ coordinates suffice, although fewer might.

%\subsection{Comparing a Trained Autoencoder with the Bounds}
\label{subsec:ae-comparison}
Together with \eqref{eq:dim-chain-2}, the last inequality gives
	$L_{\min}\le L_{\mathrm{stab}}(\eps,\bar\Gamma)\le\LAE$;
	Table~\ref{tab:dims} collects the dimension-like quantities of the
	theory and of the autoencoder. The comparisons below use three criteria: latent dimension,
	decoder sensitivity, and reconstruction accuracy.

\begin{table}[t]
	\centering
	\caption{Dimension-like quantities of the theory and of a trained autoencoder.}
	\label{tab:dims}
	\small
	\setlength{\tabcolsep}{5pt}
	\renewcommand{\arraystretch}{1.15}
	\begin{tabularx}{\textwidth}{@{}>{\raggedright\arraybackslash}p{2.4cm}>{\hsize=0.9\hsize}L>{\hsize=1.1\hsize}L@{}}
		\toprule
		Quantity & Definition & Physical interpretation \\
		\midrule
		$p$
		& number of independent controls, $\Theta\subseteq\mathbb R^p$
		& controls provided by the hardware; upper bound on $d$ \\
		$d$ (EID)
		& $\max_\theta\operatorname{rank}_{\mathbb R}D\Phi(\theta)$; slope of $H_\eps$ as $\eps\to0$, \eqref{eq:eid-definition}
		& locally independent directions of field variation; high-SNR pre-log $d/2$ \\
		$d_\eps$
		& slope of $H_\eps$ at resolution $\eps$
		& dimension seen at finite resolution; $<d$ for weak controls, possibly $>d$ for a strongly curved or folded family \\
		$d_{\mathrm{eff}}(\theta)$
		& participation ratio of the eigenvalues of $\mathbf g(\theta)$, \eqref{eq:deff}
		& spread of sensitivity over parameter directions; $\to1$ near loss of identifiability \\
		$\NDFe$
		& number of singular values of the radiation operator above $\eps$, \eqref{eq:ndf-def}
		& significant modes of the unconstrained family; grows as $(\beta a)^q$ \\
		$N_{\mathrm{SVD}}$
		& retained real singular vectors, \eqref{eq:svd_rec}
		& output dimension of the physics-constrained decoder layer \\
		$L_{\min}$
		& minimum Euclidean embedding dimension of $\mathcal E$
		& fewest coordinates for a continuous injective encoding; topological ($2$ for $S^1$, $3$ for $T^2$) \\
		$L_{\mathrm{stab}}(\eps,\bar\Gamma)$
		& smallest $L$ with $\delta^{\mathrm{dec}}_{L,\bar\Gamma}\le\eps$, \eqref{eq:stable-latent-dimension}
		& fewest coordinates compatible with accuracy $\eps$ and sensitivity bound $\bar\Gamma$ \\
		$\LAE$
		& latent dimension of a trained autoencoder meeting the target
		& constructive upper bound on $L_{\mathrm{stab}}$; depends on architecture and training \\
		\bottomrule
	\end{tabularx}
\end{table}

\emph{Latent dimension}: $\LAE$ is compared with $L_{\min}$; a
	difference $L_{\min}<\LAE$ signals either $L_{\mathrm{stab}}>L_{\min}$, e.g., because $\mathcal E$ folds, or training limitations.

\emph{Sensitivity}: since the global Lipschitz constant of a
	trained decoder cannot be evaluated exactly, we use the sampled proxy
\begin{equation}
	\widehat L_{\mathrm{dec}}
	=
	\max\!\left\{
	\max_{i\ne j}
	\frac{\|\widehat f_i-\widehat f_j\|}{\|z_i-z_j\|},
	\max_i\|D\widetilde\Phi_{\mathrm{dec}}(z_i)\|
	\right\},
	\label{eq:lip-hat}
\end{equation}
where $z_i=\Psi_{\mathrm{enc}}(f_i)$ and
$\widehat f_i=\widetilde\Phi_{\mathrm{dec}}(z_i)$ for test fields $f_i$, and set
$\widehat\Gamma_{\mathrm{dec}}
=\widehat L_{\mathrm{dec}}
\operatorname{diam}(\Psi_{\mathrm{enc}}(\mathcal E))$.
Both terms in \eqref{eq:lip-hat} are lower bounds on
	$\mathrm{Lip}(\widetilde\Phi_{\mathrm{dec}})$: the difference quotients
	capture stretching between distant samples, the Jacobian norms the
	local stretching at each sample.	
Hence $\widehat\Gamma_{\mathrm{dec}}\le\Gamma_{\mathrm{dec}}$: a proxy
	value above a lower bound of Sec.~\ref{sec:limits} shows that the bound
	is satisfied, whereas a value below it is inconclusive.

\emph{Accuracy}: the networks are trained and qualified on
	mean errors, whereas the bounds refer to the uniform error, which must
	be checked separately.

\section{Case Study 1: Topology and Latent Dimension in a Load-Controlled Array}
\label{sec:array_case_study}

\subsection{Field Geometry and Theoretical Predictions}
\label{subsec:coupling}

The first case study investigates whether the topology and geometry
of a field family can predict, before training, the number of latent
coordinates needed for an exact continuous representation and the
error incurred with fewer coordinates.

The radiating structure is a linear array of five $\lambda/2$ dipoles
at $\lambda/4$ spacing, in which driven elements with fixed excitation
(\texttt{A}) alternate with parasitic elements (\texttt{P}) terminated
in lossless reactive loads; each configuration is labeled by its
sequence of elements, here \texttt{A\_P\_A\_P\_A}.
The controls are the phases $\phi_k$ of the load reflection
coefficients $e^{\mathrm j\phi_k}$, either \emph{tied}
($\phi_1=\phi_2$) or \emph{independent}, so that the control space is a
circle or a torus. For comparison, we also consider two arrays in which
the controls act on driven elements. In the first,
\texttt{A\_F\_A\_F\_A}, the loads are replaced by elements fed with
phase-controlled excitation $e^{\mathrm j\phi_k}$ (\texttt{F}), with
the same tied or independent phases, so that the control space is again
a circle or a torus. In the second, \texttt{A\_A\_AF\_A\_A}, the
central element is fed with controlled amplitude and phase
(\texttt{AF}), and the two controls span an annulus.

A load acts on the field only through the current on its element
\cite{harrington1978reactively,pinchera2008experimental}. By
superposition, the radiated field is
\begin{equation}
	\Phi(\boldsymbol\phi)
	=F^0+\sum_{k=1}^{2}I_{\mathrm p,k}(\boldsymbol\phi)\,\mathbf f_k,
	\qquad
	\mathbf I_{\mathrm p}
	=-\bigl(\mathbf Z_{\mathrm r}+\mathbf Z_{\mathrm L}
	(\boldsymbol\phi)\bigr)^{-1}\mathbf v,
	\label{eq:load-structure}
\end{equation}
where $F^0$ is the field with the loaded ports open, $\mathbf f_k$ the
field response to a unit current at loaded port $k$, including mutual
coupling, $\mathbf Z_{\mathrm L}=\operatorname{diag}
\bigl(\mathrm jX_k(\phi_k)\bigr)$ the loads, and $\mathbf Z_{\mathrm r}$
and $\mathbf v$ the impedance matrix of the loaded ports and the
open-circuit voltage induced at them by the driven elements.
Under the nondegeneracy conditions discussed in Appendix~
\ref{supp:pvsd}, distinct controls produce distinct fields, for both the
parasitic and the driven configurations. Each field manifold thus
has the topology of its control space, and $d$ equals the number of
independent controls. With the other loads fixed, each load current
traces a circle at non-uniform speed, producing a non-uniform
field metric (Appendix~\ref{supp:pvsd}).

Before training, topology thus determines the minimum number of
latent coordinates for a continuous one-to-one representation:
$L_{\min}=2$ for the tied loads (a loop), $3$ for the independent
loads (a torus), and $2$ for the amplitude--phase control
(an annulus). Note that the torus and the annulus have the same
intrinsic dimension, $d=2$, but different $L_{\min}$.
The metric geometry further quantifies the cost of using fewer
coordinates: the loop has antipodal width $w=0.50\,\|F\|$,
so that, by \eqref{eq:antipodal}, a single coordinate gives a
worst-case error of at least $w/2$, about $-12$~dB.

For comparison, an exact affine representation of the fields,
obtained by applying the SVD to their real and imaginary
components after subtracting the mean field, requires
two real coordinates for the loop and the annulus
and four for the torus. However, the rank alone does not reveal
the intrinsic dimension or the topological obstruction to
using fewer coordinates, nor does it quantify the resulting
reconstruction error. For the loop and the annulus, exact
linear and continuous nonlinear representations have the
same minimum dimension; for the torus, a nonlinear
representation can in principle save one coordinate.
Larger dimensional savings are illustrated in
Sec.~\ref{sec:one_dimensional_manifolds}.
Since these predictions are made before training, and the networks
are given only the fields, not the controls, they provide an
independent test of the learned representations.

\subsection{Learned Representations}
\label{subsec:array_ae}

Each field is sampled at $N_{\mathrm{obs}}=180$ points in the H-plane,
with a common scaling that preserves the amplitude variations among
load configurations, and reconstructed through the SVD layer
\eqref{eq:svd_rec}, with $N_{\mathrm{SVD}}=15$ (Appendix~
\ref{app:training}). As noted in Sec.~\ref{subsec:coupling}, the encoder never
sees the controls, so that the agreement of the latent coordinates
with them is a result, not a constraint.

\begin{figure}[t]
	\centering
	\includegraphics[width=0.6\linewidth]
	{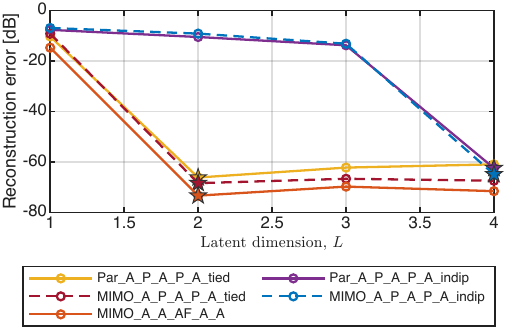}
	\caption{Reconstruction error [dB] of the physics-constrained
		autoencoder as a function of the latent dimension $L$ for the
		array configurations: mean relative error $\bar e$ over the test
		samples, averaged over three training runs. A latent dimension is
		successful when at least one of the three runs reaches
		$\bar e\le-60$~dB; pentagrams mark $L_{\mathrm{AE}}$, the smallest
		successful $L$. For all configurations shown, the run-averaged
		error is also below $-60$~dB at $L_{\mathrm{AE}}$
		($-62.4$ to $-73.3$~dB) and above $-15$~dB at smaller $L$.}
	\label{fig:figyya2}
\end{figure}

The results are consistent with the predictions for the loop and
the annulus. For the tied loads, the error drops abruptly at
$\LAE=2=L_{\min}$ (Fig.~\ref{fig:figyya2}), and at $L=1$ the mean
error remains above $-15$~dB, compatible with, although not a direct
test of, the $-12$~dB worst-case bound. For the amplitude--phase
control, the error also drops at $\LAE=2=L_{\min}$.
The learned latent representations form a circle and an annulus
(Fig.~\ref{fig:fig_latent_spaces}, left and right), reproducing the
shapes of the control spaces. The \texttt{A\_F\_A\_F\_A} arrays
show the same thresholds as the parasitic ones: the mechanism
changes the metric, not the topology.

For the independent loads, the error falls below $-60$~dB only at
$\LAE=4$, one more than $L_{\min}=3$, and at $L=3$ the learned surfaces
self-overlap (Fig.~\ref{fig:fig_latent_spaces}, center).
Although an exact continuous representation in three coordinates
exists, the networks did not learn it. Whether the extra coordinate
is needed to satisfy a prescribed stability constraint
($L_{\mathrm{stab}}>L_{\min}$), or simply reflects a limitation
of training, remains an open question.

\begin{figure*}[t]
	\centering
	\includegraphics[width=1.0\linewidth]
	{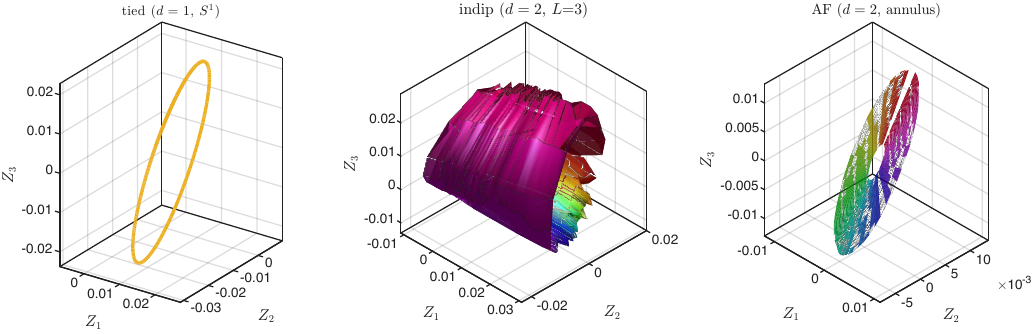}
	\caption{Emergent latent embeddings at $L=3$. Left: the tied $d=1$
		case forms a smooth circle ($S^1$), already successful at
		$L_{\mathrm{AE}}=2$. Center:
		the independent $d=2$ case (torus, $L_{\min}=3$) self-overlaps
		at $L=3$ and achieves the target accuracy at
		$L_{\mathrm{AE}}=4$, one coordinate more than the topological
		minimum. Right: the amplitude--phase (\texttt{AF})
		configuration emerges as an annulus. The
		\texttt{A\_F\_A\_F\_A} configurations, omitted, reproduce
		the same geometries and thresholds.}
	\label{fig:fig_latent_spaces}
\end{figure*}

\section{Case Study 2: Decoder Sensitivity versus Aperture in a Steered Array}
\label{sec:steering_case_study}

The second case study has two objectives. The first, and main, one is
to test the Lipschitz migration of Sec.~\ref{sec:limits}, i.e., the
consequence of the growth of the volume with electrical size, on a
steered array. Its steering family, the classical array manifold
\cite{manikas2004differential}, has a length that grows with the
aperture in closed form and, away from endfire, is a
non-self-intersecting arc. For such an arc, the bound on the decoder
sensitivity is exact (Proposition~\ref{prop:sharp}), whereas in the
other case studies only the looser bound of
Proposition~\ref{prop:lipschitz} is available. This allows an accurate
test: the theory predicts, before training, the minimum sensitivity of
an accurate decoder at each aperture, which is compared with trained
networks (Secs.~\ref{subsec:test-manifold}
and~\ref{subsec:observed}). The second objective is to show, in closed
form, that the volume exponent depends on how each control acts on the
field, not on $d$ alone; for this purpose, the array is dual-polarized,
and the polarization angle is a second control
(Sec.~\ref{sec:communication_theory}).

\subsection{Steering Arc: Field Family and Theoretical Prediction}
\label{subsec:test-manifold}
\label{subsec:dualpol_orthogonal}

We consider a dual-polarized linear array of $N$ elements centered at
$x_n=(n-(N-1)/2)\lambda/2$, so that $L_{\mathrm{ap}}=N\lambda/2$ and
$\beta L_{\mathrm{ap}}=\pi N$; the $\lambda/2$ spacing avoids grating
lobes. The field is controlled by the steering
angle $\psi$ and the linear polarization angle $\chi$:
\begin{equation}
	\mathbf{E}(u;\psi,\chi)
	=
	g(u)
	\begin{pmatrix}
		\cos\chi\\
		\sin\chi
	\end{pmatrix}
	\sum_{n=0}^{N-1}
	e^{\,\mathrm{j}\beta x_n(u-\sin\psi)},
	\label{eq:dp-field-main}
\end{equation}
where $u=\sin\vartheta$, with $\vartheta$ the observation angle, and
	$g(u)$ is the element pattern, the same for both polarizations, with
	negligible cross-polarization. In this and the next subsection, $\chi$
	is fixed.

Since, with $\lambda/2$ spacing,
	$\|\partial_\psi\mathbf E\|^2\propto\beta^2\cos^2\psi\sum_n
	x_n^2\propto N^3$, the length of the field curve swept by $\psi$,
	$\ell_\psi=\int\|\partial_\psi\mathbf E\|\,\mathrm d\psi$, scales as
\begin{equation}
	\ell_\psi
		\propto
		(\beta L_{\mathrm{ap}})^{3/2},
		\qquad
		\nu_\psi=\frac{3}{2}.
		\label{eq:prop-arclength}
\end{equation}
This exponent holds when each element is fed with unit excitation at
every aperture, so that the field norm grows as $\|F\|\propto N^{1/2}$,
i.e., $\nu_{\mathrm g}=1/2$ (Sec.~\ref{subsec:normalization}). If the
radiated power is instead kept fixed as $N$ grows, this gain is
removed, and only the shape exponent remains,
$\hat\nu_\psi=\nu_\psi-\nu_{\mathrm g}=1$.
Numerical simulations that include the mutual coupling between the
elements, modeled with the induced-EMF mutual impedances of
side-by-side $\lambda/2$ dipoles, show that coupling changes this
exponent negligibly: the fitted value is $1.51$ instead of $3/2$.

To obtain a non-self-intersecting arc, the case in which the bound
of Proposition~\ref{prop:sharp} is exact, we restrict the steering
angle to $|\psi|\le0.12$~rad. Over this range, the field family is an
arc with $d=L_{\min}=1$ and length $\mathrm{Vol}_1\propto N^{3/2}$. Its
curvature radius does not depend on the steering angle
\cite{manikas2004differential}, so the regularity condition
\eqref{eq:local-regularity} holds at every aperture, up to relative
accuracies far coarser than those used below (Table~
\ref{tab:regularity}).

Two bounds then apply, both evaluated before training. The exact bound
\eqref{eq:sharp}, computed on the analytical field
\eqref{eq:dp-field-main}, differs from $\mathrm{Vol}_1$ by less than
$1.5\%$ at all apertures considered, and the arclength parametrization
attains $\mathrm{Vol}_1$: every accurate decoder with one latent
coordinate, i.e., one that reconstructs every field within $\eps$, has
$\Gamma_{\mathrm{dec}}\gtrsim\mathrm{Vol}_1$, and the least sensitive
one has $\Gamma_{\mathrm{dec}}\simeq\mathrm{Vol}_1$. The general bound
of Proposition~\ref{prop:lipschitz} is less tight, but, unlike
\eqref{eq:sharp}, it holds in any dimension. For a curve, the constant
of \eqref{eq:migration} is $c_1=1/(16K)$, with $K$ the constant of the
regularity condition \eqref{eq:local-regularity} (Appendix~
\ref{app:minimum_constant}), so that
\begin{equation}
	\Gamma_{\mathrm{dec}}\geq\frac{\mathrm{Vol}_1(\mathcal E)}{16K}
	\quad\text{for}\quad
	\eps\leq\frac{\mathrm{Vol}_1(\mathcal E)}{16K}.
	\label{eq:prop1-explicit}
\end{equation}
For this arc, $K\approx1.01$ (Table~\ref{tab:regularity}), so that this
bound is about $\mathrm{Vol}_1/16$, sixteen times below the exact one.

\subsection{Sensitivity of the Trained Decoders}
\label{subsec:observed}
\label{subsec:certified}

Autoencoders with $L=1$ are trained on the SVD representation
\eqref{eq:svd_rec} of the fields. To test the bound, decoders with a
wide range of sensitivities are needed. Besides unconstrained
networks, we therefore train networks whose decoder sensitivity is limited by bounding the norm of its weight matrices. The latent coordinate is kept
in a fixed interval, so that the encoder cannot offset the limit on the decoder by spreading the codes (Appendix~\ref{app:training}).

A trained encoder--decoder pair is called \emph{qualified} if its RMS
error over the test angles is at most $5\times10^{-3}$ ($-46$~dB).
Qualified pairs are the practical counterpart of the accurate decoders
of Sec.~\ref{subsec:test-manifold}, with the RMS error in place of the
uniform one. Since both
this error and the sensitivity estimate \eqref{eq:lip-hat} are computed
on samples, the sensitivity estimate is only a lower bound on
the true one. To check how close the sampled values are to the true ones, we apply
interval arithmetic \cite{moore1966interval,gowal2018effectiveness},
which gives guaranteed
bounds over the whole latent interval, to an independent reproduction
of the unconstrained runs (Appendix~\ref{supp:certified}). For the pairs
that reconstruct every field within $5\times10^{-3}$, sampling proves
reliable: the sampled sensitivity and error differ from the guaranteed
values by less than $0.2\%$ and $4\%$, respectively, showing that
the sampled sensitivities of Table~\ref{tab:prop1} can be used to test
the bounds.

The results confirm the prediction of Sec.~\ref{subsec:test-manifold}.
At each aperture, the least sensitive qualified decoder
(Table~\ref{tab:prop1}) has a sampled sensitivity
$\widehat\Gamma_{\mathrm{dec}}$ close to the theoretical optimum
$\mathrm{Vol}_1$, within $15\%$ up to $N=16$, and growing with the
aperture as predicted by the Lipschitz migration \eqref{eq:migration}:
over a $23$-fold increase of $\mathrm{Vol}_1$, the fitted growth is
consistent with proportionality to $\mathrm{Vol}_1$, i.e., with
$\Gamma_{\mathrm{dec}}\propto N^{3/2}$.
	Larger apertures are harder to train: at $N=32$ the best decoder is
	$57\%$ above the optimum, and at $N=64$ only one of $28$ runs reached
	the required accuracy (not shown). This reflects the difficulty of
	fitting a longer arc with a single latent coordinate: since
	$\Gamma_{\mathrm{dec}}$ is the largest speed at which the decoder
	traces the arc (Appendix~\ref{app:minimum_constant}), and $\mathrm{Vol}_1$ its mean
	speed, any departure from a uniform parametrization raises the ratio
	above one.

\begin{table}[t]
	\centering
	\caption{Best qualified pair at each aperture on the arc
		manifold ($L=d=1$). A pair is qualified when its
			relative RMS error over $512$ test angles, with the mean field
			removed, satisfies $e_{\mathrm{rms}}\le5\times10^{-3}$ ($-46$~dB).
			``Qualified/total'': number of qualified pairs over the
			number of trained pairs at that aperture (unconstrained and
			sensitivity-limited runs). $\min\widehat\Gamma_{\mathrm{dec}}$:
			smallest sampled proxy \eqref{eq:lip-hat} among the qualified
			pairs; ratio: $\min\widehat\Gamma_{\mathrm{dec}}/\mathrm{Vol}_1$.}
	\label{tab:prop1}
	\small
	\setlength{\tabcolsep}{9pt}
	\begin{tabular}{cccccc}
		\toprule
		$N$ & $\beta L_{\mathrm{ap}}$ & $\mathrm{Vol}_1$
		& $\min\widehat{\Gamma}_{\mathrm{dec}}$
		& ratio & qualified/total \\
		\midrule
		4  & 12.6  & 1.68  & 1.76  & 1.044 & 11/29 \\
		8  & 25.1  & 4.87  & 5.17  & 1.061 & 11/29 \\
		16 & 50.3  & 13.87 & 16.00 & 1.154 & 10/29 \\
		32 & 100.5 & 39.29 & 61.67 & 1.570 & 5/29  \\
		\bottomrule
	\end{tabular}\\[2pt]
	\footnotesize Fit: $1.124\pm0.196$ (95\% CI, $t_{0.975,2}$), prefactor $0.917$.
\end{table}

Figure~\ref{fig:prop1} shows all $116$ trained pairs. Since the shaded
and hatched regions are obtained by evaluating the bounds with the RMS
error and the sampled sensitivity in place of the uniform error and the
true sensitivity, they are references, not guaranteed exclusion
regions. As predicted, no pair falls in them. The observed 
transition at $\widehat\Gamma_{\mathrm{dec}}\simeq\mathrm{Vol}_1$ is
consistent with \eqref{eq:sharp}, which requires a sensitivity of at
least about the arclength for a sufficiently small uniform
reconstruction error. Thus, a larger aperture requires a more sensitive decoder, not more
latent coordinates, and the field geometry predicts by how much before
training.

\begin{figure}[t]
	\centering
	\includegraphics[width=0.6\linewidth]
	{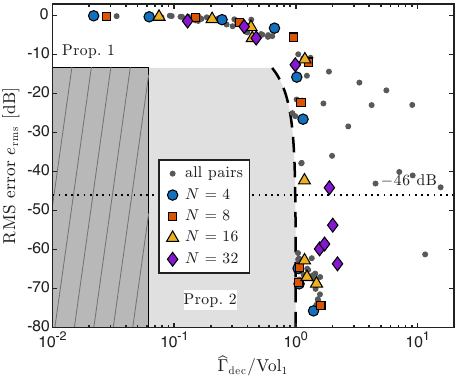}
	\caption{Decoder sensitivity on the steering arc ($L=d=1$): RMS error
	versus $\widehat\Gamma_{\mathrm{dec}}/\mathrm{Vol}_1$ for all $116$
	trained pairs (dots) and, for each aperture and decoder limit, the
	most accurate pair (markers); qualified pairs have
	$e_{\mathrm{rms}}\le-46$~dB (Table~\ref{tab:prop1}). Shaded: bound of
	Proposition~\ref{prop:sharp}, for errors at which it holds at all
	apertures; hatched: bound of Proposition~\ref{prop:lipschitz}, valid
	in any dimension. No pair falls in either region (see Sec.~\ref{subsec:observed}), and qualified
	pairs appear only at $\widehat\Gamma_{\mathrm{dec}}\gtrsim
	\mathrm{Vol}_1$.}
	\label{fig:prop1}
\end{figure}

\subsection{Polarization and the Volume Exponent}
\label{sec:communication_theory}

\begin{table}[t]
	\centering
	\caption{Gain and shape exponents, Eq.~\eqref{eq:nu-split}, of the
			benchmark families: absolute exponent $\nu$ and gain exponent
			$\nu_{\mathrm g}$ under the constraint used in the case studies,
			shape exponent $\hat\nu$, and $\nu$ at fixed radiated (array) or
			scattered (cylinder) power, $\nu|_P$. Computed values are given in
			Secs.~\ref{sec:steering_case_study} and~\ref{subsec:entropy_results}.}
	\label{tab:normalization}
	\small
	\setlength{\tabcolsep}{9pt}
	\begin{tabular}{@{}llcccc@{}}
		\toprule
		Family & Constraint & $\nu$ & $\nu_{\mathrm g}$ & $\hat\nu$ & $\nu|_P$ \\
		\midrule
		Steering $\psi$ & unit excitation & $3/2$ & $1/2$ & $1$ & $1$ \\
		Polarization $\chi$ & unit excitation & $1/2$ & $1/2$ & $0$ & $0$ \\
		Area $(\psi,\chi)$, $d=2$ & unit excitation & $2$ & $1/2$ & $1$ & $1$ \\
		PEC circ.\ cylinder & fixed incidence & $3/2$ & $1/2$ & $1$ & $1$ \\
		Born circ.\ cylinder & fixed incidence & $5/2$ & $3/2$ & $1$ & $1$ \\
		\bottomrule
	\end{tabular}
\end{table}

A change of $\chi$ produces a cross-polarized field change,
	orthogonal to the change produced by $\psi$, so that the pullback
	metric of \eqref{eq:dp-field-main} is diagonal. Unlike $\psi$, $\chi$
	acts equally on all elements: the field change it produces has the
	norm of the field itself, $\|F\|\propto N^{1/2}$, so that
	$\ell_\chi\propto(\beta L_{\mathrm{ap}})^{1/2}$ and $\nu_\chi=1/2$.
	Since the metric is diagonal and, for $g=1$, $\|F\|$ does not depend on
	$\psi$, the area of the field manifold is
	$\mathrm{Vol}_2(\mathcal E)=\ell_\psi\ell_\chi\propto(\beta
		L_{\mathrm{ap}})^2$, with $\nu=\nu_\psi+\nu_\chi=2$. The two controls
	give the same $d=1$ each, but different exponents: $\nu$ depends on how
	each control acts on the field.

At fixed radiated power, $\|F\|$ no longer grows, and only the shape
	exponents remain: $\nu_\chi=0$ and $\nu=1$
	(Table~\ref{tab:normalization}). In both cases $\ell_\psi/\ell_\chi
	\propto N$: in relative terms, the polarization length is fixed by
	the rotation angle, whereas the steering length grows as the number
	of resolvable beam positions in the scan range. Two steering
	directions on the same aperture, instead, become indistinguishable
	within about one beamwidth: $d_{\mathrm{eff}}$ falls from $2$ to $1$,
	and the Cram\'er--Rao bound on their separation diverges, the geometric
	counterpart of the resolution limit of MUSIC~\cite{Schmidt1986MUSIC}
	and ESPRIT~\cite{roy1989esprit}.

\section{Case Study 3: Scattering Regimes, Folding, and Latent Dimension in Dielectric Cylinders}
\label{sec:one_dimensional_manifolds}

The third case study applies the theory to scattering by a dielectric
cylinder whose control is the illumination angle. Unlike the arrays of
the previous case studies, whose exponents are fixed by the excitation,
the cylinder changes its scattering regime as the electrical size
grows \cite{vandehulst1981light}.

The case study has three objectives. The first, and main, one is to
show how the quantities of the theory, i.e., the intrinsic dimension
$d$, the gain exponent $\nu_{\mathrm g}$, and the shape exponent
$\hat\nu$, respond to this change. The second is to show that the
number of latent coordinates of a learned representation is set by the
intrinsic dimension and the topology of the field family. The third is
to show an effect that the pullback metric, which compares only nearby
illumination angles, cannot see: a symmetry of the cross-section can
fold the family onto itself.

The circular cylinder is the reference geometry of EIT
\cite{bucci1989degrees,migliore2026unification}, since its scattered
field has a closed-form modal series; this makes it the natural
benchmark for the first objective. Its rotational symmetry, however,
makes the field family too regular for the other two. For these, we
use an elliptic cross-section, which keeps only a mirror symmetry.

In both cases, a $TM_z$-polarized plane wave illuminates the cylinder
from the angle $\phi\in S^1$ (Fig.~\ref{fig:figincidencelegend}), so
that $d=p=1$, whereas the NDF grows with $\beta a$; each electrical size
$s=\beta a$ gives the closed curve of fields
\[
\mathcal E_{s}
=
\left\{
F(\cdot;\phi,\beta a):\phi\in S^1
\right\},
\]
each point of which is a whole scattered pattern, as a function of the
observation angle $\vartheta$. Since the curve is closed, its minimum
embedding dimension is $L_{\min}=2$ (Sec.~\ref{subsec:p-vs-d}). All
exponents refer to a fixed incident power density.

\subsection{Circular Cylinder: Exponents across Scattering Regimes}
\label{subsec:entropy_results}

\begin{figure}[t]
	\centering
	\includegraphics[width=0.6\linewidth]{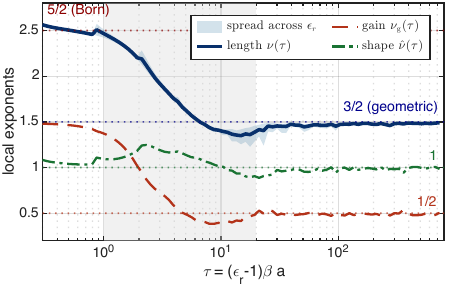}
	\caption{Local exponents of the circular dielectric cylinder, from its
		exact modal coefficients, versus the phase delay
		$\tau=(\eps_r-1)\beta a$, for $\eps_r\in[1.05,4]$ and
		$4\le\beta a\le250$: length exponent $\nu$ (solid), with its spread
		across contrasts (light band), and its gain and shape parts
		$\nu_{\mathrm g}$ (dashed) and $\hat\nu$ (dash-dotted); grey
		background: transition region. The exponents are slopes of log-log
		fits over a factor of ten in $\beta a$, which average out the
		internal resonances, and the curves are means over the contrasts. As
		$\nu$ passes from the Born value $5/2$ to the geometric value $3/2$,
		$\nu_{\mathrm g}$ falls from $3/2$ to $1/2$, whereas $\hat\nu$ stays
		close to $1$.}
	\label{fig:exponents}
\end{figure}

For the circular cylinder of radius $a$ and relative permittivity
$\eps_r$, the scattered field is the modal series
\cite[Sec.~15.1]{vandehulst1981light}, evaluated in closed form for
$\eps_r\in[1.05,4]$. Since the cylinder is invariant under rotation,
the scattered pattern depends only on the difference between the
observation and incidence angles,
$F(\vartheta;\phi)=\sum_n c_n e^{\mathrm jn(\vartheta-\phi)}$, with $c_n$
the modal coefficients: a change of $\phi$ only rotates the pattern,
and a derivative with respect to $\phi$ multiplies the $n$th harmonic
by $-\mathrm jn$. With $M_k=2\pi\sum_n n^k|c_n|^2$, Parseval's identity
gives the field norm, the speed, and the acceleration along the curve
$\mathcal E_s$,
\[
\|\Phi\|^2=M_0,\qquad \|\partial_\phi\Phi\|^2=M_2,\qquad
\|\partial_\phi^2\Phi\|^2=M_4,
\]
all independent of $\phi$. The field thus moves along the curve at
constant speed $\sqrt{M_2}$, so that the length of the curve, the speed
times the full turn $2\pi$, is $\ell_{\mathcal E}=2\pi\sqrt{M_2}$, and
its curvature radius, the squared speed divided by the acceleration, is
$M_2/\sqrt{M_4}$ (Appendix~\ref{app:derivation}). Relative to the field
norm, the curvature radius stays between $0.58$ and $0.96$ for
$4\le\beta a\le250$: the curve becomes longer, but not tighter, i.e.,
the radius of its bends does not shrink, so that a small ball centered
on the curve still contains a single, nearly straight piece of it, as
required by \eqref{eq:local-regularity}.

The moments $M_k$ have a simple physical meaning. The
harmonic $n$ of the pattern has $n$ lobes around the cylinder, and
$|c_n|^2$ is the power scattered into it, so that $M_0$ is the total
scattered power. When the illumination rotates, the harmonic $n$
changes $n$ times faster than the angle, so that $M_2$, the power
weighted by $n^2$, measures how fast the pattern changes. The length of
the curve thus splits into two factors,
$\ell_{\mathcal E}=2\pi\sqrt{M_2}=\sqrt{M_0}\cdot2\pi\sqrt{M_2/M_0}$:
the field norm $\sqrt{M_0}$, which gives the gain exponent, and the
normalized length $2\pi\sqrt{M_2/M_0}$, which gives the shape exponent.
The ratio $M_2/M_0$ is the mean of $n^2$ weighted by the power of each
harmonic: it measures how fine the pattern is, since a rotation of
about $\sqrt{M_0/M_2}$ makes the pattern distinguishable from the
original one.

The scattering regime, set by the phase lag of the wave across the
cylinder, which, for weak contrast, is $\tau=(\eps_r-1)\beta a$
\cite[Secs.~11.1 and~15.32]{vandehulst1981light}, determines how much
power is scattered. When $\tau\ll1$, the cylinder is optically thin and
scatters in the Born regime: its extinction efficiency, i.e., the
extinction width normalized to the geometric width $2a$, grows as
$\tau^2$, and $M_0$ as $(\beta a)^3$, so that $\nu_{\mathrm g}=3/2$.
When $\tau\gg1$, it scatters in the geometric-optics regime, as a
perfect conductor: the efficiency tends to $2$ after an overshoot
\cite[Fig.~67]{vandehulst1981light}, and $M_0$ grows as $\beta a$, so
that $\nu_{\mathrm g}=1/2$ (Appendix~\ref{app:derivation}). 
The size alone, instead, sets the band of harmonics that carry this
power: a cylinder of
radius $a$ radiates significant harmonics only up to
$|n|\approx\beta a$ \cite{bucci1989degrees}, the same band that sets
the NDF, so that its pattern cannot have angular details finer than
about $1/(\beta a)$, whatever the material. As long as the power is
spread over a fixed fraction of this band, $M_2/M_0$ grows as
$(\beta a)^2$ in both regimes; hence $\hat\nu=1$, and
$\nu=\hat\nu+\nu_{\mathrm g}$ is $5/2$ and $3/2$. The regime thus
changes how much power is scattered, but not how fine the pattern is,
and hence not the number of distinguishable fields at a fixed SNR.

The exact modal series confirms these predictions
(Fig.~\ref{fig:exponents}; Appendix~\ref{app:derivation}): the change of
$\nu$ with $\tau$ is almost entirely in the gain exponent, whereas the
shape exponent stays close to one, with deviations of up to about $20\%$
in the transition and at the smallest sizes, where only a few harmonics
are significant; a fit over $4\le\beta a\le250$ gives $\hat\nu=1.00$
for $\eps_r=1.5$, $0.98$ for $\eps_r=4$, and $0.98$ for the PEC. 

%The
%undershoot of $\nu$ near $\tau\approx10$ is also in the gain exponent:
%after its overshoot, the extinction efficiency decreases toward its
%geometric limit, so that the scattered power grows more slowly than
%$\beta a$.

\subsection{Elliptic Cylinder: Field Family}
\label{subsec:scattering_setup}

The elliptic cylinder has major semi-axis $a$, minor semi-axis
$b=0.3a$, and $\eps_r=1.5$. Its fields are computed by the method of moments (MoM) \cite{richmond1965scattering} at
eight sizes, $\beta a=1.57,\dots,12.57$
($a=0.25\lambda,\dots,2\lambda$, i.e., $\tau=0.8$ to $6.3$), sampled at
$N_{\mathrm{obs}}=128$ points on a circle of radius $8\lambda$; the same
samples, above the Nyquist rate of the largest cylinder
\cite{bucci1989degrees}, are used at every size (Appendix~
\ref{supp:normalization}). For this cross-section, the cylindrical-wave T-matrix couples the incident harmonic $m$ to the
scattered harmonic $n$ and is no longer diagonal, so that the speed
varies with $\phi$. Its volume exponent, which requires only the length
of the curve, is also computed on an extended sweep up to
$\beta a=31.4$ ($\tau=15.7$): the local exponent $\nu$ decreases from
$2.71$ to $1.52$ as $\tau$ grows from $0.8$ to $15.7$, the same
transition as for the circular cylinder. The regularity condition
\eqref{eq:local-regularity} holds at every size at the $-60$~dB
accuracy of the trained networks (Table~\ref{tab:regularity}). Before
training, the theory thus predicts an entropy law with slope $d=1$,
whose intercept grows with the normalized length of the curve; it
cannot predict, from local quantities, whether the curve folds back on
itself, which is examined directly on the computed fields
(Sec.~\ref{subsec:diffusion_maps}).

\subsection{Elliptic Cylinder: Entropy Law and Folding}
\label{subsec:diffusion_maps}

\begin{figure}[t]
	\centering
	\includegraphics[width=0.6\linewidth]{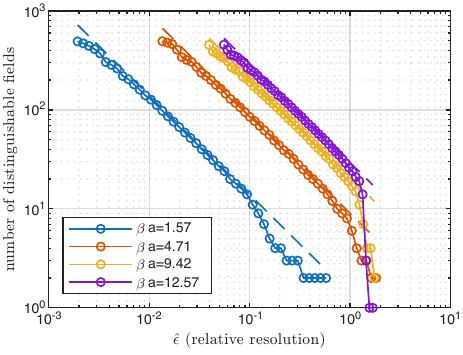}
\caption{Number of distinguishable fields of the elliptic cylinder
	(MoM, $\eps_r=1.5$), counted as maximal sets of fields at least
	$\hat\eps$ apart, whose logarithm is the $\eps$-entropy up to a
	bounded offset, versus the relative resolution $\hat\eps$ (relative to
	the RMS field norm \eqref{eq:rms-norm} at each $\beta a$), for several
	$\beta a$. Dashed lines: $\ell/\hat\eps$, with $\ell$ the normalized
	length of the curve computed independently, i.e., the count predicted
	by the metric; the parallel slopes confirm $d=1$, and the vertical
	offsets follow the growth of the normalized length. At $\beta a=1.57$,
	the count falls below the dashed line at coarse resolution, where
	mirror fields cannot be distinguished and the metric overestimates the
	number of distinguishable fields.}
	\label{fig:entropy}
\end{figure}

The entropy law can be tested directly on the computed fields, by
counting the distinguishable fields as maximal sets of fields
at least $\hat\eps$ apart (Fig.~\ref{fig:entropy}). As predicted
by \eqref{eq:intro-entropy}, the curves for different sizes are
parallel, with a common slope $d=1$, and are shifted vertically by the
growth of the normalized length $\ell_{\mathcal E_s}/\|F_s\|$. Since
the resolution is relative, the comparison is made at the same SNR: by
\eqref{eq:entropy-snr}, only the change in the shape of the patterns
shifts the curves.

Local quantities, such as the speed and the curvature of the curve,
however, do not show whether distant parts of the curve come close to
each other. This information
must be extracted from the computed fields themselves, with tools of
manifold learning, which analyze the geometry of a set of points in a
high-dimensional space without training a model (details in
Appendix~\ref{app:training}). The two-nearest-neighbor estimator
\cite{facco2017} measures the intrinsic dimension from the ratio of the
distances of each field to its two closest neighbors, whose statistics
depends only on $d$; the estimate is between $0.98$ and $1.01$ at all
eight sizes: as predicted, $d=1$ across the transition between the two
regimes. Diffusion maps \cite{coifman2006diffusion} draw the curve, which lies
in a space with as many dimensions as field samples, in two or three
coordinates; they show whether
the curve is closed, folded onto itself, or close to itself in some
region (Fig.~\ref{fig:dm}).

\begin{figure*}[!t]
	\centering
	\includegraphics[width=0.8\textwidth]{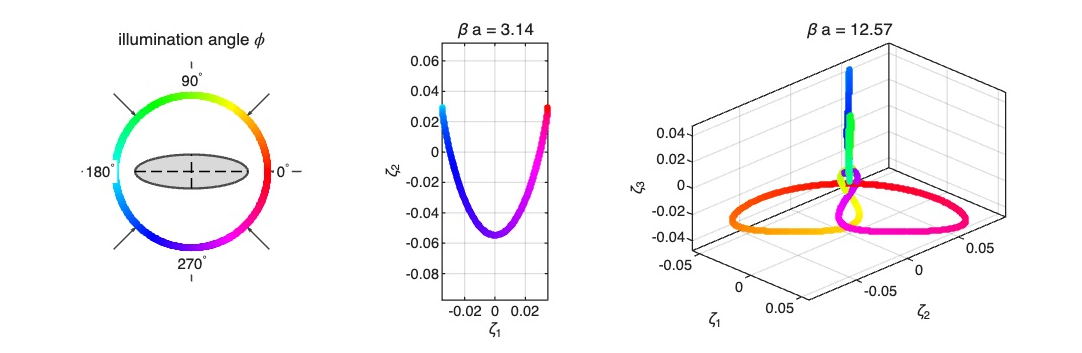}
	\caption{Left: colors associated with the illumination angle $\phi$,
		with the elliptic cross-section drawn to scale; the same colors are
		used in the diffusion maps and in the latent plots of
		Fig.~\ref{fig:orbits}. Center and right: diffusion-map pictures of the
		curve of scattered fields of the elliptic cylinder. At
		$\beta a=3.14$ (center), the closed curve is folded onto itself by the
		mirror symmetry of the ellipse and looks like an open arc. At
		$\beta a=12.57$ (right), the two mirror halves are separated; the
		curve crosses itself in the first two coordinates, but not in field
		space, as the third coordinate $\zeta_3$ shows. Gaussian kernel with
		bandwidth equal to the median pairwise distance, density
		normalization $\alpha=1$ (Appendix~\ref{app:training}).}
	\label{fig:dm}\label{fig:figincidencelegend}
\end{figure*}

Since the ellipse is mirror-symmetric, the illumination angles $\phi$
and $2\pi-\phi$ give mirror-image patterns,
$F(\vartheta;2\pi-\phi)=F(-\vartheta;\phi)$.  For a small cylinder, the
pattern is nearly even in $\vartheta$, so that
$\Phi(\phi)\approx\Phi(2\pi-\phi)$: the two halves of the closed curve
should nearly coincide. As $\beta a$ grows, the scattered field
concentrates in the forward direction
\cite[Sec.~11.1]{vandehulst1981light}, the pattern loses this symmetry,
and the two halves should separate. The diffusion maps confirm this
picture: at $\beta a=3.14$, the curve looks like an open arc
(Fig.~\ref{fig:dm}, center), and at $\beta a=12.57$ its two halves are
apart (Fig.~\ref{fig:dm}, right). The separation, computed directly from
the fields as
$\Delta_{\mathrm{br}}\coloneqq\operatorname{RMS}_\phi\|\Phi(\phi)-\Phi(2\pi-\phi)\|/(\sqrt2\,\|F\|)$,
which is $0$ for coincident halves and close to $1$ for uncorrelated
ones, grows from $0.067$ at $\beta a=1.57$ to $0.97$ at
$\beta a=12.57$ (Table~\ref{tab:metrics}).

The folding is invisible to the pullback metric, which measures only
the length of the curve: the number of distinguishable fields it
predicts, i.e., the Fisher count \eqref{eq:capacity-fisher}, is
$\ell/\hat\eps$ up to a constant (dashed lines in
Fig.~\ref{fig:entropy}). The direct count, instead, compares all pairs
of fields, so that, when the resolution is coarser than the distance
between mirror fields, about $\sqrt2\,\Delta_{\mathrm{br}}$ in relative
terms, two mirror fields count as one, as shown in
Fig.~\ref{fig:kolmogorov-fisher}, and the Fisher count overestimates
the number of distinguishable fields (Sec.~\ref{sec:fisher_bridge}).
In Fig.~\ref{fig:entropy}, this appears at $\beta a=1.57$, where
$\Delta_{\mathrm{br}}=0.067$: for $\hat\eps\gtrsim0.1$, the direct count
falls below the dashed line. The folding, however, changes neither the
dimension nor the topology: the curve remains closed, with $L_{\min}=2$
(Sec.~\ref{subsec:p-vs-d}). For a learned representation, one latent
coordinate therefore cannot represent the curve uniformly and two
suffice in principle, but they are harder to obtain at small sizes,
where the two halves nearly coincide, and at large sizes, where the
curve is long and its two-dimensional picture crosses itself.

\subsection{Elliptic Cylinder: Learned Representations}
\label{subsec:heatmap_results}

The autoencoder of Sec.~\ref{sec:stability_validation} is trained on
the fields of the elliptic cylinder, three times per size, with the SVD
layer \eqref{eq:svd_rec} truncated at $-60$~dB (Appendix~
\ref{app:training}); as in Sec.~\ref{sec:array_case_study}, a latent dimension
is successful when the mean reconstruction error reaches $-60$~dB.

As predicted in Sec.~\ref{subsec:diffusion_maps}, one coordinate is not
enough: the mean error stays between
$-36$ and $-15$~dB at all sizes. By the antipodal bound \eqref{eq:antipodal}, the worst-case error is even larger, at
least $-10.6$ to $-3.9$~dB for $\beta a\ge3.14$. Two coordinates suffice where the curve is easy to unfold: the
$-60$~dB target is met for $3.14\le\beta a\le9.42$. They fail at the
two ends of the range, as expected: at $\beta a=1.57$, where the two
halves of the curve are only a few percent apart, and for
$\beta a\ge11$, where the learned latent curve crosses itself, as the
diffusion map does. With three coordinates, the target is met at every
size, and the learned curves are smooth and do not cross themselves
(Fig.~\ref{fig:orbits}). Near these transitions, the outcome varies
between runs, so Table~\ref{tab:metrics} reports both the best run and
the mean over the runs.

\begin{figure}[t]
	\centering
	\includegraphics[width=1.0\linewidth]{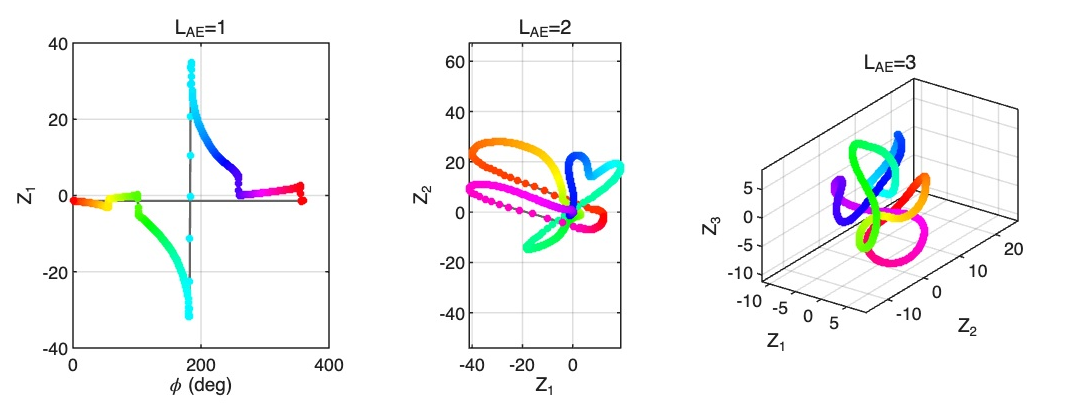}
	\caption{Learned latent trajectories of the elliptic cylinder at
	$\beta a=12.57$ and $\eps_r=1.5$, colored by illumination angle as in
	Fig.~\ref{fig:dm}. For $\LAE=1$, distinct points of the closed curve
	share a latent coordinate; for $\LAE=2$, the curve closes but
	self-intersects; for $\LAE=3$, it becomes smooth and
	non-self-intersecting. For the run shown, the mean relative test
	errors are $-22.3$, $-26.7$, and $-60.2$~dB, so that only $\LAE=3$
	meets the $-60$~dB target.}
	\label{fig:orbits}
\end{figure}

\begin{table}[t]
	\centering
	\caption{Linear NDF, estimated as $2P/\lambda$ with $P$ the perimeter
		of the cross-section \cite{bucci1989degrees}, intrinsic dimension,
		branch separation (Sec.~\ref{subsec:diffusion_maps}), and
		smallest tested latent dimension $L\in\{1,\dots,4\}$ attaining
		the reconstruction target $\bar e\le-60$~dB for the elliptic cylinder
		($\eps_r=1.5$), in at least one of three training runs (best run) and
		for the error averaged in dB over the three runs (mean).}
	\label{tab:metrics}
	\small
	\setlength{\tabcolsep}{9pt}
	\begin{tabular}{ccccccc}
		\toprule
		&&&&& \multicolumn{2}{c}{$\LAE$} \\
		\cmidrule(l){6-7}
		$a/\lambda$ & $\beta a$ & Linear $\NDF$ & $d$ & $\Delta_{\mathrm{br}}$ & best run & mean \\
		\midrule
		0.25 & 1.57  & 2  & 1 & 0.067 & 3 & 3 \\
		0.50 & 3.14  & 4  & 1 & 0.28 & 2 & 2 \\
		0.75 & 4.71  & 7  & 1 & 0.55 & 2 & 2 \\
		1.00 & 6.28  & 9  & 1 & 0.74 & 2 & 2 \\
		1.25 & 7.85  & 11 & 1 & 0.87 & 2 & 3 \\
		1.50 & 9.42  & 13 & 1 & 0.93 & 2 & 3 \\
		1.75 & 11.00 & 15 & 1 & 0.95 & 3 & 3 \\
		2.00 & 12.57 & 18 & 1 & 0.97 & 3 & 4 \\
		\bottomrule
	\end{tabular}
\end{table}

Table~\ref{tab:metrics} summarizes the result: as $\beta a$ grows from
$1.57$ to $12.57$, across the transition between the two scattering
regimes, the linear NDF grows from $2$ to $18$, i.e., up to $36$ real
coordinates, whereas $d=1$ and the trained networks need only $2$--$3$
latent coordinates in the best run ($2$--$4$ on average), close to
$L_{\min}=2$. The number of latent coordinates thus follows neither the
NDF nor the scattering regime, but the dimension and the topology of
the field family; its excess over $L_{\min}$ appears only where the
curve folds or becomes long. The present bounds cannot establish
whether this excess is intrinsic to the field family, i.e.,
$L_{\mathrm{stab}}>L_{\min}$, or due to training.

\section{Conclusion}
\label{sec:conclusion}

This paper has extended the EIT framework to electromagnetic
systems driven by a few parametric controls, whose fields form
a manifold described by its intrinsic dimension, i.e., the number
of independent field variations, and by its metric volume, which
governs the number of distinguishable fields under suitable
regularity conditions.
At an early design stage, the theory predicts how electrical size
affects the number of distinguishable states for communication
and, through the Fisher information, the attainable estimation
accuracy in sensing.

The same geometry provides benchmarks for neural representations
of the fields, available before training. The topology of the
manifold sets the minimum number of real latent coordinates
of a continuous one-to-one representation; accuracy, decoder
stability, and training determine how many coordinates an
autoencoder actually uses.
Under suitable regularity conditions, the metric volume sets
a lower bound on decoder sensitivity that grows with the volume
for a fixed latent dimension and reconstruction accuracy.
These benchmarks can guide the design of neural surrogates and
digital twins of radiating and scattering systems, suggesting
a new role for electromagnetic theory: guiding the design of
machine-learning models of electromagnetic systems, rather than
merely supplying their training data.

\appendices
% Lemmas of the appendices are numbered from 1
\setcounter{theorem}{0}

\section*{Guide to the Appendices}
\label{supp:guide}

The appendices are organized in two parts. Part~I gives the definition of the resolution and the proofs of Propositions~1 and~2, with the geometric conditions under which they hold uniformly in electrical size. Part~II gives the analyses that support the case studies: normalization, control parameters and coupling, the cylinder exponents, the guaranteed bounds on the trained decoders, and the training details.

\begin{center}
	\small
	\setlength{\tabcolsep}{6pt}
	\begin{tabularx}{\textwidth}{@{}l L l@{}}
		\toprule
		App. & \leavevmode Content & Section\\
		\midrule
		\multicolumn{3}{@{}l}{\emph{Part I: Theory and Proofs}}\\[1pt]
		\ref{supp:notation} & \leavevmode Resolution, noise, and error quantities & \ref{subsec:metric_entropy_definition}\\
		\ref{app:minimum_constant} & \leavevmode Proofs of Proposition~1 and Corollary~1 & \ref{subsec:stable_widths_lipschitz}\\
		\ref{app:regularity} & \leavevmode Uniform entropy estimate: Lemmas~\ref{lem:regularity} and~\ref{lem:curves} & \ref{subsec:stable_widths_lipschitz}, \ref{sec:steering_case_study}, \ref{sec:one_dimensional_manifolds}\\
		\ref{app:sharp} & \leavevmode Proof of Proposition~2 & \ref{subsec:sharp-arc}\\
		\multicolumn{3}{@{}l}{\emph{Part II: Supporting Analyses for the Case Studies}}\\[1pt]
		\ref{supp:normalization} & \leavevmode Normalization and observation conditions & \ref{subsec:normalization}\\
		\ref{supp:pvsd} & \leavevmode Control parameters, intrinsic dimension, and mutual coupling & \ref{subsec:p-vs-d}, \ref{sec:array_case_study}\\
		\ref{app:derivation} & \leavevmode Arclength and exponents of circular cylinders & \ref{subsec:entropy_results}\\
		\ref{supp:certified} & \leavevmode Guaranteed bounds on the decoder sensitivity & \ref{subsec:observed}\\
		\ref{app:training} & \leavevmode Network, data, training, and electromagnetic-model details & \ref{sec:stability_validation}--\ref{sec:one_dimensional_manifolds}\\
		\bottomrule
	\end{tabularx}
\end{center}

Appendix~\ref{supp:symbols} lists the symbols and acronyms.

\supppart{I}{Theory and Proofs}

\label{supp:novelty}

Most proofs in Part~I adapt established arguments to parametric
families of electromagnetic fields. The covering argument for
Lipschitz images follows Proposition~3.6 of
\cite{petrova2023lipschitz} (see also \cite{cohen2022optimal}). The
local volume estimates use the mass distribution principle
\cite{falconer2014fractal} and classical results on curvature and reach
\cite{schur1921schwarzsche,federer1959curvature,niyogi2008finding,aamari2019estimating}.
The proof of Proposition~2 (Appendix~\ref{app:sharp}) also uses the
Lipschitz bound on the nearest-point projection
\cite[Th.~4.8(8)]{federer1959curvature}. The relation between the field
metric and the Fisher information used in Sec.~\ref{sec:fisher_bridge}
is discussed in
\cite{slepian1954estimation,kay1993fundamentals,rissanen1996fisher,balasubramanian1997}.
The source of each argument, and the author's own contributions where
present, are stated at the beginning of each appendix containing a proof.

\section{Resolution, Noise, and Error Quantities}
\label{supp:notation}

\emph{Resolution.} Throughout the paper, $\eps$ is an absolute,
worst-case distance between fields: a covering of radius $\eps$ must
represent every field of $\mathcal E$ within $\eps$, and all bounds of the paper refer to this error. 
The resolution relative to the root-mean-square field norm \eqref{eq:rms-norm} is
$\hat\eps=\eps/\|F\|$, in decibels $20\log_{10}\hat\eps$.

\emph{Noise and resolution.} For the observation model
$\mathbf y=\Phi(\theta)+\mathbf n$,
$\mathbf n\sim\mathcal{CN}(\mathbf 0,\sigma^2\mathbf I)$, the pairwise
error probability of a minimum-distance decision rule between two field
states at distance $D$ is $Q\bigl(D/(\sqrt2\,\sigma)\bigr)$
\cite{vantrees2002optimum}. For the centers of a $2\eps$-packing,
$D\ge2\eps$; setting $\eps=c_\sigma\sigma$ with
$Q(\sqrt2\,c_\sigma)=10^{-3}$ gives $c_\sigma\approx2.2$, per pair of
neighboring states; for $d=1$, each state has two neighbors, so that
the error probability of each decision is at most $2\times10^{-3}$. The balls of
radius $\eps$ thus lie inside the decision regions of this rule (for the relation between deterministic and
stochastic EIT, see \cite{migliore2026unification}). With
$\mathrm{SNR}=\|F\|^2/\sigma^2$, $\hat\eps=c_\sigma\,\mathrm{SNR}^{-1/2}$,
and since $20\log_{10}c_\sigma\approx7$~dB, a relative resolution of
$-X$~dB requires an SNR of about $X+7$~dB, e.g., $67$~dB for the
$-60$~dB target of Secs.~\ref{sec:array_case_study} and~\ref{sec:one_dimensional_manifolds}.

\emph{Errors of trained networks.} The errors reported for trained
networks are averages over test samples not used in training: the mean
relative error
$\bar e=\langle\|\mathbf x-\widehat{\mathbf x}\|/\|\mathbf x\|\rangle$
 and the relative RMS error
$e_{\mathrm{rms}}=\|\widehat{\mathbf X}-\mathbf X\|_{\mathrm F}/
\|\mathbf X\|_{\mathrm F}$. Both are bounded by the maximum
relative error $\sup_\theta\|\Phi(\theta)-\widehat\Phi(\theta)\|/
\|\Phi(\theta)\|$. Worst-case errors are obtained only from
the guaranteed bounds of Appendix~\ref{supp:certified}.

\section{Proofs of Proposition~1 and Corollary~1}
\label{app:minimum_constant}

Step~2 below is Proposition~3.6 of \cite{petrova2023lipschitz}, based
on the covering argument for Lipschitz images used in
\cite{cohen2022optimal}, restated for a latent image of unit diameter;
since that result requires no encoder, it applies to encoder--decoder
pairs. To the author's knowledge, the new elements are its combination
with the entropy estimate \eqref{eq:uniform-entropy}, which brings in
the metric volume, and the conditions of Appendix~\ref{app:regularity}
for uniformity in electrical size.

\smallskip

\emph{Assumption.} The field families $\mathcal E_s$ have intrinsic
dimension $d\ge1$ and satisfy the entropy bound \eqref{eq:uniform-entropy-main},
\begin{equation}
	H_{2\varepsilon}(\mathcal E_s)
	\ge
	d\log_2\frac{1}{\varepsilon}
	+\log_2\mathrm{Vol}_d(\mathcal E_s)-C_0,
	\qquad 0<\varepsilon<\varepsilon_{\max},
	\label{eq:uniform-entropy}
\end{equation}
with $\mathrm{Vol}_d(\mathcal E_s)\ge V_{\min}>0$. The constants
$C_0$, $\varepsilon_{\max}$, and $V_{\min}$ are independent of the
electrical size $s$; this is the only point where uniformity in $s$ is
needed, and Appendix~\ref{app:regularity} shows when it holds.

\begin{figure*}[t]
	\centering
	\begin{tikzpicture}[>=stealth,
		every node/.style={font=\scriptsize},
		ball/.style={draw=black!45,fill=black!8},
		ttl/.style={font=\small}]
		\foreach \c in {0.6,1.2,...,3.61}{\draw[ball] (\c,0) circle (0.3);}
		\draw[very thick] (0.3,0) -- (3.9,0);
		\foreach \c in {0.6,1.2,...,3.61}{\fill (\c,0) circle (1.1pt);}
		\node[below=1pt,font=\tiny] at (2.4,0) {$z_i$};
		\draw[decorate,decoration={brace,amplitude=4pt}] (0.3,0.42) -- (3.9,0.42)
		node[midway,above=4pt] {latent image $\Psi_{\mathrm{enc}}(\mathcal E)$, diameter $1$};
		\draw[<->] (3.6,-0.4) -- (3.9,-0.4);
		\node[below] at (3.75,-0.42) {$\varepsilon/\Gamma_{\mathrm{dec}}$};
		\node at (1.7,-0.75) {about $\Gamma_{\mathrm{dec}}/\varepsilon$ balls};
		\node[ttl] at (2.1,-1.7) {(a) Latent space ($L=1$)};
		\draw[->,thick,red!70!black] (4.25,0.25) to[bend left=25]
		node[above,text=red!70!black] {$\Phi_{\mathrm{dec}}$} (5.35,0.25);
		\node[align=center,text=red!70!black] at (4.85,-0.75) {stretches\\distances\\by $\le\Gamma_{\mathrm{dec}}$};
		\foreach \t in {0.35,1.466,2.582,3.698,4.814,5.93}{
			\draw[ball] ({5.6+0.95*\t},{0.12+0.72*sin(deg(\t))}) circle (0.66);}
		\draw[very thick,domain=0:6.283,samples=120] plot ({5.6+0.95*\x},{0.8*sin(deg(\x))});
		\draw[thick,dashed,red!70!black,domain=0:6.283,samples=120] plot ({5.6+0.95*\x},{0.12+0.72*sin(deg(\x))});
		\foreach \t in {0.35,1.466,2.582,3.698,4.814,5.93}{
			\fill[red!70!black] ({5.6+0.95*\t},{0.12+0.72*sin(deg(\t))}) +(-1.3pt,-1.3pt) rectangle +(1.3pt,1.3pt);}
		\node[right] at (11.55,0.0) {$\mathcal E$};
		\node[text=red!70!black] at (6.9,1.62) {decoded curve};
		\draw[->] ({5.6+0.95*5.93},{0.12+0.72*sin(deg(5.93))}) -- +(60:0.66);
		\node at (11.65,0.85) {$2\varepsilon$};
		\draw[densely dotted] ({5.6+0.95*3.698-0.78},{0.12+0.72*sin(deg(3.698))-0.78}) rectangle
		({5.6+0.95*3.698+0.78},{0.12+0.72*sin(deg(3.698))+0.78});
		\node[ttl] at (8.6,-1.7) {(b) Field space $\mathcal H_E$};
		\begin{scope}[shift={(14.9,0.05)}]
			\draw[ball] (0,0) circle (1.4);
			\draw[very thick] plot[smooth] coordinates {(-2.0,1.62) (-0.617,1.184) (0.6,0.95) (1.5,0.62)};
			\draw[thick,dashed,red!70!black] plot[smooth] coordinates
			{(-2.0,1.2) (-1.2,0.95) (-0.495,0.495) (0,0) (0.6,-0.5) (1.2,-0.8)};
			\coordinate (f) at (-0.617,1.184);
			\coordinate (fh) at (-0.495,0.495);
			\coordinate (c) at (0,0);
						\draw[thick] (f) -- (fh) node[midway,right=1pt] {$\le\varepsilon$};
			\draw[thick,red!70!black] (fh) -- (c) node[midway,above right=0pt,text=red!70!black] {$\le\varepsilon$};
			\fill (f) circle (1.5pt) node[above=1pt] {$f$};
			\fill[red!70!black] (fh) circle (1.5pt) node[below left=0pt,text=red!70!black] {$\widehat f$};
			\fill[red!70!black] (c) +(-1.6pt,-1.6pt) rectangle +(1.6pt,1.6pt);
			\node[right=3pt,text=red!70!black] at (c) {$\Phi_{\mathrm{dec}}(z_i)$};
			\draw[->] (c) -- (-80:1.4) node[midway,right=1pt] {$2\varepsilon$};
			\node[right] at (1.5,0.62) {$\mathcal E$};
			\node[ttl] at (0,-1.75) {(c) Why radius $2\varepsilon$};
		\end{scope}
	\end{tikzpicture}
\caption{Covering argument behind Proposition~1, drawn for $L=d=1$.
	Intuitively, the problem is to cover a string with coins: the field
	family $\mathcal E$ is the string, and the decoder can supply only a
	limited number of coins to cover it. (a)~The latent image, of unit
	diameter, is cut into small intervals of radius
	$\varepsilon/\Gamma_{\mathrm{dec}}$; at most
	$1+\Gamma_{\mathrm{dec}}/\varepsilon$ of them are needed. The more
	sensitive the decoder, the more intervals, i.e., the more coins.
	(b)~Since the decoder stretches distances by at most
	$\Gamma_{\mathrm{dec}}$, each small interval is mapped into a coin of
	radius $\varepsilon$ in field space. The center of a coin is the
	decoded field of a latent point, i.e., the reconstruction of some field
	of the family: it need not lie on the string, but it is within
	$\varepsilon$ of it. Enlarged to radius $2\varepsilon$, these coins
	cover the whole string $\mathcal E$. (c)~Why the coins must be
	enlarged: a field $f$ of the family is within $\varepsilon$ of its
	reconstruction $\widehat f$, which is within $\varepsilon$ of the
	center of its coin, so $f$ is within $2\varepsilon$ of that center.
	Conclusion: covering a string of length $\mathrm{Vol}_1(\mathcal E)$
	requires a number of coins proportional to
	$\mathrm{Vol}_1(\mathcal E)/\varepsilon$, whereas the decoder supplies a
	number proportional to $\Gamma_{\mathrm{dec}}/\varepsilon$; hence
	$\Gamma_{\mathrm{dec}}$ must grow at least as $\mathrm{Vol}_1(\mathcal E)$:
	the decoder must stretch a unit interval over a curve of length of the
	order of $\mathrm{Vol}_1$.}
	\label{fig:covering}
\end{figure*}

\emph{Idea of the proof.} Think of the field family $\mathcal E$ as a
string ($L=1$), and of the balls of radius $2\varepsilon$ as coins laid on it
(Fig.~\ref{fig:covering}). The proof counts, in two ways, the coins
needed to cover the string. On the one hand, the decoder can supply
only a limited number of coins: the latent image is cut into small
intervals [Fig.~\ref{fig:covering}(a)], each mapped by the decoder into
a coin [Fig.~\ref{fig:covering}(b)], and a decoder with sensitivity
$\Gamma_{\mathrm{dec}}$ yields about
$(\Gamma_{\mathrm{dec}}/\varepsilon)^L$ of them (Steps~1 and~2). The
coins have radius $2\varepsilon$, not $\varepsilon$, because their
centers are reconstructions and may lie up to $\varepsilon$ off the
string [Fig.~\ref{fig:covering}(c)]. On the other hand, the string
needs many coins, as many as its length, or in general its volume,
requires \eqref{eq:uniform-entropy}. A decoder that supplies fewer coins
than the string needs leaves part of it uncovered, i.e., some field is
not reconstructed within $\varepsilon$. Requiring the first number to
be at least the second gives the lower bound on $\Gamma_{\mathrm{dec}}$
(Step~3).

\emph{Step 1: normalization.} Fix $s$, write $\mathcal E=\mathcal E_s$,
and consider an encoder--decoder pair of the type allowed in \eqref{eq:decoder-stable-width}, i.e.,
a continuous encoder $\Psi_{\mathrm{enc}}:\mathcal E\to\mathbb R^L$ and
a Lipschitz decoder $\Phi_{\mathrm{dec}}:\mathbb R^L\to\mathcal H_E$,
that reconstructs every field within $\varepsilon$.
For $\varepsilon$ small enough that the right-hand side of
\eqref{eq:uniform-entropy} is positive, $\mathcal E$ does not fit in a
single ball of radius $2\varepsilon$, so the latent image has nonzero
diameter. By the rescaling invariance of \eqref{eq:normalized-sensitivity}, we set this diameter to
$1$, so that $\operatorname{Lip}(\Phi_{\mathrm{dec}})=\Gamma_{\mathrm{dec}}$.

\emph{Step 2: the decoder supplies a limited number of balls.} A set of unit diameter in
$\mathbb R^L$ is covered by at most $c_L(1+1/r)^L$ balls of radius $r$,
with $c_L=2^L$ in general and $c_L=1$ for $L=1$, since an interval of unit
length is covered by $\lceil1/(2r)\rceil\le1+1/r$ intervals of radius
$r$. With $r=\varepsilon/\Gamma_{\mathrm{dec}}$, the decoder maps each
latent ball into a ball of radius $\varepsilon$, and every field lies
within $2\varepsilon$ of one of the decoded centers
(Fig.~\ref{fig:covering}). Hence
\begin{equation}
	N_{2\varepsilon}(\mathcal E)
	\le
	c_L\Bigl(1+\frac{\Gamma_{\mathrm{dec}}}{\varepsilon}\Bigr)^L .
	\label{eq:decoder-covering-bound}
\end{equation}

\emph{Step 3: comparison with the family.} Taking the logarithm of
\eqref{eq:decoder-covering-bound} and combining it with
\eqref{eq:uniform-entropy} gives
\begin{equation}
	\Gamma_{\mathrm{dec}}+\varepsilon
	\ge
	2^{-(C_0+\log_2c_L)/L}\,
	\varepsilon^{(L-d)/L}
	\bigl[\mathrm{Vol}_d(\mathcal E)\bigr]^{1/L}.
	\label{eq:app-general-L}
\end{equation}

For $\varepsilon$ small enough, uniformly in $s$ since
$\mathrm{Vol}_d\ge V_{\min}$, the right-hand side exceeds $2\varepsilon$,
and the term $\varepsilon$ on the left can be absorbed: this is
Corollary~1, \eqref{eq:stability-tradeoff}, with $c_1(L)=\tfrac12\,2^{-(C_0+\log_2c_L)/L}$. For
$L=d$, the powers of $\varepsilon$ cancel: this is Proposition~1, \eqref{eq:migration},
with $c_1=c_1(d)$, valid for $\varepsilon\le c_1V_{\min}^{1/d}$. The case
$L=d$ presupposes that an encoding into $\mathbb R^d$ with accuracy
$\varepsilon$ exists, which is not the case, e.g., for a closed orbit
with $L=1$ (antipodal bound \eqref{eq:antipodal}).

\smallskip

\emph{A horse cannot get farther from its starting point than its top
	speed times the duration of the ride.} The result has a simple 
intuitive explanation in the case tested in Sec.~\ref{sec:steering_case_study},
where the field family forms an arc and is represented with one latent
coordinate ($L=d=1$).

With the latent diameter set to $1$ in Step~1,
$\Gamma_{\mathrm{dec}}=\operatorname{Lip}(\Phi_{\mathrm{dec}})$, i.e., the
largest ratio
\[
\frac{\|\Phi_{\mathrm{dec}}(z_1)-\Phi_{\mathrm{dec}}(z_2)\|}{|z_1-z_2|}
\]
between the change of the decoded field and the change of the latent
coordinate, over all pairs of latent points. It helps to think of the
latent coordinate $z$ as time: as $z$ runs over its unit interval, the
decoder traces a curve in field space, and
$\|\mathrm d\Phi_{\mathrm{dec}}/\mathrm dz\|$ is the speed $v(z)$ along it.

Consider now a horse running along a winding path, and suppose that it
never runs faster than a top speed $v_{\max}$. In a ride of duration
$T$, the path it covers is at most $v_{\max}T$ long, and its
straight-line distance from the starting point is at most the length of
the path, since the path may wind. Hence the distance from the start,
divided by the duration of the ride, never exceeds $v_{\max}$; it equals
$v_{\max}$ only for a short ride at top speed, when the path is
essentially straight.

The decoder behaves in the same way, with the latent coordinate as
time, the decoded curve as the path, and the field difference as the
straight-line distance. For two nearby points, the ratio above is the
local speed $v(z)$, so that its largest value is at least the largest
speed $v_{\max}=\max_z v(z)$. For two distant points $z_1<z_2$, the straight-line distance
between the decoded fields cannot exceed the length of the piece of
curve traced between them, which is the speed integrated over the
``time'' from $z_1$ to $z_2$ and is therefore at most
$v_{\max}(z_2-z_1)$, just as for the horse.
The ratio is therefore never larger than $v_{\max}$. Hence, on an
interval, $\Gamma_{\mathrm{dec}}$ is exactly the largest speed at which
the decoder traces the curve (Fig.~\ref{fig:covering}). For a decoder
that is Lipschitz but not differentiable everywhere, the speed exists
almost everywhere, and $v_{\max}$ is its essential supremum; the
argument is unchanged.

The consequence is immediate. Like a horse that must cover a given
distance in a given time, a decoder that traces an arc of length
$\mathrm{Vol}_1$ while the coordinate covers an interval of unit length
needs an average speed equal to $\mathrm{Vol}_1$, and hence a largest
speed at least as large: an accurate decoder must have
$\Gamma_{\mathrm{dec}}$ of the order of the length of the arc. A much
slower decoder, like a much slower horse, cannot cover the whole path: some fields
of the family remain farther than $\varepsilon$ from the decoded curve
(Fig.~\ref{fig:speed}). Accuracy and stability are thus in competition:
a small $\Gamma_{\mathrm{dec}}$ makes the decoder less sensitive to
latent errors, but below about $\mathrm{Vol}_1$ it is no longer
accurate.

\begin{figure}[t]
	\centering
	\begin{tikzpicture}[x=1cm,y=1cm,font=\scriptsize,
		arc/.style={black!70,line width=1.1pt},
		dec/.style={red!75!black,line width=1.1pt},
		tube/.style={red!20,line width=7pt,line cap=round}]
		% ---------- panel (a) ----------
		\begin{scope}
			\node[anchor=west,font=\small] at (-0.1,2.05) {(a) $\Gamma_{\mathrm{dec}}\ll\mathrm{Vol}_1$};
			\draw[|-|,thick] (0,1.45) -- (1,1.45) node[midway,above=1pt] {$z\in[0,1]$};
			\draw[->,red!75!black] (0.5,1.3) to[bend right=15] node[right=1pt] {$\Phi_{\mathrm{dec}}$} (0.55,0.75);
			\draw[tube] plot[smooth,domain=0:2.1,samples=40] ({\x},{0.35+0.25*sin(deg(1.6*\x))});
			\draw[arc] plot[smooth,domain=0:3.4,samples=60] ({\x},{0.35+0.25*sin(deg(1.6*\x))});
			\draw[dec] plot[smooth,domain=0:2.1,samples=40] ({\x},{0.35+0.25*sin(deg(1.6*\x))+0.06});
			\draw[blue!70!black,line width=2.2pt] plot[smooth,domain=2.25:3.4,samples=30] ({\x},{0.35+0.25*sin(deg(1.6*\x))});
			\node[align=center,text=blue!70!black] at (2.85,-0.35) {fields not\\reconstructed};
			\node[text=red!75!black,anchor=west] at (1.25,1.2) {decoded curve};
			\fill (0,0.35) circle (1.2pt); \fill (3.4,{0.35+0.25*sin(deg(1.6*3.4))}) circle (1.2pt);
		\end{scope}
		% ---------- panel (b) ----------
		\begin{scope}[xshift=4.15cm]
			\node[anchor=west,font=\small] at (-0.1,2.05) {(b) $\Gamma_{\mathrm{dec}}\gtrsim\mathrm{Vol}_1$};
			\draw[|-|,thick] (0,1.45) -- (1,1.45) node[midway,above=1pt] {$z\in[0,1]$};
			\draw[->,red!75!black] (0.5,1.3) to[bend right=15] node[right=1pt] {$\Phi_{\mathrm{dec}}$} (0.55,0.75);
			\draw[tube] plot[smooth,domain=0:3.4,samples=60] ({\x},{0.35+0.25*sin(deg(1.6*\x))});
			\draw[arc] plot[smooth,domain=0:3.4,samples=60] ({\x},{0.35+0.25*sin(deg(1.6*\x))});
			\draw[dec] plot[smooth,domain=0:3.4,samples=60] ({\x},{0.35+0.25*sin(deg(1.6*\x))+0.06});
			\node[text=red!75!black,anchor=west] at (1.25,1.2) {decoded curve};
			\node[align=center] at (1.7,-0.35) {whole arc within $\varepsilon$};
			\fill (0,0.35) circle (1.2pt); \fill (3.4,{0.35+0.25*sin(deg(1.6*3.4))}) circle (1.2pt);
		\end{scope}
	\end{tikzpicture}
	\caption{Accuracy versus stability for an arc ($L=d=1$). Black: field
	arc $\mathcal E$, of length $\mathrm{Vol}_1$; red: decoded curve, traced
	as the latent coordinate $z$ runs over its unit interval, with the band
	of points within $\varepsilon$ of it. (a)~A decoder with
	$\Gamma_{\mathrm{dec}}$ well below $\mathrm{Vol}_1$, like a horse much
	slower than the distance requires, cannot cover the whole arc: some
	fields (blue) remain farther than $\varepsilon$ from the decoded curve
	and are not reconstructed. (b)~A decoder with $\Gamma_{\mathrm{dec}}\gtrsim
	\mathrm{Vol}_1$ covers the whole arc. A small $\Gamma_{\mathrm{dec}}$
	makes the decoder less sensitive to latent errors, but below about
	$\mathrm{Vol}_1$ it is no longer accurate.}
	\label{fig:speed}
\end{figure}

The covering argument above reaches the same conclusion, up to a
constant that depends on how the arc bends. This dependence enters
through the constant $K$ of \eqref{eq:local-regularity}, which measures how
much of the arc can fall inside a small ball: a ball of radius $\rho$
centered on the arc contains a piece of it of length at most $2K\rho$
[Fig.~\ref{fig:regularity}(a) and Appendix~\ref{app:regularity}]. A
straight segment gives exactly $2\rho$, i.e., $K=1$; $K$ grows if the
arc curls up tightly or if distant parts of it come close to each
other, because a ball then contains more of the arc. For the steering
arc of Sec.~\ref{subsec:test-manifold}, i.e., the curve formed by the fields radiated by
the uniform linear array as its steering angle varies, $K\approx1.01$, estimated numerically (Table~\ref{tab:regularity}),
below the value $\pi/3$ guaranteed by Lemma~\ref{lem:curves}: at small
scales the arc is almost straight.

For an arc, all the constants of the covering argument are explicit,
$c_L=1$ and $C_0=\log_2(8K)$, so that $c_1=1/(16K)$, the constant of
\eqref{eq:prop1-explicit}, and Proposition~1 gives
$\Gamma_{\mathrm{dec}}\gtrsim\mathrm{Vol}_1/16$ for the steering arc. The
decoder must thus be at least one sixteenth as sensitive as the length
of the arc suggests. The factor $1/16$ is not physical: it is the price of counting coins
instead of following the arc, and it can be traced through the three
steps of the proof. A coin of the covering, of radius $2\varepsilon$,
may lie off the arc [Fig.~\ref{fig:covering}(b),(c)], so that \eqref{eq:local-regularity} can
be applied only to a coin of radius $4\varepsilon$ centered on the arc,
which contains a piece of arc of length at most
$2K\cdot4\varepsilon=8K\varepsilon$ (proof of
Lemma~\ref{lem:regularity}); covering the arc therefore requires at
least $\mathrm{Vol}_1/(8K\varepsilon)$ coins. The decoder supplies at most $1+\Gamma_{\mathrm{dec}}/\varepsilon$ of
them (Step~2), so that
$\Gamma_{\mathrm{dec}}+\varepsilon\ge\mathrm{Vol}_1/(8K)$. To obtain a
bound on $\Gamma_{\mathrm{dec}}$ alone, the term $\varepsilon$ must be
removed: for $\varepsilon\le\mathrm{Vol}_1/(16K)$, i.e., for small
enough $\varepsilon$, it is at most half of the right-hand side, so
that $\Gamma_{\mathrm{dec}}\ge\mathrm{Vol}_1/(8K)-\varepsilon
\ge\mathrm{Vol}_1/(16K)$ (Step~3). Half of the bound is thus given up
to obtain a simple bound that holds for every small enough $\varepsilon$.

\section{Validity of the Uniform Entropy Estimate}
\label{app:regularity}

The two lemmas below adapt classical results:
Lemma~\ref{lem:regularity} is the mass distribution principle
\cite{falconer2014fractal}, applied to the volume of $\mathcal E_s$;
Lemma~\ref{lem:curves} is a weaker form of Schur's comparison theorem
\cite{schur1921schwarzsche},
\cite{schur1921schwarzsche}, and its condition on $\rho_1$ corresponds
to the bottleneck in the reach of a manifold
\cite{niyogi2008finding,aamari2019estimating}.

\smallskip

\emph{A string laid on a table can be measured with coins; a string
	wound into a ball cannot.} 
	Proposition~1 requires the lower bound
	\eqref{eq:uniform-entropy} on the covering number, with constants
	independent of electrical size.
	 The covering number counts the
distinguishable fields much as one could measure the length of a
string by counting the coins needed to cover it. If the string lies on
a table, straight or gently curved, the count works: the number of
coins is about the length of the string divided by the diameter of a
coin. If the string is wound into a ball, a single coin covers many
turns, and the count underestimates the length. The same happens to a
field family: its volume measures the number of distinguishable fields
only if, at the scale of the resolution, the family is laid out rather
than wound up, i.e., it neither bends tightly nor folds back close to
itself. A larger source makes the string longer; the bound requires
that it does not also wind it more tightly.

This appendix turns this picture into a precise condition, shows that
the condition gives \eqref{eq:uniform-entropy}, expresses it in terms
of the curvature radius and of the distance between distant parts of
the field manifold, and explains why electromagnetic field families
satisfy it as electrical size grows, and when they do not.

\emph{The condition.} The family $\{\mathcal E_s\}$ is
\emph{uniformly locally regular} if \eqref{eq:local-regularity} holds with
$\rho_0>0$ and $K\ge1$ independent of $s$: no ball of radius
$\rho\le\rho_0$ centered on $\mathcal E_s$ contains more than $K$ times
the volume of a flat $d$-dimensional disk of the same radius; for a
curve, a piece of length at most $2K\rho$. In the picture above, a coin
of radius $\rho$ never covers much more than a straight piece of string
of length $2\rho$. The bound can fail only if, at small scales, the
manifold bends sharply or folds back onto itself
(Fig.~\ref{fig:regularity}(c)); field variations that merely
\emph{lengthen} the manifold increase its volume but do not violate it
(Fig.~\ref{fig:regularity}(b)). Bounds of this kind are standard in the
geometry of manifolds \cite{niyogi2008finding,aamari2019estimating}; the
two lemmas below are elementary forms of them.

\begin{figure*}[t]
	\centering
	\begin{tikzpicture}[>=stealth,
		every node/.style={font=\scriptsize},
		ball/.style={draw=black!45,fill=black!8},
		osc/.style={densely dashed,black!50},
		ttl/.style={font=\small}]
		\begin{scope}
			\draw[osc] (0.714,-0.009) circle (0.459);
			\draw[->,black!60] (0.714,-0.009) -- +(-150:0.459) node[pos=0.5,below=-1pt,text=black!70] {$R$};
			\draw[ball] (1.6,-0.166) circle (0.45);
			\draw[very thick,domain=0:2.95,samples=100] plot (\x,{0.45*sin(deg(2.2*\x))});
			\draw[line width=2.2pt,red!70!black,domain=1.276:1.970,samples=30] plot (\x,{0.45*sin(deg(2.2*\x))});
			\fill (1.6,-0.166) circle (1.3pt);
			\node[above right=-1pt] at (1.6,-0.166) {$x$};
			\draw[->] (1.6,-0.166) -- +(-145:0.45) node[pos=0.6,above=-1pt] {$\rho$};
			\node[align=center,text=red!70!black] at (2.45,0.75) {one piece,\\ length $\le2K\rho$};
			\node at (3.15,0.0) {$\mathcal E$};
			\node[ttl] at (1.5,-1.3) {(a) Local regularity};
		\end{scope}
		\begin{scope}[shift={(4.3,0)}]
			\foreach \c in {0.714,3.570,6.426}{\draw[osc] (\c,-0.009) circle (0.459);}
			\draw[->,black!60] (0.714,-0.009) -- +(-150:0.459) node[pos=0.5,below=-1pt,text=black!70] {$R$};
			\foreach \c in {2.142,4.998}{\draw[osc] (\c,0.009) circle (0.459);}
			\draw[very thick,domain=0:7.0,samples=240] plot (\x,{0.45*sin(deg(2.2*\x))});
			\node at (7.2,0.2) {$\mathcal E$};
			\node[align=center] at (3.5,0.85) {larger $\beta a$: the field oscillates more, the orbit is longer,\\ but every turn keeps the same radius $R/\|F\|$};
			\node[ttl] at (3.5,-1.3) {(b) Longer, not tighter};
		\end{scope}
		\begin{scope}[shift={(13.6,0)}]
			\draw[ball] (0.5,0.25) circle (0.65);
			\draw[very thick] plot[smooth] coordinates {(-1.0,0.27) (0.5,0.25) (1.6,0.27) (2.2,0.0) (1.6,-0.27) (0.5,-0.25) (-1.0,-0.27)};
			\draw[line width=2.2pt,red!70!black] (-0.15,0.26) -- (1.15,0.26);
			\draw[line width=2.2pt,red!70!black] (0.15,-0.25) -- (0.85,-0.25);
			\fill (0.5,0.25) circle (1.3pt) node[above=1pt] {$x$};
			\draw[<->] (-0.85,0.27) -- (-0.85,-0.27) node[midway,left=1pt] {gap};
			\node[align=center,text=red!70!black] at (0.6,1.15) {$m=2$ pieces in one ball};
			\node at (2.45,0.0) {$\mathcal E$};
			\node at (0.6,-0.85) {gap shrinking with $\beta a$};
			\node[ttl] at (0.6,-1.3) {(c) Folding back};
		\end{scope}
	\end{tikzpicture}
	\caption{Local regularity of a one-dimensional field manifold.
		(a)~A ball of radius $\rho\le\rho_0$ centered on $\mathcal E$ meets it
		in a single piece of length at most $2K\rho$; bending is limited by
		the minimum curvature radius $R$ (dashed osculating circle). (b)~A
		larger electrical size makes the orbit longer, but, in units of the
		field norm, its curvature radius does not shrink. (c)~When distant
		parts of $\mathcal E$ come closer than $\rho$, a ball meets several
		pieces, and the condition holds only with a larger $K$
		[cf.~\eqref{eq:multiplicity-bound}].}
	\label{fig:regularity}
\end{figure*}

\begin{lemma}[Local regularity implies \eqref{eq:uniform-entropy}]
	\label{lem:regularity}
	Under \eqref{eq:local-regularity}, for every $s$ and every $0<\eps\le\rho_0/4$,
	$N_{2\eps}(\mathcal E_s)\ge\mathrm{Vol}_d(\mathcal E_s)/(K\,V_d\,4^d\eps^d)$,
	so that \eqref{eq:uniform-entropy} holds with $\eps_{\max}=\rho_0/4$ and
	$C_0=\log_2(K\,V_d\,4^d)$.
\end{lemma}

\emph{Proof.} In the picture of the string, if each coin can cover at
most a certain length of string, covering the whole string requires at
least its length divided by that amount. The only subtlety is that a
coin of the covering may lie beside the string rather than on it,
whereas \eqref{eq:local-regularity} applies only to balls centered on the family: moving the
coin onto a point of the string that it covers, and doubling its
radius, gives a larger coin that contains the original one and to
which \eqref{eq:local-regularity} applies. Formally, discard the balls
of a $2\eps$-covering that miss $\mathcal E_s$, and pick in each
remaining ball a point $y_i\in\mathcal E_s$; by the triangle inequality,
the ball lies within the ball of radius $4\eps$ centered at $y_i$, to
which \eqref{eq:local-regularity} applies. Hence
$\mathrm{Vol}_d(\mathcal E_s)\le N\,K\,V_d\,(4\eps)^d$ for every covering,
and in particular for the smallest one; the logarithm gives
\eqref{eq:uniform-entropy}. The constants depend only on $K$, $\rho_0$,
and $d$. \hfill$\square$

\smallskip

\emph{Curves.} All field manifolds whose decoder sensitivity is analyzed
in the paper are one-dimensional. For them, \eqref{eq:local-regularity} follows from two
quantities: the minimum curvature radius $R$, which limits bending, and
the self-approach scale $\rho_1$, below which no ball meets two distant
parts of the curve.

The horse of Appendix~\ref{app:minimum_constant} gives the picture again.
A horse that runs at constant speed and cannot turn more sharply than
a radius $R$ soon leaves a small paddock: inside it, it can cover only
a short stretch of its path. The self-approach hypothesis excludes the
other possibility, namely that the path comes back through the same
paddock later on.

\begin{figure}[t]
	\centering
	\begin{tikzpicture}[x=1cm,y=1cm,font=\scriptsize]
		\def\R{2.6}\def\rr{1.1}
		\fill[black!7] (0,0) circle (\rr);
		\draw[black!45] (0,0) circle (\rr);
		\node[black!60] at (-1.5,-0.6) {$B(x,\rho)$};
		\draw[blue!60!black,dashed] ({\R*cos(232)},{\R+\R*sin(232)}) arc (232:308:\R);
		\draw[blue!60!black,->] (0,1.75) -- (0,0.06);
		\node[blue!60!black,right] at (0,1.2) {radius $R$, centre above};
		\draw[red!75!black,line width=1.1pt] plot[smooth,domain=-2.9:3.1,samples=80]
		({\x},{0.11*\x*\x+0.01*\x*\x*\x});
		\node[red!75!black] at (2.95,1.55) {$\mathcal E_s$};
		\draw[black!40] (-3.0,0) -- (3.3,0);
		\draw[->,line width=1.2pt] (0,0) -- (0.85,0) node[above=1pt] {$T(0)$};
		\fill (0,0) circle (1.3pt) node[above left=0pt] {$x$};
		\def\tt{2.2}
		\pgfmathsetmacro{\yy}{0.11*\tt*\tt+0.01*\tt*\tt*\tt}
		\fill[red!75!black] (\tt,\yy) circle (1.3pt) node[above left=-1pt] {$\gamma(t)$};
		\draw[densely dotted] (\tt,\yy) -- (\tt,-0.5);
		\draw[|-|,orange!85!black] (0,-0.5) -- (\tt,-0.5) node[right=1pt,align=left]{$\langle\gamma(t)-x,T(0)\rangle$\\$\ge R\sin(t/R)$};
		\fill[black] (1.0906,0.1438) circle (1.1pt);
		\fill[black] (-1.0936,0.1185) circle (1.1pt);
		\draw[decorate,decoration={brace,amplitude=3pt,mirror},black!70] (-1.09,-1.25) -- (1.09,-1.25)
		node[midway,below=3pt]{piece inside the ball: length $\le2R\arcsin(\rho/R)$};
	\end{tikzpicture}
	\caption{Proof of Lemma~\ref{lem:curves}. The curve $\mathcal E_s$ (red)
		bends no more than the circle of radius $R$ tangent to it at $x$
		(dashed). Hence the component of $\gamma(t)-x$ along the initial
		direction $T(0)$ (orange) is at least $R\sin(t/R)$, the value on the
		circle, and the curve leaves the ball $B(x,\rho)$ (grey) after a short
		piece, of length at most $2R\arcsin(\rho/R)\le(\pi/3)\,2\rho$.}
	\label{fig:lemma-curves}
\end{figure}

\begin{lemma}[Curvature and self-approach]
	\label{lem:curves}
	Let $\mathcal E_s$ be a $C^2$ curve in the real Hilbert space
	$\mathcal H_E$, with curvature radius everywhere at least $R>0$, and
	such that every ball of radius $\rho\le\rho_1$ centered on it meets it
	in a single piece. Then \eqref{eq:local-regularity} holds with $\rho_0=\min(\rho_1,R/2)$ and
	$K=\pi/3$.
\end{lemma}

\emph{Proof.} Fix a point $x\in\mathcal E_s$ and a radius
$\rho\le\min(\rho_1,R/2)$. Since $\rho\le\rho_1$, the hypothesis of the lemma ensures that the
ball $B(x,\rho)$ meets $\mathcal E_s$ in a single piece, i.e., no
distant part of the curve comes back into the ball; this piece contains
$x$, and it is therefore enough to bound its length
(Fig.~\ref{fig:lemma-curves}). Let
$\gamma(u)$ be the arclength parametrization of $\mathcal E_s$, with the
arclength $u$ measured from $\gamma(0)=x$ (positive on one side of $x$,
negative on the other), and $T(u)=\gamma'(u)$ the unit tangent; $t$
denotes the arclength of a generic point $\gamma(t)$, and $u$ the
integration variable between $0$ and $t$. In the picture of the horse,
$B(x,\rho)$ is the paddock, $\gamma(t)$ is the position of the horse
after a path of length $t$, and $T$ is its direction of motion.

\emph{(i) The tangent turns slowly.} Since the curvature radius is at
least $R$, $\|T'(u)\|\le1/R$. As $T(u)$ is a unit vector, the angle
between $T(u)$ and $T(0)$ is at most the length of the path traced by
$T$ on the unit sphere, i.e., at most $|u|/R$. Hence, for
$|u|\le\pi R/2$,
$\langle T(u),T(0)\rangle\ge\cos(u/R)$. In the picture of the horse,
after a path of length $|u|$ its direction has turned by at most
$|u|/R$.

\emph{(ii) The curve moves away from $x$.} Projecting
$\gamma(t)-x=\int_0^tT(u)\,\mathrm du$ onto $T(0)$ gives, for
$0\le t\le\pi R/2$,
\[
\begin{aligned}
	\|\gamma(t)-x\|&\ge\langle\gamma(t)-x,T(0)\rangle
	=\int_0^t\langle T(u),T(0)\rangle\,\mathrm du\\
	&\ge\int_0^t\cos\frac uR\,\mathrm du=R\sin\frac tR,
\end{aligned}
\]
and the same holds for negative $t$ with $|t|$ in place of $t$. 
The right-hand side is the same component for a point moving on a
circle of radius $R$. In the picture of the horse, it moves away from
its starting point at least as fast as one running on the tightest
circle it is allowed to follow.

\emph{(iii) The piece inside the ball is short.} The right-hand side
increases with $|t|$ up to $|t|=\pi R/2$, where it equals $R\ge2\rho$,
and it reaches $\rho$ at $|t|=R\arcsin(\rho/R)$. The curve therefore
leaves $B(x,\rho)$ in both directions before $|t|=R\arcsin(\rho/R)$,
and the piece through $x$ has length at most $2R\arcsin(\rho/R)$. Since
$\arcsin y/y$ increases with $y$ and equals $\pi/3$ at $y=1/2$, for
$\rho/R\le1/2$
\[
2R\arcsin\frac\rho R\le\frac\pi3\,2\rho ,
\]
which is \eqref{eq:local-regularity} with $K=\pi/3\approx1.05$. In the picture of the horse,
inside the paddock it covers at most $5\%$ more than the straight
diameter $2\rho$. \hfill$\square$

\smallskip

For $d>1$, the same role is played by the reach $\tau_{\mathcal E}$
\cite{federer1959curvature}, which is limited by the minimum curvature
radius and by the narrowest bottleneck of the manifold
\cite{niyogi2008finding,aamari2019estimating}.

\smallskip

\emph{Electrical size makes the manifold longer, not tighter.} The
constants of the bound in relative terms (Sec.~\ref{subsec:stable_widths_lipschitz})
are uniform in $s$ if, in units of the field norm, the curvature radius
and the self-approach scale remain bounded below; for the absolute
bound, the field norm must also be bounded below, as in all the
families of the paper, whose norm is fixed or grows with $s$.
This is a property of the electromagnetics. In the picture of the
string, the modal spectrum plays the role of the stiffness of the
string: a larger source is a longer piece of the same wire, which
covers a larger table but does not bend more sharply.

This can be made explicit for the field families of the paper. Both
have the form $F(\phi)=\sum_n c_n\,e^{\mathrm in(\theta-\phi)}$, with
$\phi$ the incidence angle and $c_n$ the modal coefficients for the
cylinder, and $\phi$ the steering phase and $c_n$ the element
excitations for the array. With orthogonal modes, the geometry of the
orbit is set by the moments $M_k=\sum_n n^k|c_n|^2$ of the modal power
spectrum, here with modes of unit norm (for the cylindrical harmonics of
Sec.~\ref{subsec:entropy_results}, all $M_k$ carry a common factor $2\pi$,
which cancels in the ratios below): $\|F\|=\sqrt{M_0}$; the speed
$\|\mathrm dF/\mathrm d\phi\|=\sqrt{M_2}$ is constant, so that the length
is proportional to $\sqrt{M_2}$ and the curvature radius is
$R=\|F'\|^2/\|F''\|=M_2/\sqrt{M_4}$; and
\[
\begin{aligned}
	\frac{\|F(\phi_1)-F(\phi_2)\|^2}{\|F\|^2}
	&=2\bigl[1-\operatorname{Re}\varrho(\Delta)\bigr],\\
	\varrho(\Delta)&=\frac{1}{M_0}\sum_n|c_n|^2e^{\mathrm in\Delta},
\end{aligned}
\]
with $\Delta=\phi_1-\phi_2$: the distance between two fields is set by
the correlation $\varrho$, the Fourier transform of the modal power
spectrum. Suppose that the spectrum keeps its shape as the size grows,
$|c_n|^2\simeq P\,w(n/n_\star)$, with $P$ a power scale, $w$ a fixed
dimensionless profile, and $n_\star\approx\mathrm{NDF}_\varepsilon/2$ the
half-width of the spectrum, i.e., $n_\star\approx\beta a$ for the
cylinder and $n_\star\approx N/2$ for the array. Then
$M_k\propto n_\star^{k+1}$, and
\[
\begin{aligned}
	\frac{\ell}{\|F\|}&\propto\sqrt{\frac{M_2}{M_0}}\propto n_\star,\\
	\frac{R}{\|F\|}&=\frac{M_2}{\sqrt{M_4M_0}}=\frac{m_2}{\sqrt{m_4m_0}},
\end{aligned}
\]
with $m_k=\int x^kw(x)\,\mathrm dx$: the orbit grows longer, while its
relative curvature radius is set by $w$ alone. 

\begin{figure}[t]
	\centering
	\begin{tikzpicture}[>=stealth,every node/.style={font=\scriptsize}]
		% ---------- (a) modal planes ----------
		\node[font=\small,anchor=west] at (-0.9,1.05) {(a)};
		\foreach \n/\x in {1/0,2/1.55,3/3.1}{
			\draw[black!45] (\x,0) circle (0.55);
			\draw[black!30] (\x-0.65,0) -- (\x+0.65,0);
			\draw[black!30] (\x,-0.65) -- (\x,0.65);
			\pgfmathsetmacro{\a}{25}
			\pgfmathsetmacro{\b}{25+28*\n}
			\draw[->,very thick,red!75!black] (\x,0) ++(\a:0.55) arc (\a:\b:0.55);
			\fill (\x,0) ++(\a:0.55) circle (1.3pt);
			\node at (\x,-0.85) {mode $n=\n$};
		}
		\node at (4.35,0) {$\cdots$};
		\node[align=center] at (1.55,-1.25) {as $\phi$ advances by $\Delta\phi$, mode $n$ turns by $n\,\Delta\phi$};
		% ---------- (b) projected orbits ----------
		\begin{scope}[yshift=-3.35cm]
			\node[font=\small,anchor=west] at (-0.9,1.35) {(b)};
			\begin{scope}[xshift=0.45cm]
				\draw[thick,blue!70!black,domain=0:360,samples=200,smooth,variable=\t]
				plot ({0.85*(cos(\t)+0.8*cos(2*\t+40))/1.28},{0.85*(sin(\t)+0.8*sin(2*\t+40))/1.28});
				\node[align=center] at (0,-1.6) {few modes:\\short curve};
			\end{scope}
			\begin{scope}[xshift=3.05cm]
				\draw[thick,red!75!black,domain=0:360,samples=600,smooth,variable=\t]
				plot ({0.85*(cos(\t)+0.8*cos(2*\t+40)+0.7*cos(3*\t+100)+0.6*cos(4*\t+15)+0.55*cos(5*\t+200)+0.5*cos(6*\t+300))/2.0},
				{0.85*(sin(\t)+0.8*sin(2*\t+40)+0.7*sin(3*\t+100)+0.6*sin(4*\t+15)+0.55*sin(5*\t+200)+0.5*sin(6*\t+300))/2.0});
				\node[align=center] at (0,-1.6) {many modes:\\long curve};
			\end{scope}
		\end{scope}
	\end{tikzpicture}
	\caption{Why the length of a one-dimensional field family grows with
		the number of degrees of freedom. (a)~When the control $\phi$
		(incidence or steering angle) advances by $\Delta\phi$, each modal
		coefficient $c_ne^{-\mathrm in\phi}$ turns in its own complex plane by
		$n\,\Delta\phi$: higher-order modes turn faster. (b)~Schematic two-dimensional
		projections of the resulting closed orbits, at the same field norm,
		for a source with few significant modes (left) and with many (right);
		the self-crossings appear only in the projection.
		Both are curves, i.e., one-dimensional, but the second, which combines
		faster rotations in more modal planes, is much longer: its length,
		$2\pi\sqrt{M_2/M_0}\,\|F\|$, grows as $\mathrm{NDF}_\varepsilon$.}
	\label{fig:helix}
\end{figure}

This also links the length of the orbit to the number of degrees of
freedom of classical EIT. For a full turn of $\phi$,
$\ell/\|F\|=2\pi\sqrt{M_2/M_0}$, i.e., $2\pi$ times the RMS width of
the modal spectrum, and, with $n_\star\approx\mathrm{NDF}_\varepsilon/2$,
\[
\frac{\ell}{\|F\|}\simeq\pi\sqrt{\frac{m_2}{m_0}}\;\mathrm{NDF}_\varepsilon,
\]
about $(\pi/\sqrt3)\,\mathrm{NDF}_\varepsilon\approx1.8\,\mathrm{NDF}_\varepsilon$
for a flat spectrum. The NDF thus sets the length of the field family,
not its dimension. The radiated fields lie, at accuracy $\eps$, in a
space of dimension $\mathrm{NDF}_\varepsilon$; the field family is a
curve in this space, along which each modal coefficient
$c_ne^{-\mathrm in\phi}$ rotates at a rate proportional to $n$, like a
helix turning in all modal planes at once (Fig.~\ref{fig:helix}).
The more significant modes,
the faster the highest of them rotate, and the longer the curve,
although it remains one-dimensional. Hence the number of distinguishable
incidence or steering angles, $\ell/(2\eps)$, grows as
$\mathrm{NDF}_\varepsilon$, and so does, for the steering arc and by
Proposition~2, the minimum relative sensitivity of a decoder; this gives
the shape exponent $\hat\nu=1$ of Table~\ref{tab:normalization}. This holds for controls that act as a
rotation or a phase shift; for load controls, the length is not tied to
the NDF in such a simple way.

Likewise, for large
$n_\star$, $\varrho(\Delta)\simeq\widehat w(n_\star\Delta)/\widehat w(0)$,
with $\widehat w$ the Fourier transform of $w$: the main lobe of the
correlation narrows, while its sidelobe level is fixed by $w$. Distant
parts of the orbit stay apart if, for non-neighboring fields,
$\operatorname{Re}\varrho\le1-\delta$ with $\delta>0$ independent of the
size, since then $\|F(\phi_1)-F(\phi_2)\|\ge\sqrt{2\delta}\,\|F\|$; this
condition is distinct from the curvature bound and must be checked
separately.

For the uniform array with half-wavelength spacing, whose element
patterns are orthogonal over the visible region in the variable
$u=\beta d\cos\theta$, $w$ is exactly a rectangle, $\varrho$ is the array
factor $\sin(N\Delta/2)/[N\sin(\Delta/2)]$ \cite{vantrees2002optimum},
and $m_2/\sqrt{m_4m_0}=\sqrt5/3$: the curvature radius tends to
$(\sqrt5/3)\|F\|\approx0.75\|F\|$. Over the steering range of
Sec.~\ref{subsec:test-manifold}, which contains no grating lobe, non-neighboring fields are
at least $1.31\|F\|$ apart for all tested apertures. In the picture of
the horse, the sidelobes measure how close the path comes back to
points already visited: near a grating lobe, it would pass close to
them along a distinct branch, and $K$ would grow; on the grating lobe
itself, it would retrace the same fields, which leaves the set of
fields unchanged but makes distinct controls indistinguishable
(Sec.~\ref{subsec:p-vs-d}). For the cylinder, $w$ is nearly flat up to
$|n|\approx\beta a$, beyond which the modal coefficients decay
super-exponentially \cite{abramowitz1964handbook}; if its shape does not
change with size, the curvature radius stays bounded, and the
separation is checked numerically. Over the computed range, the
estimated values give $\rho_0\approx0.37\|F\|$, so that
Lemma~\ref{lem:regularity} applies for relative resolutions
$\eps/\|F\|$ up to at least $0.09$, far coarser than the accuracy of
Sec.~\ref{subsec:observed}.

The reach of the arc, needed in Proposition~2, is estimated on the
closed-form field from Federer's two-point characterization
\cite[Th.~4.18]{federer1959curvature},
\[
\tau_{\mathcal E}=\inf_{\substack{a,b\in\mathcal E\\ a\ne b}}
\frac{\|b-a\|^2}{2\,\operatorname{dist}\bigl(b-a,\operatorname{Tan}(a)\bigr)},
\]
where $\operatorname{Tan}(a)$ is the tangent line at interior points and
the half-line pointing into the arc at the endpoints. Over a fine grid of pairs of
points, $\tau_{\mathcal E}\approx0.65\|F\|$ for all tested apertures. This
is an estimate, since the minimum over a grid can exceed the infimum,
but it appears in \eqref{eq:sharp} only through the factor $1-\eps/\tau_{\mathcal E}$:
at the accuracy of Sec.~\ref{sec:steering_case_study}, even a reach $50\%$ smaller would change
the bound by less than $1\%$.

Table~\ref{tab:regularity} confirms these conclusions over the computed
range and extends them to the elliptic cylinder, which has no
rotational symmetry. Its tightest member is the low-frequency one,
folded onto itself by the mirror symmetry of Sec.~\ref{subsec:diffusion_maps}, which growing electrical size progressively unfolds; even there,
the condition holds at the accuracy of the trained networks. The
self-crossings of Figs.~\ref{fig:dm} and~\ref{fig:orbits} occur in
two-dimensional projections, not in field space: they reflect the cost
of representing a long orbit with few coordinates, not a failure of
\eqref{eq:local-regularity}.

\begin{table}[t]
\centering
\caption{Regularity of the field manifolds of the case studies, in units
	of the RMS field norm $\|F\|$: minimum curvature radius $R$; scale
	$\rho_1$ below which every ball centered on the curve meets it in a
	single piece, as in Lemma~\ref{lem:curves}; and constant $K$ of \eqref{eq:local-regularity},
		evaluated for $\rho\le\rho_0=\min(\rho_1,R/2)$, for the steering arc of
	Sec.~\ref{sec:steering_case_study}, the PEC circular cylinder, and the dielectric elliptic
	cylinder of Sec.~\ref{sec:one_dimensional_manifolds}. $R$ is analytic for the steering arc and the
	circular cylinder; all other values are estimated on a fine grid of
	sampled fields. All remain bounded over the computed range of electrical sizes, which
	does not by itself prove uniformity for arbitrarily large sizes: for
	the steering arc, $R/\|F\|$ tends to a constant (about $0.745$ for
	uniform excitation); for the circular cylinder, uniformity is supported
	by the modal argument of this appendix; for the elliptic cylinder, it is
	verified only over the computed range.}
\label{tab:regularity}
\small
\setlength{\tabcolsep}{9pt}
\begin{tabular}{@{}lcccc@{}}
	\toprule
	Field family & Range & $R/\|F\|$ & $\rho_1/\|F\|$ & $K$ \\
	\midrule
	Steering arc & $N=4$--$128$ & $0.75$--$0.78$ & $\ge1.31$ & $1.011$ \\
	PEC circ.\ cyl. & $\beta a=1.57$--$31.4$ & $0.69$--$0.78$ & $\ge1.10$ & $\approx1.01$ \\
	Diel.\ ellip.\ cyl. & $\beta a=1.57$ & $0.014$ & $0.069$ & $\approx1.01$ \\
	& $\beta a=3.14$--$12.57$ & $0.07$--$0.51$ & $\ge0.34$ & $\approx1.01$ \\
	\bottomrule
\end{tabular}
\end{table}

\emph{When the estimate weakens.} The count of coins becomes unreliable
when the string is wound on a spool: a single coin then covers several
turns, and the count underestimates the length by the number of turns it
covers. For a field family, this happens when a small ball meets several
distinct pieces of the family; the estimate then holds with $K$
multiplied by the number of pieces, as quantified below, and fails only
if this number grows with the electrical size. Three mechanisms can
cause it. (i)~Grating lobes: if the array is
steered beyond its alias-free region, two different steering directions
produce almost the same pattern; in the picture of the horse, the path
comes back close to a paddock already visited, along a different
branch. (ii)~Symmetries
whose mirror branches come closer, rather than separate, as the size
grows, the opposite of
the elliptic cylinder of Sec.~\ref{subsec:diffusion_maps}. (iii)~Orbits that
wind ever more densely around a set of higher dimension, like a string
on a spool.

The effect on Proposition~1 can be quantified. If a ball of radius
$\rho$ meets at most $m_s(\rho)$ geometrically distinct pieces of
$\mathcal E_s$, the
number of turns covered by a coin, the proof of Lemma~\ref{lem:curves}
applied to each piece through a ball of radius $2\rho$ centered on it,
gives $\mathrm{Vol}_1(\mathcal E_s\cap B(x,\rho))\le4K\,m_s(\rho)\,\rho$
for $2\rho\le R/2$, i.e., \eqref{eq:local-regularity} with $K$ replaced by $2m_sK$, and
repeating the proof of
Proposition~1 yields, for curves ($d=1$), $16\eps\le R$, and $\eps$
small enough that the right-hand side exceeds $2\eps$,
\begin{equation}
	\Gamma_{\mathrm{dec}}
	\ge
	\frac{c_1}{2}\,
	\frac{\mathrm{Vol}_1(\mathcal E_s)}{m_s(4\eps)}.
	\label{eq:multiplicity-bound}
\end{equation}
The volume is thus divided by the number of turns, as in the exactly
non-injective case of Sec.~\ref{subsec:p-vs-d}, where distinct
controls produce the same field. If $m_s(4\eps)$ grows as fast as the
volume, the bound no longer grows with electrical size. This does not
mean that the required sensitivity stops growing: the bound only says
how sensitive a decoder must be at least, not how sensitive it needs to
be. When the bound stops growing, it simply no longer gives any
information on how the required sensitivity changes with size.

\section{Proof of Proposition~2}
\label{app:sharp}

The proof combines the Lipschitz bound on the nearest-point projection
onto a set of positive reach \cite[Th.~4.8(8)]{federer1959curvature}
with the bound of Lemma~\ref{lem:curves}, a weaker form of Schur's
comparison theorem \cite{schur1921schwarzsche}. To the author's
knowledge, their use to bound the sensitivity of a decoder is new.

\smallskip

For an arc, the covering argument of Appendix~\ref{app:minimum_constant}
can be replaced by a direct one: the decoded curve must follow the
whole arc, so its length, at most $\Gamma_{\mathrm{dec}}$, cannot be
much shorter than the length $\ell=\mathrm{Vol}_1(\mathcal E)$ of the arc
(Fig.~\ref{fig:sharp}). Two corrections make this precise. First, the
decoded curve does not run exactly on the arc, but within $\eps$ of it,
since each field is reconstructed within $\eps$; the length covered
along the arc is measured by the projection of the decoded curve onto
the arc, which may be slightly longer than the curve itself where the
arc bends. This gives the factor $1-\eps/\tau_{\mathcal E}$ of \eqref{eq:sharp}.
Second, the decoded curve need not start and end exactly at the
endpoints of the arc, but only within $\eps$ of them, so it may miss a
short piece at each end; this gives the term $2R\arcsin(2\eps/R)$. The
hypotheses are those of Lemma~\ref{lem:curves}, plus a reach
$\tau_{\mathcal E}>\eps$; $\gamma$ is the arclength parametrization of
$\mathcal E$, with endpoints $e_1=\gamma(0)$ and $e_2=\gamma(\ell)$.

\begin{figure}[t]
	\centering
	\begin{tikzpicture}[>=stealth,every node/.style={font=\scriptsize}]
		\draw[line width=0.8cm,black!9,line cap=round,domain=0:7.4,samples=100]
		plot (\x,{0.6*sin(22.5*\x)});
		\draw[very thick,domain=0:7.4,samples=100] plot (\x,{0.6*sin(22.5*\x)});
		\draw[line width=2.4pt,orange!85!black,domain=0:0.328,samples=10] plot (\x,{0.6*sin(22.5*\x)});
		\draw[line width=2.4pt,orange!85!black,domain=7.117:7.4,samples=10] plot (\x,{0.6*sin(22.5*\x)});
		\draw[thick,dashed,red!70!black,domain=0.3:7.1,samples=200]
		plot (\x,{0.6*sin(22.5*\x)+0.2*sin(130*\x)});
		\foreach \a/\b/\c/\d in {0.692/0.361/0.692/0.131, 2.077/0.237/2.077/0.467, 3.462/0.787/3.462/0.557, 4.846/0.367/4.846/0.597, 6.231/0.584/6.231/0.354}{
			\draw[->,black!55,line width=0.9pt] (\a,\b) -- (\c,\d);}
		\fill (0,0) circle (1.5pt) node[below left=-1pt] {$e_1$};
		\fill (7.4,0.138) circle (1.5pt) node[below right=-1pt] {$e_2$};
		\fill[red!70!black] (0.3,0.196) circle (1.4pt);
		\node[above left=-1pt,text=red!70!black] at (0.3,0.25) {$c$};
		\draw[<->] (3.0,0.546) -- (3.0,0.946) node[midway,left=-1pt] {$\eps$};
		\node[text=red!70!black,align=center] at (2.2,1.35) {decoded curve $\mathcal C$,\\ length $\le\Gamma_{\mathrm{dec}}$};
		\draw[red!70!black] (2.2,1.05) -- (2.35,0.75);
		\node at (5.6,-0.15) {$\mathcal E$, length $\ell$};
		\node[black!60] at (5.3,1.2) {projection $\pi$ (gray arrows)};
		\draw[black!60] (5.3,1.05) -- (4.9,0.55);
		\node[text=orange!85!black,align=center] at (1.0,-0.75) {end pieces (thick), each\\ $\le R\arcsin(2\eps/R)$};
		\draw[orange!85!black] (0.55,-0.45) -- (0.18,0.0);
	\end{tikzpicture}
\caption{Proof of Proposition~2. Both curves lie in field space. The
	\emph{road} is the arc $\mathcal E$ (solid black), of length $\ell$:
	the set of true fields. The \emph{path} is the decoded curve
	$\mathcal C$ (dashed red): the set of fields produced by the decoder
	as the latent coordinate runs over its interval; its wavy shape is
	only schematic.  Since every field is reconstructed within $\eps$,
	the path may leave the road, but it stays inside a lane of
	half-width $\eps$ around it (shaded), and it is at most
	$\Gamma_{\mathrm{dec}}$ long. To measure how much of the road the
	path covers, each point of the path is projected onto the road
	(gray arrows), giving the shadow of the path on the road. The shadow
	links the two lengths: on the one hand, where the road bends, it can
	be longer than the path, but at most by the factor
	$\tau_{\mathcal E}/(\tau_{\mathcal E}-\eps)$; on the other hand, it
	covers the whole road except two short end pieces (thick orange),
	each of length at most $R\arcsin(2\eps/R)$, because the path passes
	within $\eps$ of the endpoints of the road, e.g., through the point
	$c$ near $e_1$. Hence the path cannot be much shorter than the
	road: a decoder cannot be much less sensitive than the length of the
	arc it must trace.}
	\label{fig:sharp}
\end{figure}

\emph{Lower bound.} The proof has three steps.

\emph{(i) The decoded curve stays in the tube.} Since $\mathcal E$ is
connected and the encoder is continuous, the latent image is an
interval $[a,b]$, of length equal to its diameter. Every point of it is
the code of a field reconstructed within $\eps$, so that the decoded
curve $\mathcal C=\Phi_{\mathrm{dec}}([a,b])$ lies in the tube of radius
$\eps$ around $\mathcal E$, and its length is at most
$\operatorname{Lip}(\Phi_{\mathrm{dec}})(b-a)=\Gamma_{\mathrm{dec}}$. In
the picture of the horse, it rides inside the lane, along a path at
most $\Gamma_{\mathrm{dec}}$ long.

\emph{(ii) The shadow is at most slightly longer than the curve.}
Since $\eps<\tau_{\mathcal E}$, the nearest-point projection $\pi$ onto
$\mathcal E$ is defined on the tube and is Lipschitz with constant
$\tau_{\mathcal E}/(\tau_{\mathcal E}-\eps)$
\cite[Th.~4.8(8)]{federer1959curvature}. As $\mathcal C$ is connected and
$\pi$ is continuous, $\pi(\mathcal C)$ is a sub-arc
$\gamma([s_1,s_2])$, whose length satisfies
\[
s_2-s_1\le\Gamma_{\mathrm{dec}}\,
\frac{\tau_{\mathcal E}}{\tau_{\mathcal E}-\eps}.
\]
In the picture of the horse, its shadow on the road is a single stretch
of road, slightly longer than its path at most where the road bends.

\emph{(iii) The shadow covers almost the whole arc.} The decoded image
$c=\Phi_{\mathrm{dec}}(\Psi_{\mathrm{enc}}(e_1))$ of $e_1$ is within $\eps$
of $e_1$. Its projection $\pi(c)$ is the point of $\mathcal E$ closest to
$c$, so that $\|\pi(c)-c\|\le\|e_1-c\|\le\eps$ and
$\|\pi(c)-e_1\|\le2\eps$. By the bound
$\|\gamma(t)-e_1\|\ge R\sin(t/R)$ of the proof of Lemma~\ref{lem:curves}
and the self-approach hypothesis, applied at the scale
$2\eps\le\min(R/2,\rho_1)$, the point $\pi(c)$ lies within arclength
$R\arcsin(2\eps/R)$ of $e_1$. Hence $s_1\le R\arcsin(2\eps/R)$, and
likewise $\ell-s_2\le R\arcsin(2\eps/R)$, i.e.,
$s_2-s_1\ge\ell-2R\arcsin(2\eps/R)$. In the picture of the horse, it
starts and ends within $\eps$ of the ends of the road, so its shadow
misses only two short end pieces.

Combining (ii) and (iii),
\[
\Gamma_{\mathrm{dec}}\ge
\Bigl(1-\frac{\eps}{\tau_{\mathcal E}}\Bigr)
\bigl[\ell-2R\arcsin(2\eps/R)\bigr],
\]
which is \eqref{eq:sharp}.

\emph{Attainment.} The bound is attained as $\eps\to0$. The arclength
chart, $\Psi_{\mathrm{enc}}(\gamma(s))=s$ and
$\Phi_{\mathrm{dec}}(z)=\gamma(\min\{\max\{z,0\},\ell\})$, reconstructs
$\mathcal E$ exactly. The encoder is continuous, since $\gamma$ is
continuous and injective on the compact interval $[0,\ell]$, so that its
inverse is continuous; the decoder clamps the latent
coordinate to $[0,\ell]$ and follows $\gamma$ at unit speed, so that
$\operatorname{Lip}(\Phi_{\mathrm{dec}})=1$ and, with a latent diameter
$\ell$, $\Gamma_{\mathrm{dec}}=\ell$. In the picture of the horse, it
rides exactly on the road, at constant speed, from one end to the
other. \hfill$\square$

As $\eps\to0$, the lower bound tends to $\ell$: for arcs, the least
sensitive accurate decoder has $\Gamma_{\mathrm{dec}}=\ell-\mathcal
O(\eps)$, i.e., Proposition~1 holds with constant $c_1\to1$ instead of
$1/(16K)$. The argument relies on the latent image being an interval
and does not extend directly to $d\ge2$ or $L>d$.

\supppart{II}{Supporting Analyses for the Case Studies}

Part~II gives the analyses and details needed to interpret and
reproduce the case studies of Secs.~\ref{sec:array_case_study}--\ref{sec:one_dimensional_manifolds}. It uses classical tools:
multiport and array-manifold models of arrays
\cite{harrington1978reactively,manikas2004differential}; modal
expansions, Bessel asymptotics, and the first-Born approximation for
cylinders \cite{harrington1961time,abramowitz1964handbook,colton1998inverse};
the Borsuk--Ulam theorem \cite{borsuk1933drei}, which for closed curves
reduces to Bolzano's theorem, for the antipodal bound;
and interval arithmetic for guaranteed bounds on trained decoders
\cite{moore1966interval}.

\section{Normalization and Observation Conditions}
\label{supp:normalization}

This appendix complements Sec.~\ref{subsec:normalization} with the
observation conditions and the normalization used in each part of the
paper (Table~\ref{tab:norm-summary}).

\emph{Observation conditions.} The exponents describe the source,
rather than the measurement set-up, if the field is sampled at or above
the Nyquist rate of the largest source over the whole visible region
\cite{bucci1989degrees}; otherwise the gain exponent $\nu_{\mathrm g}$ also depends
on the set-up.

% For the circular PEC cylinder, the circle of radius $8\lambda$ used for
%the elliptic cylinder of Sec.~\ref{subsec:scattering_setup} and the far field give the same
%exponents within $0.02$.

\emph{Normalization.} Except in the coupling analysis of
Appendix~\ref{supp:pvsd}, all fields of a family are scaled by a common
factor, fixed by a unit source amplitude (unit voltage, unit excitation
per element, or unit incident field). As shown in Sec.~\ref{subsec:normalization}, such a scale affects only the gain exponent, and hence $\nu$,
which is therefore always reported with the source constraint; the
dimension, the topology, $\hat\nu$, and dimensionless ratios such as
$\Gamma_{\mathrm{dec}}/\mathrm{Vol}_1$ are unchanged. The coupling analysis instead normalizes each field to unit norm, as in
\eqref{eq:perconfig-metric}, so that only the shape of the pattern is
retained. Two controls whose fields differ only by a scale factor then
give the same normalized field, which in general can reduce the
dimension or change the topology of the family. This cannot happen for
the load-controlled array. Its field \eqref{eq:load-structure} is
$F^0+I_{\mathrm p,1}\mathbf f_1+I_{\mathrm p,2}\mathbf f_2$, where $F^0$
is the field radiated by the driven elements with the loaded ports open
and $\mathbf f_1$, $\mathbf f_2$ are the patterns of the two parasitic
elements. These three patterns are linearly independent: if two fields
were proportional, the coefficient of $F^0$ would force the factor to be
one, so that two distinct fields never give the same normalized field.

\begin{table}[t]
\centering
\caption{Normalization used in each part of the paper and its influence
	on the reported results. All entries use a scale common to the whole
	family, except the coupling analysis, which normalizes each field to
	unit norm.}
\label{tab:norm-summary}
\small
\setlength{\tabcolsep}{5pt}
\renewcommand{\arraystretch}{1.15}
\begin{tabularx}{\textwidth}{@{}>{\hsize=0.9\hsize}L>{\hsize=0.8\hsize}L>{\hsize=0.8\hsize}L>{\hsize=1.5\hsize}L@{}}
	\toprule
	Part & Normalization & Reported & Influence \\
	\midrule
	Sec.~\ref{sec:array_case_study} & unit voltage & $L_{\mathrm{AE}}$, errors & none (relative errors) \\
		Appendix~\ref{supp:pvsd} & per field & dimension, slopes, $w$ & dimension and topology preserved for this array \\
	Secs.~\ref{subsec:test-manifold}, \ref{subsec:observed} & unit excitation & $\mathrm{Vol}_1$, $\Gamma_{\mathrm{dec}}$ & ratio $\Gamma_{\mathrm{dec}}/\mathrm{Vol}_1$ unchanged \\
	Sec.~\ref{sec:communication_theory} & unit excitation & $\nu_\psi$, $\nu_\chi$ & $\nu_{\mathrm g}=1/2$; at fixed power, $1$ and $0$ \\
	Sec.~\ref{subsec:entropy_results} & unit incident field & $\nu$, $\nu_{\mathrm g}$, $\hat\nu$ & $\nu_{\mathrm g}$ from $3/2$ to $1/2$; $\hat\nu\approx1$ \\
	Sec.~\ref{subsec:scattering_setup} & unit incident field & $\nu$ & as in Sec.~\ref{subsec:entropy_results} \\
	Secs.~\ref{subsec:diffusion_maps}, \ref{subsec:heatmap_results} & unit incident field & $d$, $\Delta_{\mathrm{br}}$, $L_{\mathrm{AE}}$ & none (scale-invariant) \\
	\bottomrule
\end{tabularx}
\end{table}

\section{Control Parameters, Intrinsic Dimension, and Mutual Coupling}
\label{supp:pvsd}

This appendix completes Secs.~\ref{subsec:p-vs-d} and~\ref{sec:array_case_study}. It shows
that the intrinsic dimension $d$ equals the number of controls $p$ in
the case studies, that the field manifold of the load-controlled array
has the shape of its control space, and how mutual coupling affects the
result. Table~\ref{tab:p-vs-d} lists $p$ and $d$ for the configurations
of the paper and for some limiting cases in which they differ.

\begin{table}[t]
	\centering
	\caption{Number of control parameters $p$ versus intrinsic dimension
		$d$ of the accessible field manifold in the benchmark and limiting
		configurations.}
	\label{tab:p-vs-d}
	\small
	\setlength{\tabcolsep}{9pt}
	\begin{tabular}{@{}lccl@{}}
		\toprule
		Configuration & $p$ & $d$ & Remark \\
		\midrule
		Tied loads (Sec.~\ref{subsec:coupling}) & 1 & 1 & $S^1$, $L_{\min}=2$ \\
		Independent loads & 2 & 2 & $d=2$ on the whole torus \\
		\quad one load weakly coupled & 2 & 2 & $d_\eps=1$ above tube width \\
		Amplitude--phase feed & 2 & 2 & annulus, $L_{\min}=2$ \\
		Steering + polarization & 2 & 2 & orthogonal directions \\
		Two steering beams & 2 & 2 & $d_{\mathrm{eff}}\to1$ for close beams \\
		Cylinder illumination & 1 & 1 & near 2-to-1 at low $\beta a$ \\
		Feed phase + steering, power obs. & 2 & 1 & global phase unobservable \\
		Cascaded phase shifters & 2 & 1 & fibers $\theta_1+\theta_2=$ const \\
		\bottomrule
	\end{tabular}
\end{table}

\emph{The dimension equals the number of controls.} The intrinsic
dimension is the number of independent directions in which the field
changes, i.e., the rank of the Jacobian $D\Phi$. In all the case
studies, the field depends smoothly, in fact analytically, on the
controls: the steered array through trigonometric and exponential
functions of $\psi$ and $\chi$, \eqref{eq:dp-field-main}; the fed and amplitude--phase
arrays linearly through their excitations; the cylinder linearly
through the incident plane wave; and the load-controlled array through
load currents that are rational functions of the reflection
coefficients $e^{\mathrm j\phi_k}$, with a denominator that never
vanishes, since the loaded ports radiate power for every nonzero
current. For such maps, the rank can be lower than its maximum only on
a set of control values of measure zero \cite{krantz2002primer}. In
all the case studies this maximum equals $p$
(Table~\ref{tab:p-vs-d}), so that $d=p$ at almost all control values.

\emph{The load-controlled array has the shape of its control space.}
Sec.~\ref{subsec:p-vs-d} shows that, if distinct controls always
produce distinct fields, the field manifold has the topology of the
control space, which fixes $L_{\min}$. For the load-controlled array,
this is checked in two steps. First, the field determines the load
currents. By \eqref{eq:load-structure}, the field is
$F^0+I_{\mathrm p,1}\mathbf f_1+I_{\mathrm p,2}\mathbf f_2$, where $F^0$
is radiated by the driven elements with the loaded ports open and
$\mathbf f_1$, $\mathbf f_2$ are the patterns produced by a unit
current at each loaded port, including mutual coupling. These three
patterns are linearly independent: with columns normalized to unit
norm, the matrix $[F^0\ \mathbf f_1\ \mathbf f_2]$ sampled at the $180$
observation points has smallest singular value $0.33$, far from zero.
Two different pairs of currents therefore cannot produce the same
field. Second, the load currents determine the loads: each reactance
follows from the current and the voltage at its port whenever the
current is nonzero, and, for the symmetric array of Sec.~\ref{subsec:coupling}, a zero current occurs only for an open circuit. Distinct
loads therefore produce distinct fields. The same holds for the fed
and amplitude--phase arrays, whose fields are the sum of a fixed term
and of linearly independent element patterns weighted by the
controlled excitations.

\begin{figure}[t]
	\centering
	\begin{tikzpicture}[>=stealth,every node/.style={font=\scriptsize}]
		\begin{scope}
			\node[font=\small,anchor=west] at (-1.55,1.75) {(a)};
			\foreach \i in {0,...,119}{
				\pgfmathsetmacro{\pa}{3*\i}
				\pgfmathsetmacro{\pb}{3*\i+3.2}
				\pgfmathsetmacro{\T}{cos(\pa+1.5)+0.6*sin(2*(\pa+1.5)+57.3)+0.3*cos(3*(\pa+1.5))}
				\pgfmathsetmacro{\c}{max(0,min(100,50+31*\T))}
				\draw[line width=5pt,red!\c!blue] (\pa:1.15) arc (\pa:\pb:1.15);
			}
			\draw[thick,dashed] (90:1.15) -- (270:1.15);
			\fill[black] (90:1.15) circle (2.2pt); \fill[black] (270:1.15) circle (2.2pt);
			\node[above=3pt] at (90:1.15) {$\phi^\ast$};
			\node[below=3pt] at (270:1.15) {$\phi^\ast+\pi$};
			\node[align=center] at (0,0) {same\\temperature};
			\node[red!80!black] at (1.55,0.95) {hot};
			\node[blue!80!black] at (-1.55,-0.95) {cold};
		\end{scope}
		\begin{scope}[xshift=2.35cm,yshift=-0.9cm]
			\node[font=\small,anchor=west] at (-0.45,2.65) {(b)};
			\draw[->] (0,0.9) -- (4.0,0.9) node[right] {$\phi$};
			\draw[->] (0,-0.15) -- (0,2.4);
			\draw[thick,domain=0:360,samples=120,smooth,variable=\t]
			plot ({\t/100},{0.9+0.38*(-2*cos(\t)-0.6*cos(3*\t))});
			\fill (0.9,0.9) circle (1.6pt) node[below right=-1pt] {$\phi^\ast$};
			\fill (2.7,0.9) circle (1.6pt);
			\node[align=left] at (2.3,2.6) {$h(\phi)=T(\phi+\pi)-T(\phi)$};
			\node at (0.32,0.62) {$h<0$};
			\node at (1.8,1.3) {$h>0$};
			\draw (3.6,0.85) -- (3.6,0.95) node[below=3pt] {$2\pi$};
		\end{scope}
	\end{tikzpicture}
	\caption{An example of the antipodal bound. (a)~A ring with a
		continuous temperature distribution (blue cold, red hot). (b)~The
		difference $h(\phi)=T(\phi+\pi)-T(\phi)$ between the temperatures of
		two diametrically opposite points changes sign after half a turn, since
		the two points exchange roles, and must therefore vanish at some
		$\phi^\ast$: there the two opposite points have the same temperature
		(dots in (a)); the second zero, at $\phi^\ast+\pi$, describes the same
		pair in reverse order. The same argument applies to a closed field orbit, with
		the temperature replaced by the code of a continuous one-dimensional
		encoder: two fields half a period apart receive the same code and
		hence the same reconstruction, so one of them has error at least
		$w/2$, the bound \eqref{eq:antipodal}.}
	\label{fig:antipodal}
\end{figure}

\emph{Why one coordinate is not enough for a closed orbit.} The
antipodal bound \eqref{eq:antipodal} follows from Bolzano's theorem: a
continuous function that changes sign must vanish somewhere. Let
$\Psi_{\mathrm{enc}}$ be any continuous encoder with one coordinate,
and compare the codes of two fields half a period apart,
\[
h(\phi)=\Psi_{\mathrm{enc}}(\Phi(\phi+\pi))-\Psi_{\mathrm{enc}}(\Phi(\phi)).
\]
Since $h(\phi+\pi)=-h(\phi)$, $h$ changes sign over half a turn and
vanishes at some $\phi^\ast$. The two fields $\Phi(\phi^\ast)$ and
$\Phi(\phi^\ast+\pi)$, at distance at least $w$, then receive the same
code and the same reconstruction, which cannot be within less than
$w/2$ of both (Fig.~\ref{fig:antipodal}). The argument does not require
distinct controls to produce distinct fields. It is the
one-dimensional case of the Borsuk--Ulam theorem \cite{borsuk1933drei},
which also underlies classical lower bounds on nonlinear widths
\cite{devore1989optimal}.

\emph{Mutual coupling.} In the induced-EMF model of the load-controlled
array, mutual coupling does not reduce the dimension, but makes the two
controls unequal: the field manifold remains two-dimensional on the
whole torus, but the field changes up to about fifty times more slowly
along some control directions than along others. A weakly coupled
element gives a weak control in the sense of Sec.~\ref{subsec:p-vs-d}: when the second parasitic element is moved away from the driven
elements, its load changes the field less and less, and at a distance
of $16\lambda$ the slope of $H_\eps$ is close to $1$ instead of $2$,
down to a relative resolution of $-32$~dB. In the configuration of the
case studies, the slope is close to $2$ at the resolutions of interest.

\section{Arclength Scaling for Cylindrical Scatterers}
\label{app:derivation}
\label{app:born}

This appendix derives the exponents of Sec.~\ref{subsec:entropy_results} for
the circular cylinder, first in the geometric (PEC) and weak-contrast
(Born) limits, and then from the exact modal series across the
transition between them.

The physical picture is the same in all cases, and can be seen as a
clock with many hands, one for each cylindrical harmonic $n$ of the
scattered field (Fig.~\ref{fig:helix}). The length of a hand is the
amplitude of the harmonic, and, when the illumination turns by a small
angle, the hand $n$ turns $n$ times as much. The speed of the field
combines the speeds of the tips of all the hands, and the length of the
orbit is this speed times the full turn $2\pi$. A larger cylinder has
more hands, up to $|n|\approx\beta a$, and its fastest hands turn
faster, so that its orbit is longer. The scattering regime mainly
changes the length of the hands, i.e., how much power is scattered, but
not how many hands there are.

\emph{PEC cylinder.} For a circular PEC cylinder of radius $a$ under
$\mathrm{TM}_z$ plane-wave illumination at angle $\phi$, the far-field
pattern is \cite{harrington1961time}
\begin{equation}
	F(\vartheta;\phi)
	=
	-\sum_{n}
	R_n\,e^{\mathrm{j}n(\vartheta-\phi)},
	\qquad
	R_n
	=
	\frac{J_n(\beta a)}{H_n^{(2)}(\beta a)} ,
\end{equation}
so that, by orthogonality,
$\|\partial_\phi F\|^2=2\pi\sum_nn^2|R_n|^2$, independent of $\phi$. For
$|n|$ below $\beta a$, $|R_n|^2$ oscillates around $1/2$; above, it
decays exponentially; the transition region, of width
$\mathcal O((\beta a)^{1/3})$, contributes only a lower-order term
\cite{abramowitz1964handbook,colton1998inverse}. Hence
$\sum_nn^2|R_n|^2\sim\sum_{|n|\le\beta a}n^2/2\sim(\beta a)^3/3$, and the
length of the orbit over a full turn is
\[
\ell_{\mathcal E}
=
2\pi\|\partial_\phi F\|
=
\Theta\bigl((\beta a)^{3/2}\bigr),
\]
i.e., $\nu=3/2$, the geometric value of Sec.~\ref{subsec:entropy_results}.
In the same way,
$\|F\|^2=2\pi\sum_n|R_n|^2\sim2\pi\beta a$, so that $\nu_{\mathrm g}=1/2$
and $\ell_{\mathcal E}/\|F\|\sim2\pi\beta a/\sqrt3$, i.e., $\hat\nu=1$.
For the elliptic cylinder, the speed varies with $\phi$, but the band
edge is still set by the largest dimension, $\approx\beta a$, so the
same exponent is expected in the geometric regime: on the extended MoM
sweep of Sec.~\ref{subsec:scattering_setup}, the local exponent reaches $1.52$ at $\tau=15.7$.

\emph{Weak-contrast (Born) regime.} For a homogeneous dielectric
cylinder with susceptibility $\chi_{\mathrm e}=\eps_r-1$, the first-Born
approximation gives the modal coefficients \cite{colton1998inverse}
\begin{equation}
	a_n^{\mathrm B}(\phi)
	=
	C_{\mathrm B}\,\chi_{\mathrm e}\,\beta^2 I_n
	e^{-\mathrm j n\phi},
	\qquad
	I_n
	=
	\int_0^a J_n^2(\beta r)\,r\,\mathrm dr ,
\end{equation}
with $C_{\mathrm B}$ a constant independent of $n$, $a$, $\beta$, and
$\phi$, so that
$\|\partial_\phi F_{\mathrm B}\|^2=2\pi|C_{\mathrm B}|^2\chi_{\mathrm e}^{2}
\beta^4\sum_nn^2I_n^2$. Lommel's identity and the Debye expansions of
$J_n$ for $|n|<z=\beta a$ \cite{abramowitz1964handbook} give
$I_n\sim(a/\pi\beta)\sqrt{1-(n/z)^2}$, and replacing the sum by an
integral,
$\sum_nn^2I_n^2\sim4a^2z^3/(15\pi^2\beta^2)$. Hence
\[
\ell_{\mathcal E}^{(\mathrm B)}
=
\Theta\bigl(|\chi_{\mathrm e}|\,z^{5/2}\bigr),
\qquad
\nu_{\mathrm{Born}}=\frac52,
\]
at fixed contrast; the same asymptotics give
$\|F_{\mathrm B}\|^2=\Theta(\chi_{\mathrm e}^2z^3)$, so that
$\nu_{\mathrm g}=3/2$ and, again, $\hat\nu=1$. The extra power of
$\beta a$ comes from the $\beta^2$ factor of the Born integral, i.e.,
from the polarization current filling the whole cross-section. In both
regimes, the power is spread over the same band $|n|\lesssim\beta a$, so
that the normalized length, and hence $\hat\nu=1$, is the same; only the
scattered power, and hence $\nu_{\mathrm g}$, changes.

\emph{Validity and transition.} The Born approximation requires a small
phase delay across the cylinder, $\tau=(\eps_r-1)\beta a\ll1$, which
fails at fixed contrast as $\beta a$ grows, whereas the Debye
expansions used above require $\beta a\gg1$. The exponent $5/2$ thus
describes the window $1\ll\beta a\ll|\eps_r-1|^{-1}$ of electrically
large, optically thin cylinders at weak contrast, not the electrically
small limit; for $\tau\gtrsim10$ the
modal reflectivity saturates as for the PEC, and the exponent $3/2$ is
recovered. The local exponents computed from the exact modal series,
described below, fall approximately on a common curve of $\tau$ for
$\eps_r$ between $1.05$ and $4$ (Fig.~\ref{fig:exponents}); the local
exponent of the elliptic cylinder, computed on the MoM sweep of
Sec.~\ref{subsec:scattering_setup}, follows the same transition, from $2.71$ to $1.52$.

\emph{Dielectric cylinder: exact exponents.} For a homogeneous
dielectric cylinder of radius $a$ and refractive index
$m=\sqrt{\eps_r}$, the modal coefficients of the $\mathrm{TM}_z$
scattered field are \cite{harrington1961time}
\begin{equation}
	c_n=\frac{J_n(m\beta a)\,J_n'(\beta a)-m\,J_n'(m\beta a)\,J_n(\beta a)}
	{m\,J_n'(m\beta a)\,H_n^{(2)}(\beta a)-J_n(m\beta a)\,H_n^{(2)\prime}(\beta a)},
	\label{eq:diel-coeff}
\end{equation}
which approach the PEC coefficients $c_n=-R_n$ for large $\eps_r$, away
from the internal resonances of the cylinder. The moments
$M_k=2\pi\sum_nn^k|c_n|^2$ of Sec.~\ref{subsec:entropy_results} are summed
over $|n|\le\beta a+8(\beta a)^{1/3}+15$, beyond which the coefficients
are negligible, for $\eps_r=1.05$, $1.2$, $1.5$, $2$, and $4$, at $64$
values of $\beta a$ evenly spaced on a logarithmic scale between $4$
and $250$. At each $\beta a$, the local exponents $\nu$ and
$\nu_{\mathrm g}$ are the slopes of least-squares straight-line fits of
$\log\sqrt{M_2}$ and $\log\sqrt{M_0}$ versus $\log\beta a$ over the
window $[\beta a/3.2,\,3.2\,\beta a]$, i.e., a factor of about ten, and
$\hat\nu=\nu-\nu_{\mathrm g}$. The window averages out the internal
resonances of the cylinder: with a window half as wide on the
logarithmic scale, $\hat\nu$ oscillates between $0.66$ and $1.1$ for
$\eps_r=1.5$, instead of between $0.88$ and $1.1$. Fig.~\ref{fig:exponents} plots the means over the five contrasts versus $\tau$, with their
spread. In the Born regime ($\eps_r=1.05$, $\beta a=10$, $\tau=0.5$),
$\nu=2.51$, $\nu_{\mathrm g}=1.47$, and $\hat\nu=1.04$; at the
undershoot ($\eps_r=1.5$, $\beta a=20$, $\tau=10$), $\nu=1.39$,
$\nu_{\mathrm g}=0.43$, and $\hat\nu=0.96$; in the geometric regime
($\eps_r=1.5$, $\beta a=50$--$120$, $\tau=25$--$60$), $\nu=1.47$--$1.49$,
$\nu_{\mathrm g}=0.49$--$0.53$, and $\hat\nu=0.94$--$1.00$. Fits over
the whole range $4\le\beta a\le250$ give $\hat\nu=1.00$ for
$\eps_r=1.5$, $0.98$ for $\eps_r=4$, and $0.98$ for the PEC. Over the
same range, the relative curvature radius $M_2/\sqrt{M_4M_0}$ stays
between $0.58$ and $0.96$ for all contrasts, and between $0.75$ and
$0.78$ for the PEC: the curve becomes longer, not tighter.

\section{Guaranteed Bounds by Interval Arithmetic}
\label{supp:certified}

This appendix supports Sec.~\ref{subsec:observed}. The sensitivity of a
trained decoder is estimated there by sampling, i.e., at a finite
number of test points, and its accuracy by an RMS error; Propositions~1
and~2, instead, refer to the largest sensitivity and to the worst-case
error over the whole continuous range of the control. Sampling can miss
the worst point. For $L=1$, this limitation can be removed.

The tool is interval arithmetic
\cite{moore1966interval,gowal2018effectiveness}, the analogue of a
worst-case tolerance analysis. Instead of evaluating the network at a
single input, it is evaluated on a whole interval of inputs, propagating
a lower and an upper bound through each layer; the result is an
interval that is guaranteed to contain the output for
every input in
the interval, including points that sampling would miss. Here the
steering-angle range is divided into $10^5$ sub-intervals, and the
bounds are propagated through the whole chain, from the field model
through the encoder and the decoder. This gives guaranteed bounds on
the reconstruction error over the whole continuous range and, since
for $L=1$ $\Gamma_{\mathrm{dec}}$ is the largest speed of the decoded
curve (Appendix~\ref{app:minimum_constant}), on the true sensitivity.
The bounds are computed in double precision without directed rounding;
they are therefore rigorous up to floating-point rounding errors, which
are many orders of magnitude below the differences reported below.

The method was applied to a second set of trained
networks, obtained by repeating the training of Sec.~\ref{subsec:observed} without the
limits on the network weights (Appendix~\ref{app:training}); the results
are in Table~\ref{tab:certified}. For the networks that reconstruct
every field with an error of at most $5\times10^{-3}$, the sensitivity
and the error estimated from samples differ from the guaranteed values
by less than $0.2\%$ and $4\%$, respectively: for accurate networks,
sampling is reliable, as assumed in Sec.~\ref{subsec:observed}. For less accurate
networks, the sampled sensitivity can be up to $39\%$ below its
guaranteed upper bound; for this reason, Sec.~\ref{subsec:observed} uses sampled
values only for accurate networks. Without the weight limits, training
gave no accurate network for $N\ge32$: the networks obtained were more
sensitive than necessary. This is a limitation of the training, not of
the bound.

\begin{table}[t]
	\centering
	\caption{Guaranteed bounds for the networks trained without weight
		limits. ``Bound'': numerical evaluation of \eqref{eq:sharp} at
		$\eps=5\times10^{-3}$, relative to $\mathrm{Vol}_1$; ``accurate'':
		number of networks with guaranteed worst-case error at most
		$5\times10^{-3}$, over the number of trained networks;
		$\Gamma^{\mathrm{hi}}_{\mathrm{dec}}$: guaranteed upper bound on the
		sensitivity of the best accurate network, relative to
		$\mathrm{Vol}_1$; ``sampling difference'': largest relative difference between
		$\Gamma^{\mathrm{hi}}_{\mathrm{dec}}$ and the sampled sensitivity \eqref{eq:lip-hat}.}
	\label{tab:certified}
	\small
	\setlength{\tabcolsep}{9pt}
	\begin{tabular}{@{}rcccc@{}}
		\toprule
		$N$ & bound & accurate & $\Gamma^{\mathrm{hi}}_{\mathrm{dec}}$ & sampling difference \\
		\midrule
		4 & $0.985$ & 4/4 & $1.011$ & $0.04\%$ \\
		8 & $0.994$ & 4/4 & $1.060$ & $0.09\%$ \\
		16 & $0.997$ & 2/4 & $1.597$ & $0.13\%$ \\
		32 & $0.998$ & 0/8 & -- & $27\%^\ast$ \\
		64 & $0.999$ & 0/1 & -- & $39\%^\ast$ \\
		\bottomrule
	\end{tabular}
	\\[2pt]
	\footnotesize $^\ast$No network was accurate; difference over all networks.
\end{table}

\section{Network, Data, and Training Details}
\label{app:training}

Table~\ref{tab:training} summarizes the data, architectures, and
training parameters of all experiments; its last column refers to the
networks trained without weight limits, used for the guaranteed bounds of Appendix~\ref{supp:certified}.

\begin{table}[!tp]
	\centering
	\caption{Network, data, and training details of the autoencoder experiments.}
	\label{tab:training}
	\small
	\setlength{\tabcolsep}{4pt}
	\renewcommand{\arraystretch}{1.12}
	\begin{tabularx}{\textwidth}{@{}>{\raggedright\arraybackslash}p{2.3cm}LLLL@{}}
		\toprule
		& Sec.~\ref{sec:array_case_study} (arrays) & Sec.~\ref{sec:one_dimensional_manifolds} (cylinder) & Sec.~\ref{sec:steering_case_study} (steering arc) & Appendix~\ref{supp:certified} (reproduction) \\
		\midrule
		Field data
		& five $\lambda/2$ dipoles, spacing $\lambda/4$, induced-EMF impedances; $1$-V sources; loads $Z_0(1+e^{\mathrm j\phi})/(1-e^{\mathrm j\phi})$, $Z_0=50\,\Omega$; fed: unit-modulus excitations; AF: amplitude in $[0.3,1.7]$
		& Richmond volume integral equation \cite{richmond1965scattering}, TM$_z$, square cells of side $\lambda/15$ ($\lambda/20$ for the extended sweep of Sec.~\ref{subsec:scattering_setup}); elliptic cylinder, $b=0.3a$, $\eps_r=1.5$, $a=0.25\lambda$--$2\lambda$ ($\beta a=1.57$--$12.57$)
		& array factor, element spacing $\lambda/2$, $|\psi|\le0.12$, $N=4$--$64$
		& array-mode coordinates, $|\psi|\le0.12$, $N=4$--$64$ \\
		Observation
		& $180$ azimuthal samples, H-plane
		& $128$ probes, circle of radius $8\lambda$
		& $256$ samples uniform in $u=\sin\vartheta\in[-1,1)$, weight $256^{-1/2}$
		& $2N$ array-mode coefficients \\
		Samples
		& $1500$ control states uniform at random; train/validation/test $75/10/15\%$
		& $1200$ incidence angles, uniform grid; train/validation/test $75/10/15\%$
		& $2048$ random training angles per run; $512$ uniform test angles
		& $400$ uniform in $\psi$; guaranteed bounds on $10^5$ sub-intervals \\
		Normalization
		& fixed scale $M^{-1/2}$; training mean removed
		& unit incident field; training mean removed
		& unit per-element excitation; mean removed; standardized coefficients
		& unit per-element excitation \\
		SVD output layer
		& $15$ vectors of the SVD of all $1500$ snapshots (family rank $\le5$)
		& radiation operator, $\sigma_n/\sigma_1\ge10^{-3}$ ($-60$~dB)
		& all $N$ array modes ($\sigma_n/\sigma_1\ge10^{-3}$), $2N$ real outputs
		& all $2N$ modes (exact) \\
		Encoder widths
		& $360\to192\to192\to L$
		& $256\to256\to256\to L$
		& $2N\to256\to256\to1$
		& $2N\to192\to192\to1$ \\
		Decoder widths
		& $L\to192\to192\to96\to15$
		& $L\to256\to256\to128\to N_{\mathrm{SVD}}$
		& $1\to256\to128\to2N$
		& $1\to192\to192\to96\to2N$ \\
		Latent dimension
		& $L\in\{1,2,3,4\}$
		& $L\in\{1,2,3,4\}$
		& $L=1$; linear, or $\mathrm{Vol}_1\tanh(\cdot)$ in the sensitivity-limited runs
		& $L=1$, unbounded \\
		Loss
		& mean square error between $\mathbf U_r\boldsymbol\alpha$ and the centred input
		& as for the arrays
		& mean square error on the standardized coefficients
		& mean of $\|\mathbf x-\widehat{\mathbf x}\|^2$ \\
		Optimizer
		& Adam \cite{kingma2015adam}, $2\times10^{-3}$, $\times0.7$ every $400$ epochs; gradient-norm clipping at $5$
		& as for the arrays
		& Adam, $2\times10^{-3}$, halved every $5000$ iterations
		& Adam, $2\times10^{-3}$, cosine decay to $10^{-6}$ ($5\times10^{-4}$ at $N=32$) \\
		Batch, iterations
		& mini-batch $128$; $\le4000$ epochs, early stopping on validation (patience $800$)
		& as for the arrays, $\le5000$ epochs
		& mini-batch $256$; $25000$ iterations
		& full batch; $4000$--$12000$ ($20000$ at $N=32$) \\
		Runs
		& $3$ per case and $L$ (seeds $8$, $9$, $10$)
		& $3$ per size and $L$ (seeds $8$, $9$, $10$)
		& $29$ per aperture ($28$ at $N=64$): $5$ unconstrained, $24$ sensitivity-limited ($8$ values of $f$ $\times3$ seeds)
		& $4$ seeds per aperture ($8$ at $N=32$, $1$ at $N=64$) \\
		Error metric
		& $\bar e$ over $225$ test samples, in dB, averaged over runs
		& $\bar e$ over $180$ test angles, in dB, averaged over runs
		& $e_{\mathrm{rms}}$ over the test angles; gate $5\times10^{-3}$ ($-46$~dB)
		& guaranteed $\sup\|\mathbf x-\widehat{\mathbf x}\|\le5\times10^{-3}$ \\
		Implementation
		& MATLAB R2026a, Deep Learning Toolbox
		& as for the arrays
		& as for the arrays
		& as for the arrays; checked in PyTorch, NumPy, SciPy, double precision \\
		\bottomrule
	\end{tabularx}
\end{table}

\smallskip

\emph{SVD output layer for the arrays.} The fields of the arrays are
affine in at most two complex currents or excitations, so that their
real rank is at most $5$: one direction for the fixed term $F^0$ and two
for each complex current or excitation. The singular vectors up to this
rank span these fields; the others have singular values at round-off
level and do not correspond to radiated fields. Since the fields have
no component along them, the decoder output along these vectors is part
of the reconstruction error, and is therefore below the $-60$~dB
target. The span of the first singular vectors is fixed by the physics
of the array, not by the snapshots, so that computing the SVD on all
snapshots gives no information on the test set.

\smallskip

\emph{Sensitivity-limited runs (Sec.~\ref{subsec:observed}).} The sensitivity of a
decoder, i.e., its largest amplification from the latent coordinate to
the field, is at most the product of the largest singular values of its
three weight matrices (Table~\ref{tab:training}), up to a factor set by
the activation functions. To limit it, every $20$ training iterations
each of these singular values is reduced, if larger, to
$(f\,\Pi_{\mathrm{unc}})^{1/3}$, where $\Pi_{\mathrm{unc}}$ is the same
product for an unconstrained decoder at the same aperture; the product
is thus at most $f$ times that of the unconstrained decoder, with $f$
from $10^{-3}$ to $1$. The encoder output passes through a $\tanh$,
which keeps the latent coordinate in a fixed interval, so that the
encoder cannot offset the limit on the decoder by spreading the codes.

\smallskip

\emph{TwoNN (Sec.~\ref{subsec:diffusion_maps}).} The estimator \cite{facco2017} infers $d$
from the ratio of the distances of each sample to its second and first
nearest neighbors. It assumes randomly placed samples: on a uniform grid
the two neighbors are nearly equidistant and the estimate diverges. It is
therefore applied to a separate set of $1200$ angles drawn uniformly at
random, not to the uniform grid used for training; over five independent
random subsamplings per dataset size, the estimate is $0.98$--$1.01$ at
all eight sizes.

\smallskip

\begin{figure}[t]
	\centering
	\begin{tikzpicture}[
		>=stealth,
		every node/.style={font=\scriptsize}
		]
		
		\def\rad{1.40}
		
		% Two panels: small bandwidth and median bandwidth
		\foreach \shift/\bw in {0/0.6,4.1/1.41421356}{
			
			\begin{scope}[xshift=\shift cm]
				
				% Sampled points of the field manifold
				\foreach \i in {0,...,15}{
					\coordinate (p\i) at ({22.5*\i}:\rad);
				}
				
				% Pairwise connections with Gaussian weights
				\foreach \i in {0,...,15}{
					\foreach \j in {0,...,15}{
						
						\ifnum\j>\i
						
						% Squared distance on the unit circle
						\pgfmathsetmacro{\distSquared}{
							4*(sin(11.25*(\i-\j)))^2
						}
						
						% Gaussian similarity weight
						\pgfmathsetmacro{\wij}{
							exp(-\distSquared/(\bw*\bw))
						}
						
						% Visual thickness and opacity
						\pgfmathsetmacro{\lw}{0.2+1.4*\wij}
						\pgfmathsetmacro{\op}{0.9*\wij}
						
						\draw[
						blue!70!black,
						opacity=\op,
						line width=\lw pt
						] (p\i)--(p\j);
						
						\fi
					}
				}
				
				% Nodes
				\foreach \i in {0,...,15}{
					\fill[blue!70!black]
					(p\i) circle (1.8pt);
				}
				
			\end{scope}
		}
		
		% Panel titles
		\node[align=center] at (0,1.95)
		{\textbf{(a) Small bandwidth}};
		
		\node[align=center] at (4.1,1.95)
		{\textbf{(b) Median bandwidth}};
		
		% Panel descriptions
		\node[align=center] at (0,-1.85)
		{Predominantly local\\similarities};
		
		\node[align=center] at (4.1,-1.85)
		{Local and nonlocal\\similarities};
		
		% Legend
		\draw[blue!70!black,line width=1.5pt]
		(-0.7,-2.65)--(0.0,-2.65);
		\node[anchor=west] at (0.1,-2.65)
		{High similarity};
		
		\draw[blue!70!black,opacity=0.2,line width=0.5pt]
		(3.25,-2.65)--(3.95,-2.65);
		\node[anchor=west] at (4.05,-2.65)
		{Low similarity};
		
		\node[align=center] at (2.05,-3.15)
		{Same sampled field manifold};
		
	\end{tikzpicture}
	
	\caption{Effect of the bandwidth on the similarity
		network associated with a family of electromagnetic
		fields forming a closed orbit.
		Each node represents a field pattern, i.e., a point
		of the field manifold in the ambient Hilbert space
		$\mathcal H_E$, rather than an observation point in
		physical space.
		The line intensity represents the Gaussian similarity
		weight before row normalization.
		(a)~A small bandwidth emphasizes connections between
		neighboring fields, so that the random walk proceeds
		predominantly along the curve.
		(b)~A bandwidth equal to the median pairwise distance,
		as used in Sec.~\ref{subsec:diffusion_maps}, also retains significant
		connections between more distant fields.
		Both networks represent exactly the same sampled
		field family, but emphasize similarities at
		different scales.}
	\label{fig:diffusion_explanation}
	
\end{figure}

\smallskip

\emph{Diffusion maps (Sec.~\ref{subsec:diffusion_maps}).}
The SVD and PCA identify low-dimensional \emph{linear} subspaces that
approximate a collection of fields. A parametric field family, however,
may have a small intrinsic dimension and still require many modes for an
accurate linear representation. Diffusion maps \cite{coifman2006diffusion}
offer a complementary, nonlinear approach: they build low-dimensional
coordinates from the similarities between fields alone.

Each radiation pattern or scattered field is a point in the field space
$\mathcal H_E$ of Sec.~\ref{sec:entropy_capacity}. Consider, for example,
two antennas, one with a fixed excitation and one with a variable phase
$\phi$, so that $F(\phi)=\mathbf f_1+e^{\mathrm j\phi}\mathbf f_2$. Each
field may need hundreds of complex samples, yet the family is a closed
curve in $\mathcal H_E$, with
$\|F(\phi_i)-F(\phi_j)\|=2\|\mathbf f_2\|\,|\sin\frac{\phi_i-\phi_j}{2}|$:
fields with nearby phases, modulo $2\pi$, are close in field space. The
question is whether this closed structure can be recovered from the
fields alone, without the phase labels.

Diffusion maps associate a weighted network with the sampled fields: each
node is a field, and the weight of a connection, $W_{ij}=\exp(-\Delta_{ij}^2/h^2)$
with $\Delta_{ij}$ the field distance, measures the similarity of two
fields on the scale set by the bandwidth $h$. A particle performing a
random walk on this network, with transition probabilities obtained by
normalizing the weights (with the density normalization $\alpha=1$ of
\cite{coifman2006diffusion}, which reduces the effect of nonuniform
sampling), moves mainly between neighboring points of the curve when the
connections are local, Fig.~\ref{fig:diffusion_explanation}(a), and can
jump between distant parts of the curve when longer connections are also
significant, Fig.~\ref{fig:diffusion_explanation}(b). This diffusion is a
mathematical process on the set of fields, not a physical propagation in
the observation region.

Repeated steps of the walk average a quantity assigned to the nodes over
the network: rapid variations between strongly connected fields are
smoothed out quickly, whereas slow, large-scale variations persist. The
leading nonconstant eigenvectors $\boldsymbol\psi_2,\boldsymbol\psi_3,\dots$
of the transition matrix describe these persistent variations, and their
components assign to each field coordinates that reflect the organization
of the family. 

In Sec.~\ref{subsec:diffusion_maps}, each field $\mathbf x_i$ is
represented by the coordinates
$\zeta_1(\mathbf x_i)=\psi_2(i)$, $\zeta_2(\mathbf x_i)=\psi_3(i)$,
$\zeta_3(\mathbf x_i)=\psi_4(i)$, and $h$ is the median pairwise distance. The
network is then not restricted to nearest neighbors,
Fig.~\ref{fig:diffusion_explanation}(b), which suits the large-scale
organization of the family rather than its local differential geometry.
The map is neither isometric nor necessarily injective, so apparent
crossings or proximities in the $\zeta$ coordinates need not exist in
field space. Diffusion maps are therefore used only as a qualitative
visualization tool; quantitative conclusions on distinguishability and
on the separation of branches rely on distances in $\mathcal H_E$.

\newpage

\section{List of Symbols}
\label{supp:symbols}

The symbols are listed below in alphabetical order,
Latin before Greek letters, followed by the general notation and the
acronyms. The last column gives the section, and the
equation or footnote when applicable, where each symbol is defined or
first used. Symbols with more than one meaning are listed with all of
them; the meaning is always clear from the section in which they
appear.

{\small
	\setlength{\tabcolsep}{5pt}
	\renewcommand{\arraystretch}{1.05}
	\begin{longtable}{@{}>{\raggedright\arraybackslash}p{\dimexpr0.20\linewidth-2\tabcolsep\relax}>{\raggedright\arraybackslash}p{\dimexpr0.58\linewidth-2\tabcolsep\relax}>{\raggedright\arraybackslash}p{0.22\linewidth}@{}}
        \toprule Symbol & Meaning & Section \\ \midrule
        \endfirsthead
        \toprule Symbol & Meaning (continued) & Section \\ \midrule
        \endhead
        \bottomrule
        \endfoot
        \bottomrule
        \endlastfoot
		\multicolumn{3}{@{}l}{\emph{Latin letters}} \\[1pt]
		$A$ & field amplitude $\|\Phi_s(\theta)\|$ of a single configuration & \ref{subsec:normalization}, \eqref{eq:perconfig-metric} \\
		$a$ & characteristic size of the source; radius of the circular cylinder; major semi-axis of the elliptic cylinder & \ref{subsec:metric_entropy_definition}, \ref{subsec:entropy_results}, \ref{subsec:scattering_setup} \\
		$B_{\mathcal H_E}(x,\rho)$ & open ball of radius $\rho$ centered at $x$ in $\mathcal H_E$ & \ref{subsec:metric_entropy_definition}, \eqref{eq:covering-def} \\
		$B_d(\eps)$ & Euclidean ball of radius $\eps$ in $\mathbb R^d$ & \ref{subsec:metric_entropy_definition}, \eqref{eq:covering-asymp} \\
		$b$ & minor semi-axis of the elliptic cylinder, $b=0.3a$ & \ref{subsec:scattering_setup} \\
		$C_\eps(\mathcal E)$ & $\eps$-capacity, $\log_2M_{2\eps}(\mathcal E)$, in bits & \ref{subsec:metric_entropy_definition} \\
		$C'$ & proportionality factor of the linear entropy & \ref{subsec:metric_entropy_definition}, \eqref{eq:intro-entropy-ndf} \\
		$c_n$ & modal coefficients of the circular cylinder & \ref{subsec:entropy_results} \\
		$C_0$ & constant of the entropy bound \eqref{eq:uniform-entropy-main}; $C_0=\log_2(KV_d4^d)$ under \eqref{eq:local-regularity} & \ref{subsec:stable_widths_lipschitz} \\
		$c_1$, $c_1(L)$ & constants of the bounds \eqref{eq:migration} and \eqref{eq:stability-tradeoff} & \ref{subsec:stable_widths_lipschitz} \\
		$c_L$ & covering constant of unit-diameter sets in $\mathbb R^L$ & \ref{subsec:stable_widths_lipschitz} \\
		$c_\sigma$ & resolution-to-noise ratio, $\eps=c_\sigma\sigma$ & \ref{subsec:normalization} \\
		$D\Phi(\theta)$ & differential (Jacobian) of $\Phi$ at $\theta$ & \ref{subsec:metric_entropy_definition} \\
		$d$ & intrinsic dimension (EID) of $\mathcal E$ & \ref{subsec:metric_entropy_definition}, \eqref{eq:eid-definition} \\
		$d_\eps$ & slope of $H_\eps$ at resolution $\eps$ & \ref{subsec:p-vs-d} \\
		$d_{\mathrm{eff}}(\theta)$ & participation ratio of the eigenvalues of $\mathbf g(\theta)$ & \ref{sec:fisher_bridge}, \eqref{eq:deff} \\
		$E$, $\mathbf E$ & field in $\mathcal H_E$; sampled complex field & \ref{subsec:metric_entropy_definition}, \ref{sec:stability_validation} \\
		$\mathbf E(u;\psi,\chi)$ & two-component far field of the dual-polarized array & \ref{subsec:test-manifold}, \eqref{eq:dp-field-main} \\
		$\mathcal E$ & field manifold $\Phi(\Theta)$ & \ref{subsec:metric_entropy_definition}, \eqref{eq:manifold-def} \\
		$\mathcal E_s$ & field manifold at electrical size $s$ & \ref{subsec:metric_entropy_definition} \\
		$\hat{\mathcal E}_s$ & normalized family $\mathcal E_s/\|F_s\|$ & \ref{subsec:normalization} \\
		$\mathcal E_{\mathrm{lin}}$ & linear field family of unconstrained sources & \ref{subsec:metric_entropy_definition}, \eqref{eq:intro-entropy-ndf} \\
		$\bar e$ & mean relative test error (Appendix~\ref{supp:notation}) & Fig.~\ref{fig:figyya2} \\
		$e_{\mathrm{rms}}$ & relative RMS test error (Appendix~\ref{supp:notation}) & \ref{subsec:observed} \\
		$F$ & field; scattered far field $F(\vartheta;\phi,\beta a)$ of the cylinder & \ref{sec:one_dimensional_manifolds} \\
		$\|F\|$, $\|F_s\|$ & RMS field norm over $\mathcal E$ (over $\Theta$ at size $s$) & \ref{subsec:normalization}, \eqref{eq:rms-norm} \\
		$F^0$ & field radiated with the loaded ports open & \ref{subsec:coupling}, \eqref{eq:load-structure} \\
		$f$, $f_i$; $\widehat f_i$ & field of $\mathcal E$, test sample; reconstruction & \ref{subsec:stable_widths_lipschitz}, \ref{sec:stability_validation} \\
		$\mathbf f_k$ & pattern of a unit current at loaded port $k$ & \ref{subsec:coupling}, \eqref{eq:load-structure} \\
		$\mathbf g(\theta)=[g_{ij}]$ & pullback metric (Gram matrix of $\partial_{\theta_i}\Phi$) & \ref{subsec:metric_entropy_definition} \\
		$\hat{\mathbf g}$ & pullback metric of the normalized family & \ref{sec:fisher_bridge} \\
		$\check g_{ij}$ & pullback metric under per-configuration normalization & \ref{subsec:normalization}, \eqref{eq:perconfig-metric} \\
		$g(u)$ & element pattern & \ref{subsec:test-manifold}, \eqref{eq:dp-field-main} \\
		$H_\eps(\mathcal E)$ & Kolmogorov $\eps$-entropy, $\log_2N_\eps(\mathcal E)$, in bits & \ref{subsec:metric_entropy_definition}, \eqref{eq:intro-entropy} \\
		$\mathcal H_E$ & field-observation space ($L^2$ norm); $\mathbb C^{N_{\mathrm{obs}}}$ if sampled & \ref{subsec:metric_entropy_definition}, \ref{sec:fisher_bridge} \\
		$\mathbf I_{\mathrm p}$, $I_{\mathrm p,k}$ & currents at the loaded ports & \ref{subsec:coupling}, \eqref{eq:load-structure} \\
		$\mathbf J(\theta)$ & Fisher information matrix & \ref{sec:fisher_bridge}, \eqref{eq:fim-equals-metric} \\
		$K$ & local-volume constant of the regularity condition & \ref{subsec:stable_widths_lipschitz}, \eqref{eq:local-regularity} \\
		$L$ & latent dimension (number of latent coordinates) & \ref{subsec:stable_widths_lipschitz} \\
		$L_{\mathrm{AE}}$ & autoencoder latent dimension & \ref{sec:stability_validation} \\
		$L_{\mathrm{ap}}$ & array aperture length, $L_{\mathrm{ap}}=N\lambda/2$ & \ref{subsec:test-manifold} \\
		$L_{\min}$ & minimum Euclidean embedding dimension of $\mathcal E$ & \ref{subsec:p-vs-d} \\
		$L_{\mathrm{stab}}(\eps,\bar\Gamma)$ & operational stable latent dimension & \ref{subsec:opdimension}, \eqref{eq:stable-latent-dimension} \\
		$\widehat L_{\mathrm{dec}}$ & sampled estimate of $\operatorname{Lip}(\widetilde\Phi_{\mathrm{dec}})$ & \ref{sec:stability_validation}, \eqref{eq:lip-hat} \\
		$\ell_{\mathcal E}$ & length of a one-dimensional $\mathcal E$ & \ref{subsec:entropy_results} \\
		$\ell_\psi$, $\ell_\chi$ & lengths of the steering and polarization families & \ref{subsec:test-manifold}, \eqref{eq:prop-arclength}; \ref{sec:communication_theory} \\
		$M_\eps(\mathcal E)$ & packing number & \ref{subsec:metric_entropy_definition} \\
		$M_k$ & modal moments $2\pi\sum_nn^k|c_n|^2$ of the circular cylinder & \ref{subsec:entropy_results} \\
		$m$ & multiplicity of an $m$-to-one map $\Phi$ & \ref{subsec:p-vs-d} \\
		$N$ & number of array elements & \ref{subsec:test-manifold} \\
		$N_\eps(\mathcal E)$ & covering number & \ref{subsec:metric_entropy_definition}, \eqref{eq:covering-def} \\
		$N_{\mathrm{obs}}$ & number of complex field observations & \ref{sec:stability_validation} \\
		$N_{\mathrm{SVD}}$ & number of retained real singular vectors & \ref{sec:stability_validation}, \eqref{eq:svd_rec} \\
		$\mathrm{NDF}_\eps$ & number of singular values above $\eps$ & \ref{subsec:metric_entropy_definition}, \eqref{eq:ndf-def} \\
		$\mathbf n$ & observation noise & \ref{sec:fisher_bridge} \\
		$n_\star$ & modal band edge, $n_\star\approx\beta a$ & App.~\ref{app:regularity} \\
		$P$ & perimeter of the cross-section; in $\nu|_P$, fixed radiated or scattered power & Tabs.~\ref{tab:metrics}, \ref{tab:normalization} \\
		$p$ & number of independent controls & \ref{subsec:metric_entropy_definition} \\
		$q$ & growth exponent of $\mathrm{NDF}_\eps$ ($1$ or $2$) & \ref{subsec:metric_entropy_definition} \\
		$R$ & latent rescaling factor; minimum curvature radius of the arc & \ref{subsec:stable_widths_lipschitz}, \ref{subsec:sharp-arc} \\
		$\mathrm{SNR}$ & signal-to-noise ratio of the observed fields, $\|F\|^2/\sigma^2$ ($\|F_s\|^2/\sigma^2$ at size $s$) & \ref{subsec:normalization} \\
		$S^1$, $T^2$ & circle and torus & \ref{subsec:coupling} \\
		$s$ & electrical size, $s=\beta a$ ($\beta L_{\mathrm{ap}}$ for arrays) & \ref{subsec:metric_entropy_definition} \\
		$t_{0.975,2}$ & Student-$t$ quantile ($95\%$ confidence) & Tab.~\ref{tab:prop1} \\
		$\mathbf U_r$, $\mathbf u_i$ & retained left singular vectors & \ref{sec:stability_validation}, \eqref{eq:svd_rec} \\
		$u$ & direction cosine, $u=\sin\vartheta$ & \ref{subsec:test-manifold}, \eqref{eq:dp-field-main} \\
		$\mathbf v$ & open-circuit voltage at the loaded ports & \ref{subsec:coupling}, \eqref{eq:load-structure} \\
		$V_d$ & volume of the unit ball in $\mathbb R^d$ & \ref{subsec:metric_entropy_definition}, \eqref{eq:covering-asymp} \\
		$\mathrm{Vol}_d(\mathcal E)$ & intrinsic (Riemannian) volume of $\mathcal E$; length for $d=1$, area for $d=2$ & \ref{subsec:metric_entropy_definition} \\
		$w$ & antipodal width: smallest distance between fields with controls half a period apart & \ref{subsec:p-vs-d}, \eqref{eq:antipodal} \\
		$X_k$ & reactance of the load at port $k$ & \ref{subsec:coupling}, \eqref{eq:load-structure} \\
		$\mathbf x$, $\mathbf x_{\mathrm{pred}}$ & real-stacked field vector; its reconstruction & \ref{sec:stability_validation}, \eqref{eq:svd_rec} \\
		$x_n$ & position of the $n$th array element & \ref{subsec:test-manifold} \\
		$\mathbf y$ & noisy observation & \ref{sec:fisher_bridge} \\
		$\mathbf z$, $z_i$ & latent vectors (codes) & \ref{subsec:stable_widths_lipschitz}, \ref{sec:stability_validation} \\
		$\mathbf Z_{\mathrm r}$ & impedance matrix of the loaded ports & \ref{subsec:coupling}, \eqref{eq:load-structure} \\
		$\mathbf Z_{\mathrm L}$ & load impedances, $\operatorname{diag}(\mathrm jX_k)$ & \ref{subsec:coupling}, \eqref{eq:load-structure} \\[2pt]
		\multicolumn{3}{@{}l}{\emph{Greek letters}} \\[1pt]
		$\alpha_i$ & coefficients of the SVD layer & \ref{sec:stability_validation}, \eqref{eq:svd_rec} \\
		$\alpha$ & density normalization of the diffusion maps & Fig.~\ref{fig:dm} \\
		$\beta$ & wavenumber; $\beta a$, $\beta L_{\mathrm{ap}}$: electrical size & \ref{subsec:metric_entropy_definition} \\
		$\Gamma(\cdot)$ & Euler Gamma function & \ref{subsec:metric_entropy_definition} \\
		$\Gamma_{\mathrm{dec}}$ & normalized decoder sensitivity $\operatorname{Lip}(\Phi_{\mathrm{dec}})\operatorname{diam}(\Psi_{\mathrm{enc}}(\mathcal E))$ & \ref{subsec:stable_widths_lipschitz}, \eqref{eq:normalized-sensitivity} \\
		$\bar\Gamma$ & prescribed upper bound on $\Gamma_{\mathrm{dec}}$ & \ref{subsec:stable_widths_lipschitz}, \eqref{eq:decoder-stable-width} \\
		$\Gamma^\star_{\mathrm{dec}}(\eps)$ & smallest $\Gamma_{\mathrm{dec}}$ compatible with accuracy $\eps$ & \ref{subsec:sharp-arc} \\
		$\widehat\Gamma_{\mathrm{dec}}$ & sampled proxy of $\Gamma_{\mathrm{dec}}$ & \ref{sec:stability_validation} \\
		$\Gamma^{\mathrm{hi}}_{\mathrm{dec}}$ & guaranteed upper bound on $\Gamma_{\mathrm{dec}}$ & App.~\ref{supp:certified} \\
		$\gamma$ & Lipschitz constant of the decoder in \eqref{eq:decoder-lipschitz} & \ref{subsec:stable_widths_lipschitz}, \eqref{eq:decoder-lipschitz} \\
		$\delta^{\mathrm{dec}}_{L,\bar\Gamma}(\mathcal E)$ & decoder-stable width & \ref{subsec:stable_widths_lipschitz}, \eqref{eq:decoder-stable-width} \\
		$\delta(\Delta)$ & injectivity modulus: smallest distance between fields with controls at least $\Delta$ apart & \ref{subsec:p-vs-d} \\
		$\Delta$ & control separation; angular separation of two beams & \ref{subsec:p-vs-d}, \ref{sec:communication_theory} \\
		$\Delta\theta$ & increment of the control parameters & \ref{sec:fisher_bridge} \\
		$\Delta_{\mathrm{br}}$ & branch separation of the elliptic orbit & \ref{subsec:diffusion_maps} \\
		$\eps$ & resolution or worst-case error (absolute distance in $\mathcal H_E$) & \ref{subsec:metric_entropy_definition} \\
		$\hat\eps$ & relative resolution, $\eps/\|F\|$ & \ref{subsec:normalization} \\
		$\eps_{\max}$ & validity limit of \eqref{eq:uniform-entropy-main}, $\eps_{\max}=\rho_0/4$ & \ref{subsec:stable_widths_lipschitz} \\
		$\eps_r$ & relative permittivity of the cylinder & \ref{subsec:entropy_results} \\
		$\zeta_1,\zeta_2,\zeta_3$ & diffusion-map coordinates & \ref{subsec:diffusion_maps}, Fig.~\ref{fig:dm} \\
		$\boldsymbol\eta$, $\boldsymbol\xi$ & trainable parameters of the encoder and the decoder & \ref{sec:stability_validation} \\
		$\Theta$ & control domain, $\Theta\subseteq\mathbb R^p$; $|\Theta|$ its measure & \ref{subsec:metric_entropy_definition} \\
		$\theta$ & vector of control parameters & \ref{subsec:metric_entropy_definition} \\
		$\vartheta$ & observation angle & \ref{subsec:test-manifold}, \ref{subsec:entropy_results} \\
		$\kappa_s$ & common rescaling factor of the field family & \ref{subsec:normalization} \\
		$\lambda$ & wavelength & \ref{subsec:coupling} \\
		$\lambda_k(\theta)$ & eigenvalues of $\mathbf g(\theta)$ & \ref{sec:fisher_bridge} \\
		$\nu$ & volumetric exponent, $\mathrm{Vol}_d(\mathcal E_s)\propto s^\nu$ & \ref{subsec:metric_entropy_definition}, \eqref{eq:intro-entropy-nu} \\
		$\hat\nu$, $\nu_{\mathrm g}$ & shape and gain exponents & \ref{subsec:normalization}, \eqref{eq:nu-split} \\
		$\nu_k$, $\hat\nu_k$ & directional exponents ($\sqrt{\lambda_k}\propto s^{\nu_k}$) & \ref{sec:fisher_bridge} \\
		$\nu_\psi$, $\nu_\chi$ & exponents of the steering and polarization lengths & \ref{subsec:test-manifold}, \ref{sec:communication_theory} \\
		$\nu_{\mathrm{PEC}}$, $\nu_{\mathrm{Born}}$ & exponents of the geometric (PEC) and Born regimes & App.~\ref{app:derivation} \\
		$\nu|_P$ & $\nu$ at fixed radiated or scattered power & Tab.~\ref{tab:normalization} \\
		$\rho$ & radius of a ball in $\mathcal H_E$ & \ref{subsec:stable_widths_lipschitz}, \eqref{eq:local-regularity} \\
		$\rho_0$ & largest radius for which \eqref{eq:local-regularity} holds & \ref{subsec:stable_widths_lipschitz}, \eqref{eq:local-regularity} \\
		$\rho_1$ & radius below which $\mathcal E\cap B_{\mathcal H_E}(x,\rho)$ is connected & \ref{subsec:sharp-arc} \\
		$\sigma$ & noise standard deviation & \ref{subsec:normalization} \\
		$\sigma_n$ & singular values of the radiation operator & \ref{subsec:metric_entropy_definition}, \eqref{eq:ndf-def} \\
		$\tau$ & phase lag across the cylinder, $(\eps_r-1)\beta a$ for weak contrast & \ref{subsec:entropy_results} \\
		$\tau_{\mathcal E}$ & reach of $\mathcal E$ & \ref{subsec:metric_entropy_definition} \\
		$\Phi$, $\Phi_s$ & parameter-to-field map (at size $s$) & \ref{subsec:metric_entropy_definition} \\
		$\Phi_{\mathrm{dec}}$ & decoder & \ref{subsec:stable_widths_lipschitz} \\
		$\Phi_{\mathrm{dec},\boldsymbol\xi}$ & network decoder (latent code to SVD coefficients) & \ref{sec:stability_validation} \\
		$\widetilde\Phi_{\mathrm{dec}}$ & composite field decoder $\mathbf U_r\Phi_{\mathrm{dec},\boldsymbol\xi}$ & \ref{sec:stability_validation} \\
		$\check\Phi_s$ & map normalized per configuration, $\Phi_s/A$ & \ref{subsec:normalization} \\
		$\phi$, $\boldsymbol\phi$, $\phi_k$ & periodic control in \eqref{eq:antipodal}; load phases (reflection coefficients $e^{\mathrm j\phi_k}$); illumination angle & \ref{subsec:p-vs-d}, \ref{subsec:coupling}, \ref{sec:one_dimensional_manifolds} \\
		$\chi$ & linear polarization angle & \ref{subsec:test-manifold}, \eqref{eq:dp-field-main} \\
		$\chi_{\mathrm e}$ & susceptibility $\eps_r-1$ & App.~\ref{app:derivation} \\
		$\Psi_{\mathrm{enc}}$, $\Psi_{\mathrm{enc},\boldsymbol\eta}$ & encoder; network encoder & \ref{subsec:stable_widths_lipschitz}, \ref{sec:stability_validation} \\
		$\psi$ & steering angle & \ref{subsec:test-manifold}, \eqref{eq:dp-field-main} \\
		[2pt]
		\multicolumn{3}{@{}l}{\emph{General notation}} \\[1pt]
		$\|\cdot\|$, $\langle\cdot,\cdot\rangle$ & norm and inner product of $\mathcal H_E$; Euclidean for latent vectors & \ref{subsec:metric_entropy_definition} \\
		$\widehat{(\cdot)}$, $\hat{(\cdot)}$ & normalized quantity ($\hat\eps$, $\hat{\mathcal E}_s$, $\hat{\mathbf g}$, $\hat\nu$); reconstruction or sampled estimate ($\widehat f_i$, $\widehat{\mathbf x}$, $\widehat\Gamma_{\mathrm{dec}}$) &  \\
		$\asymp$ & equality up to a factor depending only on $d$ & \ref{subsec:metric_entropy_definition}, \eqref{eq:covering-asymp} \\
		$\mathcal O(\cdot)$, $\Theta(\cdot)$ & at most of the order of; of the same order as (Appendix~\ref{app:derivation}) & \ref{subsec:metric_entropy_definition} \\
		$\propto$, $\coloneqq$, $\#$ & proportional to; defined as; cardinality & \ref{subsec:metric_entropy_definition} \\
		$\mathcal C(\mathcal E,\mathbb R^L)$ & continuous maps from $\mathcal E$ to $\mathbb R^L$ & \ref{subsec:stable_widths_lipschitz} \\
		$\operatorname{diam}(\cdot)$ & diameter of a set & \ref{subsec:stable_widths_lipschitz}, \eqref{eq:normalized-sensitivity} \\
		$\operatorname{Lip}(\cdot)$ & smallest Lipschitz constant of a map & \ref{subsec:stable_widths_lipschitz} \\
		$\mathrm{Lip}(\mathbb R^L,\mathcal H_E)$ & class of globally Lipschitz maps & \ref{subsec:stable_widths_lipschitz} \\
		$\operatorname{rank}_{\mathbb R}$ & rank over the real numbers & \ref{subsec:metric_entropy_definition} \\
		$\operatorname{Re}$, $\operatorname{Im}$, $(\cdot)^{\mathsf T}$ & real part, imaginary part, transpose & \ref{subsec:metric_entropy_definition} \\
		$\operatorname{RMS}_\phi$ & root mean square over $\phi$ & \ref{subsec:diffusion_maps} \\
		$\mathcal{CN}(\mathbf 0,\sigma^2\mathbf I)$ & circularly symmetric complex Gaussian distribution & \ref{sec:fisher_bridge} \\
\end{longtable}
}

\bigskip

\emph{Acronyms.} A, P, F, AF: driven, parasitic (loaded),
electronically fed, and amplitude--phase-controlled elements of the
arrays of Sec.~\ref{sec:array_case_study}; AWGN: additive white Gaussian noise; EID:
electromagnetic intrinsic dimension; EIT: electromagnetic information
theory; ESIT: electromagnetic signal and information theory; FIM:
Fisher information matrix; GELU: Gaussian error linear unit; MoM:
method of moments; MIMO: multiple-input multiple-output; NDF: number of degrees of freedom; PEC: perfect
electric conductor; RMS: root mean square; SNR: signal-to-noise ratio;
SVD: singular value decomposition; $TM_z$: transverse magnetic with
respect to the cylinder axis $z$; TwoNN: Two-Nearest-Neighbors
estimator.

\newpage

{\small
	\bibliographystyle{IEEEtran}
	\bibliography{Biblio_NLNDF_Ax}
}

\end{document}